\documentclass[11pt,reqno]{article}
\usepackage[dvipsnames]{xcolor}
\usepackage{multirow}
\usepackage{upgreek}
\usepackage{tikz}
\usepackage[english]{babel}
\usepackage{array}
\usepackage{comment}
\usepackage{graphicx}
\usepackage[dvipsnames]{xcolor}
\usepackage[centertags]{amsmath}
\usepackage{amsfonts,amssymb,amsthm}
\usepackage[]{hyperref}
\usepackage[latin1]{inputenc}
\usepackage{xcolor} 
\usepackage[noadjust]{cite}
\usepackage{mathtools,xparse}
\usepackage{mathtools}
\usepackage{bbm}
\usepackage{enumitem}

\makeatletter

\def\@tvsp{\mathchoice{{}\mkern-4.5mu}{{}\mkern-4.5mu}{{}\mkern-2.5mu}{}}
\def\ltrivert{\left|\@tvsp\left|\@tvsp\left|}
\def\rtrivert{\right|\@tvsp\right|\@tvsp\right|}
\makeatother

\DeclarePairedDelimiter{\oldnormaux}{\bracevert}{\bracevert}

\NewDocumentCommand{\oldnorm}{som}{%
  \IfBooleanTF{#1}
    {\oldnormaux*{#3}}
    {\IfNoValueTF{#2}
       {\oldnormaux*{\vphantom{dq}#3}}
       {\oldnormaux[#2]{#3}}%
    }%
}

\newcommand{\bn}[1]{\oldnorm[\big]{#1}}

\newcommand*\xbar[1]{%
  \hbox{%
    \vbox{%
      \hrule height 0.5pt % The actual bar
      \kern0.5ex%         % Distance between bar and symbol
      \hbox{%
        \kern+0,00000000000001cm%      % Shortening on the left side
        \ensuremath{#1}%
        \kern+0,00000000000001cm%      % Shortening on the right side
      }%
    }%
  }%
} 

\numberwithin{equation}{section}

\theoremstyle{plain}
\newtheorem{theorem}{Theorem}[section]

\newtheorem{lemma}[theorem]{Lemma}
\newtheorem{corollary}[theorem]{Corollary}
\newtheorem{remark}[theorem]{Remark}

\theoremstyle{definition}

\newtheorem{definition}[theorem]{Definition}

\newtheoremstyle{hypstyle}{}{}{}{}{\bfseries}{.}{ }%
{\thmname{#1}\thmnumber{ (H#2)}\thmnote{}}

\theoremstyle{hypstyle}

\def\notina[#1]#2{\begingroup\def\thefootnote{\fnsymbol{footnote}}\footnote[#1]{#2}\endgroup}

\newcommand{\iu}{\textnormal{i}}

\newcommand{\R}{{\mathbb R}}
\newcommand{\C}{{\mathbb C}}
\newcommand{\Z}{{\mathbb Z}}

\newcommand{\T}{{\mathbb T}}
\newcommand{\te}{{\tt e}}

\newcommand{\mG}{\mathcal{G}}

\newcommand{\mL}{\mathcal{L}}
\newcommand{\mM}{\mathcal{M}}

\newcommand{\mR}{\mathcal{R}}
\newcommand{\mZ}{\mathcal{Z}}
\newcommand{\mq}{\rho}
\def\cZ{\mathcal{Z}}
\def\cN{\mathcal{N}}
\def\cM{\mathcal{M}}

\def\cJ{\mathcal{J}}
\def\cP{\mathcal{P}}
\def\cC{\mathcal{C}}
\def\cR{\mathcal{R}}

\def\cIR{\mathcal{IR}}
\def\cS{\mathcal{J}}
\def\ccS{\mathcal{S}}

\def\cL{\mathcal{L}}
\def\ng{\mathcal{NG}}
\def\cG{\mathcal{G}}
\def\caW{{\tt \Omega}}

\def\tj{{\tt j}}
\def\tz{{\tt z}}

\def\vj{\vec{{\it \j}}}
\def\vsigma{\vec{\sigma}}
\def\vlambda{\vec{\lambda}}
\def\vjs{\vj,\vsigma}
\def\supp{{\rm supp}}

\newcommand{\e}{\varepsilon}

\newcommand{\m}{\mu}

\newcommand{\diff}{\textnormal{d}}
\newcommand{\bs}{\boldsymbol{\sigma}}

\renewcommand{\Im}{\mathrm{Im}\,}

\newcommand{\oN}[1]{\big\| {#1} \big\|}

\newcommand{\jbs}[1]{\left\langle{#1}\right\rangle}

\newcommand{\bcb}{\begin{color}{blue}}
\newcommand{\bcr}{\begin{color}{red}}
\newcommand{\ec}{\end{color}}

\def\blu#1{{\color{blue}#1}}

\def\nablac{\langle\nabla\rangle_c}

\def\gna{G^{\tt app}}
\def\uno{{\textnormal{id}}}

\def\cT{{\mathcal T}}
\def\cU{{\mathcal U}}
\def\cV{{\mathcal V}}
\def\cO{{\mathcal O}}
\def\sgn{{\rm sgn}}
\def\im{{\rm i}}

\def\tw{{\tt w}}

\def\sleq{\lesssim}

\def\tR{{\tt R}}

\def\matA{A}
\def\matB{B}
\def\ttw{{\cal W}}
\def\vs{\varsigma}

\title{Non relativistic limit of the nonlinear Klein-Gordon
  equation: \\ Uniform {in time} approximation of KAM
  solutions.}
\author{Dario Bambusi\footnote{Dipartimento di Matematica, Universit\`a degli Studi di Milano, Via Saldini 50, I-20133 Milano, Italy email\,\,:\,\,\texttt{dario.bambusi@unimi.it, andrea.belloni@unimi.it}}\,,  Andrea Belloni$^*$, Filippo Giuliani\footnote{Dipartimento di Matematica, Politecnico di Milano, Piazza Leonardo Da Vinci 32, 20133, Milano, Italy email\,\,:\,\,\texttt{filippo.giuliani@polimi.it}} }

\allowdisplaybreaks

\begin{document}

\maketitle
\begin{abstract}
We study the non relativistic limit of the solutions of the cubic
nonlinear Klein--Gordon (KG) equation with periodic boundary
conditions on an interval and we construct a family of time quasi
periodic solutions which, after a Gauge transformation, converge
\emph{globally uniformly in time} to quasi periodic solutions of the
cubic NLS. The proof is based on KAM theory. We emphasize that,
regardless of the spatial domain, all the previous results concern
approximations valid over compact time intervals.
\end{abstract}
  
\tableofcontents

\section{Main result}
In this paper, we study the non relativistic  limit of the
solutions of the nonlinear
Klein-Gordon (KG) equation 
\begin{equation}\label{KG}
  \frac{1}{c^2} u_{tt}-u_{xx} +c^2 u+{a_1} u^3=0 ,\qquad a_1=\pm 1,
\end{equation}
where $u=u(t, x)\colon \R\times \mathbb{T} \to \R$ and
$\mathbb{T} \vcentcolon= \mathbb{R}/2\pi\mathbb{Z}$. Precisely, we prove the
existence of a family of quasi periodic in time solutions which, up to
a Gauge transformation, converge as $c\to\infty$, uniformly {for
  $t\in\R$}, to quasi periodic solutions of the cubic NLS equation
\begin{equation}\label{NLS}
    \iu\varphi_{t}=-\frac{1}{2}\varphi_{xx}+\frac{3}{4}a_1 |\varphi|^2
    \varphi \qquad \varphi=\varphi(t, x),\quad t\in \R, \quad
    x\in \mathbb{T}\,.
\end{equation}
To state our results we introduce the complex variable $\psi$ defined by 
\begin{equation}
  \label{psi}
\psi\vcentcolon= \frac{1}{\sqrt2}\left(\left(\frac{\nablac}{c}\right)^{1/2}u+\im\left(\frac{\nablac}{c}\right)^{-1/2}\frac{
u_t}{c^2}\right),
  \end{equation}
where 
\begin{align}
  \label{giap}
\langle \nabla \rangle_c\vcentcolon= (c^2-\Delta)^{1/2}.
\end{align}
In these variables the  Klein-Gordon equation \eqref{KG} reads 
\begin{equation}
  \label{KG.psi}
\im \psi_t=c\nablac
\psi+\frac{a_1}{4}\left(\frac{\nablac}{c}\right)^{-1/2}
\left[\left(\frac{\nablac}{c}\right)^{-1/2}(\psi+\bar\psi)  \right]^3\,.
  \end{equation}

We are interested in small amplitude solutions of \eqref{KG.psi}, which arise as perturbations of linear oscillations.
By considering the Fourier expansion
\begin{equation}
  \label{fourier}
\psi(t, x)=\sum_{j \in \mathbb{Z}}z_j(t)\,\frac{1}{\sqrt{2\pi}}e^{\im jx},
\end{equation}
the KG equation \eqref{KG.psi} 
 {becomes} an infinite dimensional system of ODEs 
\begin{equation}
  \label{ilambda}
\dot z_j=\blu{-}\mathrm{i}\lambda_j
z_j+\cO(|z|^3),\quad \lambda_j\vcentcolon= c\sqrt{j^2+c^2}, \qquad {j \in \mathbb{Z}}.
  \end{equation}

If we neglect the nonlinear term, the solutions of \eqref{ilambda} are given by a
superposition of infinitely many rotations
\[
z_j(t)=z_j(0)\,e^{\blu{-}\iu\lambda_j t} \qquad \mathrm{and} \qquad |z_j(t)|^2= |z_j(0)|^2 \qquad \forall t\in \R.
\]
This shows that the
phase space is foliated by tori invariant {for the linear flow}.

 We fix $N\geq 3$ and we look for solutions which bifurcate from the
 linear solutions obtained by the excitation of  $N$ linear modes.
 Our main result, Theorem \ref{VicinanzaTh}, ensures that, in a
 sufficiently small neighborhood of the elliptic equilibrium $\psi=0$
 and for sufficiently large values of $c$, {there exist quasi periodic
 in time solutions both to  KG
 and to NLS equations \emph{having a set of
 common frequencies}, which are close to each other uniformly globally in time.

\smallskip

{Before stating our main result we need some preliminary definitions.}

For any $a\geq 0,\, p\geq0$, we define the phase spaces
\begin{equation}
  \label{lap}
\ell^{a, p}\vcentcolon= \left\{ z=\{z_j\}_{j \in \mathbb{Z}}\in\ell^2=\ell^2(\mathbb{Z}; \mathbb{C}) \mid \| z \|_{a, p}^2\vcentcolon= \sum_{j \in \mathbb{Z}
  } |z_j|^2 \, \jbs{j}^{2p}\,e^{2 |j| a}<+\infty  \right\}.
\end{equation}
In the following we will identify a function $\psi$ with the
  corresponding sequence of Fourier coefficients $z$ and say that
  $\psi\in\ell^{a,p}$ if $z\in\ell^{a,p}$.

{We fix now a finite set $\cS \subset\Z$, and we denote 
  \begin{equation}
    \label{insieme}
N\vcentcolon = \# \cJ \geq 3, \quad \cS= \left\{\tj_1,...,\tj_N\right\}, \quad \tj_1<...<\tj_N.
    \end{equation}}
Correspondingly we define an immersion of $\C^N$ in $\ell^{a,p}$ by
  \begin{align}
\label{immersione}
    \iota_\cS :\C^N&\to\ell^{a,p}
    \\
\nonumber
    (\tz_{\tj_1},...,\tz_{\tj_N})&\mapsto \iota_\cS(\tz) \equiv \{z_j\}_{j \in \mathbb{Z}},
  \end{align}
with $z_j=0$ if $j\not\in\cJ$ and $z_{j}=\tz_{j}$ otherwise.

For any $\tR>0$ we define 
\begin{equation}\label{xi0:def}
\Xi_0(\tR)\vcentcolon= \bigg[\frac{1}{2}\tR^2,\frac{3}{2} \tR^2\bigg]^N.
\end{equation}
Any element $\xi \in \Xi_0(\tR)$ is an $N$-dimensional vector indexed by the set $\cS$, i.e. $\xi = (\xi_j)_{j \in \cS}$.

\begin{theorem}\label{VicinanzaTh}
 Fix $p>9/2$, $a \geq 0$ and {$\cS$ as in \eqref{insieme}}, then there exist $C>0$ and $\tR_*>0$ such that for any $\tR \in (0,\tR_*)$ and $c\geq \tR^{{-73/72}}$ the following holds. There exist a set $\caW=\caW (c) \subset \R^N$ with positive measure, two Lipeomorphisms 
\begin{align}
  \label{lipeo}
\caW\ni\omega\mapsto\xi_c^{KG}(\omega)\in \Xi_0(\tR), \quad
\caW\ni\omega\mapsto\xi^{\tt NLS}(\omega)\in \Xi_0(\tR),
\end{align}
and for any $\omega \in \caW$, two embeddings
\begin{align}
{\Psi_\omega^{KG},\Psi_\omega^{\tt NLS}}:\T^N \to \ell^{a,p},
\end{align}
such that the functions
\begin{align}\label{main.1}
  &\psi_{c, \omega}^{KG}(t)\vcentcolon= \Psi_\omega^{KG}((\omega+c^2)\, t),
  \\
  \label{main.2}
  &\varphi^{\tt NLS}_{\omega}(t)\vcentcolon= \Psi_\omega^{ \tt NLS}(\omega\, t),
\end{align}
solve respectively \eqref{KG.psi} and \eqref{NLS}.
Furthermore, we have
\begin{equation}\label{dist.main}
\begin{aligned}
& \frac{1}{\tR}\sup_{\omega \in \caW,\ \theta \in \T^N}
\left\|
\Psi_\omega^{KG}(\theta)
- \iota_\cS\bigl(
  \sqrt{\xi_{\tj_1}^{KG}(\omega)}\,e^{\mathrm{i}\theta_{\tj_1}},
  \dots,
  \sqrt{\xi_{\tj_N}^{KG}(\omega)}\,e^{\mathrm{i}\theta_{\tj_N}}
\bigr)
\right\|_{a,p}
\le C\,\tR^{\frac{1}{12}},\\
& \frac{1}{\tR}\sup_{\omega \in \caW,\ \theta \in \T^N}
\left\|
\Psi_\omega^{\tt NLS}(\theta)
- \iota_\cS\bigl(
  \sqrt{\xi_{\tj_1}^{\tt NLS}(\omega)}\,e^{\mathrm{i}\theta_{\tj_1}},
  \dots,
  \sqrt{\xi_{\tj_N}^{\tt NLS}(\omega)}\,e^{\mathrm{i}\theta_{\tj_N}}
\bigr)
\right\|_{a,p}
\le C\,\tR^{\frac{1}{12}}.
\end{aligned}
\end{equation}
There exists $K>0$,
independent of $\tR$ and $c$, such that for any $\sigma \in \big[0, 1\big]$ it holds 
\begin{equation}
    \label{dist.sole}
\frac{1}{\tR}\sup_{t\in\R}\sup_{\omega \in \caW}{\left\|e^{\iu c^2t }\psi_{c,
    \omega}^{KG}(t)-\varphi_\omega^{\tt NLS}(t)
  \right\|_{a,p-4\sigma}} \leq K
\frac{ \tR^{\frac{1}{36}-\frac{215}{72} \sigma}}{c^{2\sigma}}.
  \end{equation}

\medskip

Moreover, if we define
\begin{align*}
\Xi(c)\vcentcolon= \xi^{\tt KG}_{c}(\caW), \qquad \Xi^{\tt NLS}\vcentcolon= \xi^{\tt NLS}(\caW),
\end{align*}
the following estimates hold
  \begin{equation}
  \label{misura.0}
  \frac{|\Xi_0(\tR)\setminus \Xi(c)|}{ |\Xi_0(\tR)|} \leq C\tR^{{\frac{1}{36}}},\qquad  \frac{|\Xi_0(\tR)
    \setminus \Xi^{\tt NLS}|}{ |\Xi_0(\tR)|}\leq C \tR^{{\frac{1}{36}}} .
  \end{equation}
\end{theorem}

\begin{corollary}\label{Maintheorem}
 Fix $p>9/2$, $a \ge 0$ and $\cS$ as in \eqref{insieme}. 
There exists $\tR_*>0$ such that for any $\tR \in (0,\tR_*)$ and 
$c \in (c_*,\infty)$, with $c_* := \tR^{-73/72}$, the following holds.
For any map
\begin{align*}
(c_*,\infty)\ni c \mapsto \omega(c)\in\caW(c)\subset\R^N, 
\end{align*}
we have
\begin{equation}\label{main.44}
\lim_{c\to+\infty}\sup_{t\in\R}
\left\| e^{\iu c^2t}\psi^{KG}_{c,\omega(c)}(t)
      -\varphi^{\tt NLS}_{\omega(c)}(t)
\right\|_{a,p-4\sigma} = 0,
\end{equation}
for any $\sigma\in(0,1]$.
\end{corollary}
\bigskip

We point out that our results apply also in the case of arbitrary analytic
nonlinearities, namely equations of the form 
\[
\frac{1}{c^2} u_{tt}-u_{xx} +c^2 u+f(u)=0, \qquad f(u)=\sum_{k\geq 3}
a_k\, u^{k}.
\]

We observe that our results hold both for the focusing and defocusing KG equation. For concreteness, from now on we consider $a_1=1$ in \eqref{KG}.

%\green{We also remark that, at the price of reducing the rate of convergence as $c\to\infty$,
%  we can obtain
%  approximations \eqref{dist.sole}, \eqref{main.44} with an arbitrary small loss of regularity.} 

\subsection{Global in time results:
  a purely nonlinear phenomenon}

{We emphasize that the global (in time) results obtained in
Theorems \ref{VicinanzaTh} and Corollary \ref{Maintheorem} are purely nonlinear
phenomena and the uniform in time approximation /convergence result between solutions of  \eqref{dist.sole}
and \eqref{main.44} does not hold at linear level.
Indeed, the solutions $z^{LKG}(t)$ of the linearized KG equation satisfy
\[
e^{\mathrm{i} c^2 t} z_j^{LKG}(t)=z_j(0) e^{-\mathrm{i} (\lambda_j-c^2)  t} \qquad j \in \mathbb{Z},
\]
while the solutions $z^{LS}(t)$ of the linear Schr\"odinger equation are given by
\[
z_j^{LS}(t)=z_j(0) e^{-\mathrm{i} \frac{j^2}{2}  t} \qquad j \in \mathbb{Z}.
\]
Since 
\[
\lambda_j=c^2+\frac{j^2}{2}+ {\lambda_j^R, \qquad \lambda_j^R=\cO\left( \frac{j^4}{ c^2} \right),}
\]
we have
\[
\left|e^{\mathrm {i} c^2t}z^{LKG}(t)-z^{LS}(t)\right|= |z_j(0)
(e^{-\mathrm{i} (\lambda_j-c^2 ) t}-e^{-\mathrm{i} \frac{j^2}{2}
  t})|=|z_j(0)|\,  |1-e^{\mathrm{i} \lambda_j^R t}|.
\]
Since $\limsup_{t\to \infty} |1-e^{\mathrm{i} \lambda_j^R t}|=2$,
independently of $c$, the above difference cannot be small in the non
relativistic limit {uniformly for $t\in\R$}.
    According to Theorem \ref{VicinanzaTh} and Corollary
\ref{Maintheorem} the presence of the nonlinearity {provides a frequency shift that, given a torus of the NLS equation, allows to select
  a torus of the KG equation with {\it exactly} the same frequency. }

This explains why our results are valid for sufficiently large
values of $c$, depending on the size of the solutions we are
considering, and this is not merely a technical condition.

% \green{As usual in
%KAM theory, the small amplitude quasi periodic in time solutions, both
%to KG and NLS equations, form a Cantor manifold which is tangent to
%the plane foliated by invariant tori for the linear equation. Even
%though these Cantor manifolds are dense at the origin, we cannot
%expect a global result of vicinity of the solutions lying on these
%manifolds when we approach the origin. Indeed, although the invariant
%tori for both equations become closer and closer, the displacement
%between the linear frequencies cannot be balanced by the the nonlinear
%terms, which become smaller and smaller when approaching the origin.}

\subsection{{Discussions}}

It is known since 40 years that solutions of the KG equation
converge (after a Gauge transformation) to solutions of the NLS {\it
  over compact intervals of time} (see
  e.g. \cite{tsutsumi1984nonrelativistic},
  \cite{najman1990nonrelativistic},
  \cite{machihara2001nonrelativistic},\cite{faouSchratz},\cite{pasquali2019dynamics}).
 We recall in particular the
works \cite{machihara2002nonrelativistic},
  \cite{masmoudi2002nonlinear} in which the authors study the problem
in $\R^d$, for any {$ d \geq 2$}, and prove
convergence in the energy space.  We also mention the paper
  \cite{nakanishi2008transfer}, where it is proved that the
scattering operator of the KG equation converges to the scattering
operator of the NLS, at least for a particular class of initial
data. In that result a central role is played by dispersion. 

\medskip

The novelty of the present paper is that we construct a class of
solutions of the KG equation which converge \emph{uniformly
for $t\in\R$} to solutions of the NLS equation. 
  We point out that, since
we work on a compact spatial domain, we cannot rely on dispersive
effects. 
Instead, we use KAM techniques to construct quasi periodic in time
solutions of KG and study their limit as $c\to\infty$. 

{It is well known that such limit is singular,}
as one can see in a very clear way
expanding the linear operator in \eqref{KG.psi}, namely writing 
\begin{equation}\label{espansione}
c \langle \nabla \rangle_c=c^2-\frac{1}{2} \Delta+\mathcal{O}\left(
\frac{\Delta^2}{c^2} \right),
\end{equation}
which shows that the expansion in $c^{-2}$, contains more and more
derivatives of the unknown $\psi$ (see \eqref{KG.psi}).  Furthermore, we emphasize that the nonlinear terms of the equation also depend on the parameter $c$, and controlling the nonlinear effects  both in the limit $c\to +\infty$ and over infinite times is not a trivial task.

Here we decided to start the investigation of the nonlinear problem by
analyzing the simplest possible solutions, namely those of small
amplitude in which the nonlinearity is small compared to the linear
part. The typical tools allowing to study such a regime are
perturbative techniques, such as normal forms, or equivalently
averaging methods. Here we use
the tools of Hamiltonian perturbation theory, namely Normal Form and KAM theory, to study the dynamics of
\eqref{KG.psi} uniformly with respect to $c$ and we compare the constructions for the KG equation to that for the NLS
equation.

Here, we use the approach by Kuksin and P\"oschel
\cite{kuksin1996invariant}, \cite{poschel1996quasi} and \cite{berti2013kam}
who used the KAM theory of \cite{Kuk87},\cite{way90}, \cite{poschel1996kam}, to construct families of small
amplitude solutions of the KG and of the NLS equation, developing a KAM theory which
 remains valid in the singular non relativistic limit. In particular we rely on such papers for the parts which are not directly affected by the singular limit nature of the problem, while we write the details for the parts which need new techniques to study the limit $c \to \infty$. 

A substantial difficulty is related to the fact that the frequencies
of the linear part of the KG equation diverge as $c \to \infty$. This
makes not possible to rely on the strategy of \cite{poschel1996kam}
to control the small divisor problem, because there are new
resonances depending on $c$, between tangential modes and arbitrarily
large normal modes. To solve this problem we base our analysis on the
selection rule imposed by space translation invariance of the
equations. In particular, we analyze in detail some resonances
that are usually harmless in the standard theory, namely the ones
involving the sum of two normal frequencies. This is done in Section
\ref{nontrivial:1}.

Furthermore, we need to show that the construction for the KG equation
 converges to the construction for the NLS equation as $c\to\infty$. This brings us to another important point: one of the main differences between the
 KG and NLS systems is that the latter is Gauge invariant, while the former
 is not. In order to show that the non Gauge invariant part of the
 Klein Gordon equation does not affect the convergence of solution, we need to
 analyze in detail the structure of the Hamiltonian to show that the non
 Gauge invariant terms give contributions of order $c^{-2}$. These are the main technical novelties needed to prove our result.

\medskip

We think that our method could be used also to study the
  non relativistic limit of the Maxwell-Klein-Gordon equation (see
  \cite{bechouche2004nonrelativistic},
  \cite{masmoudi2003nonrelativistic}), however one has to solve
  several nontrivial technical difficulties.  Furthermore it would
be very interesting to extend the result of the present paper to
almost periodic solutions following \cite{bernier2024infinite} (see also \cite{biasco2023small}, \cite{corsi2024almost}).

We recall now the result of \cite{pasquali2018almost}, 
  where the author uses normal form theory to prove almost global
  existence in a variant of the Klein-Gordon equation uniformly in the
  non relativistic limit.

We conclude this discussion by mentioning the work by Franzoi-Montalto
\cite{franzoi2024kam}, where the authors {exploit KAM techniques} to prove the uniform
in time convergence of quasi periodic solutions of the forced
Navier-Stokes equation to forced quasi periodic solutions of the Euler
equation in the inviscid limit. 
{This is a singular limit too, however }the difficulties in such a procedure are of completely different
  nature compared to those in the present paper: indeed, the Euler and
  Navier-Stokes (NS) equations are quasilinear, {while the model
    considered here is semilinear.} On the other
  hand, as discussed above, one of the main issues of our results
  is dealing with the fact that the small divisors depend on the
  parameter $c$, while in
  \cite{franzoi2024kam} the viscosity parameter does not affect the small divisors
  of the KAM procedure.  In fact, thanks to the presence of the
  dissipative term $-\nu \Delta$ in the NS equation, the small
  divisors problems are restricted to the inversion of the linearized
  operator of the {limiting equations, namely} the Euler equation at a Euler solution, which is
  independent on the parameter $\nu$.

\medskip
{
\textbf{Notations:}
We collect here some notational conventions used throughout the paper.
\begin{itemize}
  \item 
  For $x, y \in \mathbb{R}$, we write
  \[
    x \lesssim y
  \]
  if there exists a constant $C>0$, independent of all relevant parameters (typically $c$), such that $x \leq C y$.  In the same way, we use $x \gtrsim y$. 
  If both $x \lesssim y$ and $x \gtrsim y$ hold, we write $x \simeq y.$ Moreover, anytime we introduce a constant $C$, it will be independent of $c$ (if it is not differently specified).
  \item 
We denote $h \vcentcolon = c^{-2}$. Both $h$ and $c$ will be used throughout the paper, depending on which is more convenient in a given context, without further reference to their relation.
  \item 
We denote
\begin{equation}\label{def:h}
\nu_j(h)\vcentcolon= \frac{j^2}{1+\sqrt{1+hj^2}}, \qquad \forall \, j \in \mathbb{Z}.
\end{equation}
Then the linear frequencies \eqref{ilambda} can be written as 
\begin{equation}\label{def:lambda2}
\lambda_j=c \sqrt{c^2+j^2}=c^2+\nu_j(h) = h^{-1}+\nu_j(h), \qquad \forall \, j \in \mathbb{Z}.
\end{equation}
Most of the time we will simply write $\nu_j$ omitting the
  dependence on $h$.
  \item For any set $\cZ \subseteq \mathbb{Z}$, fixing $a \in \cZ$ we define 
\begin{equation}\label{Kronecker}
\delta_a\vcentcolon=\{ \delta_{a,j} \}_{j \in \cZ}, \quad \delta_{a,j}\vcentcolon=
\begin{cases} 
1  & j=a, \\
0  & j \neq a.
\end{cases}
\end{equation}
\item For any $x \in \mathbb{R}^N$ we denote 
\[
|x| \vcentcolon = \max_{j=1, \dots, N} |x_j|, \quad |x|_1\vcentcolon=\sum_{j=1}^N |x_j|, \quad |x|_2 \vcentcolon = \bigg( \sum_{j=1}^N |x_j|^2 \bigg)^{1/2}.
\]
  \item  
 For $\ell \in \mathbb{Z}^{\cS^c}$, we define
  \begin{equation}\label{jb}
    \jbs{\ell}_1 \vcentcolon= \max\left(1, \left| \sum_{j \in \cJ } j\ell_j \right| \right).
  \end{equation}
  Moreover, for $k \in \mathbb{Z}^{\cS}$, we set
  \begin{equation}\label{jb0}
    \jbs{k} \vcentcolon= \max \left(1, |k|_1 \right).
  \end{equation}
We use the same notation for $j \in \mathbb{Z}$.
  \item 
  Given a sequence $z \in \ell^2$ and a subset of indices $\cZ \subseteq \Z$, we introduce the projector
  \begin{equation}\label{cut-off}
    \big(\Pi_{\cZ} z\big)_j \vcentcolon=
    \begin{cases}
      z_j, & j \in \cZ, \\
      0, & j \in \Z \setminus \cZ,
    \end{cases}
    \qquad
    \Pi_\cZ^{\perp} z \vcentcolon= z - \Pi_\cZ z,
  \end{equation}
which is the projector on the subspace of $\ell^2$, containing only the sequences indexed on $\cZ$.
 \item 
 Let $F$,$E$ two complex Banach spaces with norms $\| \cdot \|_E$ and $\| \cdot \|_F$. Let $B : F \to E$ a bounded operator, then we denote the operator norm as
\begin{equation}\label{norma:op}
\| B \|_{F,E}\vcentcolon= \sup_{\substack{x \in F, \\ x \neq 0}} \frac{ \|F(x)\|_E}{\|x\|_F}.
\end{equation}
\end{itemize}}

\bigskip

	\noindent
	{\bf Acknowledgements.} D. Bambusi, and A. Belloni have been supported by the research projects PRIN 2020XBFL 
	``Hamiltonian and dispersive PDEs'' of the Italian Ministry of
        Education and Research (MIUR) and by GNFM. 

A. Belloni is partially supported by the ERC STARTING GRANT 2021 "Hamiltonian Dynamics, Normal Forms and Water Waves" (HamDyWWa), Project Number: 101039762. The Views and opinions expressed are however those of the authors only and do not necessarily reflect those of the European Union or the European Research Council. Neither the European Union nor the granting authority can be held responsible for them. 

F. Giuliani has received funding from INdAM/GNAMPA
Project Stable and unstable phenomena in propagation of Waves in dispersive media,
CUP E5324001950001 and from PRIN-20227HX33Z Pattern formation in nonlinear phenomena-
Funded by the European Union-Next Generation EU, Miss. 4-Comp. 1-CUP D53D23005690006.

\vspace{0.2cm}

\noindent\textbf{Statements and Declarations.} Data sharing is not applicable to this article as no datasets were generated or analyzed during the current study.

\noindent\textbf{Competing Interests}: The authors have no conflicts of interest to declare.

\section{Structure of the Klein-Gordon Hamiltonian}

\subsection{The phase space}\label{phase.space.sec}

We introduce here the phase space and some notations concerning the
Hamiltonian structure of the equations \eqref{KG.psi} and \eqref{NLS}.

First, we need to define a further $c$-dependent space. For any $\beta\geq 0$
  we put
\begin{equation}\label{ellappiu}
\ell^{a,p,\beta}\vcentcolon= \left\{ z \in \ell^2 \, | \, \, \| z\|_{a,p,\beta}^2\vcentcolon= 
\sum_{j \in \mathbb{Z}} |z_j|^2 e^{2|j|a} \jbs{j}^{2p} \tw_j^{2\beta} \, < \, \infty
\right\}, \quad \tw_j\vcentcolon=\sqrt{1+\frac{j^2}{c^2}}=\frac{\lambda_j}{c^2},
\end{equation}
with the frequencies $\lambda_j$ defined as in \eqref{ilambda}. {\bf  We fix now once for all a value $p>9/2$ and $a\geq0$, while we will consider
 values of $\beta\in[0,1]$. 
}
{When $\beta=0$ we will sometimes omit to
write the index $\beta$.}
We consider the Fourier expansions 
\begin{equation}
  \label{fourier.2}
\psi(t, x)=\sum_{j \in \mathbb{Z}}z_j(t)\,\frac{1}{\sqrt{2\pi}}e^{\im jx}, \quad \bar{\psi}(t, x)=\sum_{j \in \mathbb{Z}}\bar{z}_j(t)\,\frac{1}{\sqrt{2\pi}}e^{-\im jx}.
\end{equation}
Correspondingly, we define the phase space 
\begin{equation}
  \label{phase.space}
\cP^{a,p,\beta}\vcentcolon= \ell^{a,p,\beta}\times \ell^{a,p,\beta}\ni(z,\bar
z),
  \end{equation}
where $\bar{z}$ is considered as independent of $z$.

For $(z,\bar z)\in\cP^{a,p,\beta}$ we define
\begin{equation}
  \label{norma.doppia}
\left\|(z,\bar
z)\right\|^2_{a,p,\beta}\vcentcolon= \left\|z\right\|^2_{a,p,\beta}+\left\|\bar z\right\|^2_{a,p,\beta}\,.
  \end{equation}
In the following we will denote by $ B_{a,p,\beta}((z,\bar
z), \tR)$ 
the open ball of
radius $\tR$ in $\cP^{a,p,\beta}$ centered at $(z,\bar z)$. For
open balls centered at the origin we will simply write $ B_{a,p,\beta}(\tR)$.

{Finally, recalling \eqref{cut-off}, for any subset of indices $\cZ \subseteq \mathbb{Z}$ we define 
\begin{equation}\label{cut-off-doppio}
\Pi_\cZ(z,\bar
z)\vcentcolon= (\Pi_\cZ z,\Pi_\cZ\bar z), \quad 
\Pi^\perp_\cZ(z,\bar{z}) \vcentcolon= (\Pi_\cZ^\perp z, \Pi_\cZ^\perp
\bar{z}).
\end{equation}}

\subsection{Hamiltonian formalism }
Consider an open set $\cU \subseteq \cP^{a,p,\beta}$. Given a Hamiltonian function $F\in C^{\omega}(\cU;\C)$ \footnote{{It denotes the space of analytic functions from $\cU$ to $\C$.}}, the corresponding Hamiltonian vector
field is defined formally by 
\begin{equation}
  \label{campo.definisco}
X_F\vcentcolon= ([X_F]_z,[X_F]_{\bar{z}}), \quad  [X_F]_{z_j} \equiv ([X_F]_z)_j \vcentcolon= -\im\frac{\partial F}{\partial \bar
  z_{j}}, \, \, \,  [X_F]_{\bar{z}_{j}} \equiv ([X_F]_{\bar z})_j \vcentcolon=\im\frac{\partial F}{\partial
  z_j}.
  \end{equation}

\begin{definition}
Let $X \in C^{\omega}(\cU; \cP^{a,p,\beta})$ with $\cU$ an open subset of $\cP^{a,p,\beta}$ and 
\[
X(z,w)= (X_1(z,w),X_2(z,w))\in\cP^{a,p,\beta}, \qquad (z, w)\in\mathcal{U}.
\]
We will say that $X$ is \emph{real analytic} if, whenever $w$ is the complex conjugate of
$z$, one has that $X_2$ is the complex conjugate of $X_1$, namely $[X_2(z,w)]_j=\overline{[X_1(z,w)]_{j}}$ whenever $w_j=\overline{z_{j}}$ for all $j \in \mathbb{Z}$.
\end{definition}

Let $F_1 \in
  C^{\omega}(\cU;\C)$ and let $F_2\in
  C^{\omega}(\cU;\C)$ be such that  $X_{F_2}\in
  C^{\omega}(\cU;\cP^{a,p,\beta})$, then the Poisson Bracket of $F_1$
  and $F_2$ is defined by
  \begin{equation}
    \label{poisson}
{\left\{F_1,F_2\right\}\vcentcolon=dF_1X_{F_2}=\im \sum_{j \in \mathbb{Z}}\bigg(\frac{\partial F_1}{\partial \bar
  z_{j}}\frac{\partial F_2}{\partial
  z_j}-\frac{\partial F_1}{\partial
  z_j}\frac{\partial F_2}{\partial \bar
  z_{j}}\bigg)}
  \end{equation}
  and is of class $C^{\omega}(\cU; \C)$.
  
  We recall that, if also $X_{F_1}\in
  C^{\omega}(\cU;\cP^{a,p,\beta})$, then 
  \begin{equation}
    \label{reg.poisson}
X_{\left\{F_1,F_2\right\}}=\left[X_{F_1},X_{F_2}\right]
\vcentcolon= dX_{F_1}X_{F_2}-dX_{F_2}X_{F_1} \in
C^{\omega}(\cU;\cP^{a,p,\beta}). 
    \end{equation}

We state some properties of the flows and Poisson Brackets. Let  $G\in C^{\omega}(\cU;\C)$ be such that $X_G\in
  C^{\omega}(\cU;\cP^{a,p,\beta})$ and furthermore
  \begin{equation}
    \label{esisto.1}
\sup_{(z,\bar z)\in\,\cU}\left\|X_G(z,\bar
z)\right\|_{a,p,\beta}<\delta
  \end{equation}
  for some $\delta>0$ .
  Let $\cV \subset\mathcal{U}$ be an open subset such that 
  $$
\bigcup_{(z,\bar z)\in \cV}B_{a,p,\beta}(( z,\bar z),\delta)\subset \cU\ 
$$
(and assume it is not empty),
then for all times $|t|\leq 1$ the flow $\phi_G^t$ of $X_G$ is well defined on $\cV$ and fulfills
\begin{equation}
  \label{esisto.3}
\forall |t|\leq 1,\quad  \phi^t_G\in C^{\omega}(\cV;\cU),\quad 
\sup_{(z,\bar z)\in\cV}\left\|\phi_G^t(z,\bar z)-(z,\bar
z)\right\|_{a,p,\beta}\leq\delta.
  \end{equation}

Furthermore, if $F\in C^{\omega}(\cU;\C)$, then,
in $\cV$ and for all $|t|\leq 1$ one has
\begin{equation}
  \label{esisto.2}
\frac{d}{dt}F\circ \phi_G^t=\left\{ F,G\right\}\circ \phi_G^t.
\end{equation}
If $X_F\in C^\omega(\cU;\C)$ then $X_{F\circ\phi_G^t}\in
C^\omega(\cV;\C)$, and one has
\begin{equation}
  \label{esisto.4}
X_{F\circ\phi_G^t}(z,\bar
z)=\left(d(\phi_G^{t})^{-1}X_F\right)(\phi^t_G(z,\bar z)),\quad \forall
(z,\bar z)\in\cV,\quad \left|t\right|\leq 1.
  \end{equation}
Finally, we recall that for any $s,t \in [-1,1]$ such that $|s+t| \leq 1$ the flow of $G$ satisfies the following equalities on $\mathcal{V}$
\begin{equation}\label{prop:group}
\phi_G^s \circ \phi_G^t= \phi_G^{s+t}, \qquad (\phi_G^s)^{-1}=\phi_G^{-s}\,.
\end{equation}

\subsubsection{Symmetries}
A particular role is played by the Gauge and the
  translation groups, which are generated respectively by the Hamiltonians
  \begin{align}
    \label{gauge}
\cN(z,\bar z) \vcentcolon= \sum_{j\in\Z}z_j\bar z_j,\quad \cM(z,\bar
z) \vcentcolon= \sum_{j\in\Z}j z_j\bar z_j.
  \end{align}
The flows are periodic and given explicitly by
  \begin{equation}
    \label{gruppi}
\phi^\alpha_{\cN}(z,\bar z)=\left\{(e^{-\im\alpha}z_j,e^{\im\alpha}\bar
z_j)\right\}_{j \in \mathbb{Z}},\quad \phi^\alpha_{\cM}(z,\bar z)=\left\{\left(e^{-\im j
  \alpha}z_j,e^{\im j \alpha}\bar z_j\right)\right\}_{j \in \mathbb{Z}}, \qquad \forall \, \alpha\in \T.
  \end{equation}
  
  \begin{definition}\label{def:HamGT}
A domain $\cU\subseteq \mathcal{P}^{a, p, \beta}$ is said to be Gauge invariant if $\phi^\alpha_{\cN}(\cU) \subseteq \cU$, and it is said to be translation invariant if $\phi^\alpha_{\cM}(\cU) \subseteq \cU$. Correspondingly, a Hamiltonian $F$ is called Gauge invariant if $F \circ \phi^\alpha_{\cN} = F$, and it is called translation invariant if $F \circ \phi^\alpha_{\cM} = F$.
\end{definition}
{
\begin{remark}
  \label{inva}
We introduce the notation $(z_j^+,z_j^-) \vcentcolon= (z_j, \bar{z}_j)$. Then for any $n \in \mathbb{N}$, for any $\vsigma \in \{ \pm \}^n$ and $\vj \in \mathbb{Z}^n$, we denote a monomial of degree $n \in \mathbb{N}$ as
\begin{equation}\label{not:compatta}
 z_{\vj}^{\vsigma}\vcentcolon=\prod_{i=1}^n z_{j_i}^{\sigma_i},
\end{equation}
which is 
\begin{align}
  \label{gauge.1}
  \mathrm{Gauge\ invariant}\ &\mathrm{if}\ \sum_{i=1}^n\sigma_i=0,
  \\
  \label{mom.zero}
  \mathrm{translation\ invariant}\ &\mathrm{if}\ \sum_{i=1}^n \sigma_i j_i=\vsigma \cdot \vj =0.
\end{align}
\end{remark}}

Now we need to define in a precise way what we mean by \emph{Gauge
invariant} and by \emph{non Gauge invariant} part of the Hamiltonian. 

\begin{definition}
  \label{gauge}
Let $F\in C^{\omega}(\cU; \mathbb{C})$ for some \textit{Gauge invariant} open set 
$\cU\subseteq \cP^{a,p,\beta}$. We define the \emph{Gauge Invariant Part} of $F$ as
\begin{equation}\label{def:PiG}
\Pi_{\mathcal{G}} F(z, \bar{z})\vcentcolon= \frac{1}{2\pi}\int_0^{2\pi}F\left(\phi_{\mathcal{N}}^\alpha(z, \bar{z})\right)\,d\alpha.
\end{equation}
The \emph{non Gauge invariant part} of $F$ is $\Pi_{\mathcal{NG}}F=F-\Pi_{\mathcal{G}}F$.

A Hamiltonian $F$ will be said to be \emph{Gauge invariant} if $F=\Pi_{\mathcal{G}}F$, while it
will be said to be \emph{non Gauge invariant} if $\Pi_{\mathcal{G}}F=0$. 
\end{definition}

\begin{lemma}\label{lem:gauge2}
Let $\cU\subset \cP^{a,p,\beta}$ be a Gauge invariant open set.
Let $G\in C^{\omega}(\cU; \C)$ be a Gauge invariant function having
an analytic vector field $X_G\in C^{\omega}(\cU; \cP^{a,p,\beta})$. Denote by
$\phi^t_G$ the flow generated by $X_G$. Let $\cV\subset\cU$ be open
and Gauge invariant,
and 
assume  $\phi_G^t\in C^{\omega}(\cV; \cU)$, $\forall t\in I$, with
$I\subseteq \R$ a suitable interval containing $0$.
Then
\begin{itemize}
\item[i)] $\left(\phi_G^t\circ\phi_{\cN}^\alpha\right)(z,\bar z)=
  \left(\phi_{\cN}^\alpha\circ\phi_G^t\right)(z,\bar z)$ for all
  $(z,\bar{z})\in \cV, t\in I$ and $\alpha \in \T$.
\item[ii)] If $F\in C^{\omega}(\cU,\C)$ is Gauge invariant, then, for any fixed $t\in I$, $F\circ\phi_G^t$ is Gauge invariant. Moreover, $\{ F, G\}$ is Gauge invariant.
\item[iii)] If $F\in C^{\omega}(\cU,\C)$ is non Gauge invariant, then, for any fixed $t\in I$, $F\circ\phi_G^t$ is non Gauge invariant.
Moreover, the function $\{ F, G\}$ is non Gauge invariant. 
\end{itemize}
\end{lemma}
\proof Item $i)$ is easily
proved by remarking that the l.h.s. and the r.h.s. fulfill the same Cauchy problem.
To prove item $ii)$, we compute the Gauge invariant part of
$F\circ\phi_G^t$. We have
{
\begin{align*}
\Pi_{\mathcal{G}}\left(F\circ\phi_G^t\right)(z, \bar{z})&=\frac{1}{2\pi}\int_0^{2\pi}
F\left(\phi_G^t(\phi_{\cN}^\alpha(z,\bar z))\right)d\alpha&
\\
&=\frac{1}{2\pi}\int_0^{2\pi}
F\circ\phi^\alpha_{\cN}\circ\phi_G^t(z,\bar z)  d\alpha&
\\
&=(\Pi_{\mathcal{G}}F\circ\phi_G^t)(z, \bar{z})=(F\circ\phi_G^t)(z, \bar{z})   \qquad \forall (z, \bar{z}) \in \mathcal{V}, &
\end{align*}
where the last equality holds because $F$ is Gauge invariant.

Concerning the Poisson bracket, from \eqref{esisto.2} we have that
\begin{align*}
\Pi_{\mathcal{G}} \{ F, G\}(z, \bar{z})&=\frac{1}{2\pi}\int_0^{2\pi} \{ F, G\}\circ\phi_{\cN}^\alpha(z,\bar{z})\,d\alpha\\
&=\frac{1}{2\pi} \int_0^{2\pi} \frac{d}{dt}\bigg|_{t=0} F \circ \phi^t_G\circ\phi_{\cN}^\alpha (z,\bar{z})\,d\alpha\\
&=\frac{d}{dt}\bigg|_{t=0} \Pi_{\mathcal{G}} (F \circ \phi^t_G)(z, \bar{z})=\{ F, G\}(z, \bar{z}) \qquad \forall (z, \bar{z}) \in \mathcal{U}.
\end{align*}

The first part of item $iii)$ is proved in the same way of $ii)$. From this we have that $\Pi_{\mathcal{G}} (F \circ \phi^t_G)\equiv 0$ for all $|t|\le 1$, and we can reason exactly as in item $ii)$ to deal with the Poisson bracket.
\qed

\begin{remark}
  \label{traslo}
By the same argument one has that if $G$ and $F$ are as above, but
translational instead than Gauge invariant, then $F\circ\phi^t_G$ and
$\left\{F, G\right\}$ are also translational invariant.

\end{remark}

\subsection{The Hamiltonians}
 Equation \eqref{KG.psi} is Hamiltonian
  with Hamiltonian function
\begin{equation} \label{HKG.0}
  H(\psi,\bar\psi)=\int_{-\pi}^\pi\bar\psi\, c\,\nablac \psi \,dx+\frac{1}{16}\int_{-\pi}^\pi
  \left[
    \left(\frac{\nablac}{c}\right)^{-1/2}(\psi+\bar\psi)\right]^4dx,
\end{equation}
while the Hamiltonian of equation \eqref{NLS} is given by 
\begin{equation}
  \label{HNLS.0}
H^{\tt NLS}(\psi,\bar\psi)=\frac{1}{2}\int_{-\pi}^\pi \psi_x  \bar{\psi}_x\,dx +\frac{3}{8}\int_{-\pi}^\pi \psi^2 \bar{\psi}^2 dx.
\end{equation}
{
\begin{remark}
  \label{invario}
The Hamiltonian $H$ is translation invariant, and $H^{\tt NLS}$ is both translation and Gauge invariant. Hence all Hamiltonians under consideration contain only translation invariant monomials, which satisfy \eqref{mom.zero}.
\end{remark}}

We decompose the Hamiltonian \eqref{HKG.0} of the Klein-Gordon equation as
\begin{equation}\label{HKG}
  H=\Lambda+ P.
\end{equation}
Recalling the notation introduced in \eqref{not:compatta}, we define
\begin{align}
\label{LKG}
  &\Lambda(\psi, \bar{\psi})\vcentcolon= \int_{-\pi}^\pi\bar\psi\, c\,\nablac \psi \,dx= { \sum_{j\in \mathbb{Z}}
  \lambda_j z_j \bar{z}_{j} ,\qquad \lambda_j\vcentcolon= c\sqrt{j^2+c^2}\  ,\qquad j\in \mathbb{Z}},\\
\label{PKG}
&P(\psi, \bar{\psi})\vcentcolon= \frac{1}{16}\int_{-\pi}^\pi
  \left[
    \left(\frac{\nablac}{c}\right)^{-1/2}(\psi+\bar\psi)\right]^4dx=\sum_{\vsigma \in \{\pm \}^4, \, \,  \vj \in \mathbb{Z}^4}P_{\vj, \vsigma}\, z_{\vj}^{\vsigma}&
\end{align}
with the coefficients of $P$ defined as
\begin{equation}\label{G}
\begin{aligned}
&P_{\vj, \vsigma}\vcentcolon= \begin{cases}
\dfrac{1}{32 \pi}\dfrac{1}{\sqrt{\tw_{j_1}\,\tw_{j_2}\,\tw_{j_3}\,\tw_{j_4}}}\displaystyle\binom{4}{\hat{\sigma}}, & \mathrm{if}\,\, \vsigma \cdot \vj
  =0,\\
  0 & \mathrm{otherwise},
    \end{cases}
\end{aligned}
\end{equation}
where the weights $\tw_j$ are introduced in \eqref{ellappiu}, {$\binom{4}{\hat{\sigma}}$ is the binomial coefficient} and
\begin{equation}\label{hatsig}
\hat{\sigma} \vcentcolon=\frac{4+\sum_{m=1}^4 \sigma_m}{2}.
\end{equation}

\smallskip

In the same way, we write the NLS Hamiltonian as 
\[
H^{\tt NLS}=\Lambda^{\tt NLS}+P^{\tt NLS}
\]
with 
\begin{align}
\label{N:NLS}
  &\Lambda^{\tt NLS}(\psi, \bar{\psi})\vcentcolon= \frac{1}{2}\int_{-\pi}^\pi  \psi_x  \bar{\psi}_x \,dx = \sum_{j \in \mathbb{Z}}
  \lambda_j^{\tt NLS} z_j \bar{z}_{j}, \quad \lambda_j^{\tt NLS}\vcentcolon= j^2/2\, \quad j \in \mathbb{Z},\\
\label{P0:NLS}
&P^{\tt NLS}(\psi, \bar{\psi})\vcentcolon=\frac{3}{8} \int_{-\pi}^\pi\psi^2 \bar{\psi}^2 dx=
\sum_{\substack{\vsigma \in \{\pm \}^4, \, \,  \vj \in \mathbb{Z}^4 \\ \sum_i \sigma_i=0}} P^{{\tt NLS}}_{\vj,\vsigma}\,z_{\vj}^{\vsigma},&
\end{align}
and the coefficients are given by
\begin{equation}\label{P:NLS}
P^{{\tt NLS}}_{\vj,\vsigma}
\vcentcolon= \begin{cases}
\dfrac{3}{16\pi} , & \mathrm{if}\,\, \vsigma\cdot\vj
  =0,\\
  0 & \mathrm{otherwise}\,.
    \end{cases}
\end{equation}

\subsection{Decomposition of the KG Hamiltonian}

We are going to decompose the Hamiltonian \eqref{HKG} as follows
$$
H=H^{\tt NLS}+H^{\ng}+H^{R},
$$
where $H^{\ng}$ is non Gauge invariant, while the vector field of $H^R$
locally uniformly vanishes as $h \to 0^+$ {as a map from $\cP^{a,p}$ to $\cP^{a,p-4}$.}

\medskip

Concerning the quadratic part $\Lambda$ of the Hamiltonian \eqref{HKG} we have the following remark.

\begin{remark}\label{expaHlin}
From definitions \eqref{LKG} and \eqref{N:NLS}, we have 
\[
\Lambda=\Lambda^h+\Lambda^{\tt NLS}+\Lambda^R,
\]
  where
  $$
\Lambda^h(\psi, \bar{\psi})\vcentcolon= c^2\int_{-\pi}^\pi \psi \, \bar \psi dx= h^{-1}\int_{-\pi}^\pi \psi \, \bar \psi dx,
  $$
and $\Lambda^R\vcentcolon= \Lambda-\Lambda^h-\Lambda^{\tt NLS}$, which fulfills 
\begin{equation*}
\left\|X_{\Lambda^R}(\psi)\right\|_{a,p-4}\lesssim
h\, \left\|\psi\right\|_{a,p}, \qquad \forall \psi\in \ell^{a, p}.
\end{equation*}
The above estimate follows by 
\begin{equation}\label{difffreq}
\bigg| \nu_j(h)- \lambda_j^{\tt NLS} \bigg|= \frac{j^2}{2}\bigg| \frac{2}{1+\sqrt{1+hj^2}}-1\bigg| \leq h\, \frac{j^4}{2},
\end{equation}
with $\nu_j$ defined as in \eqref{def:h}.
\end{remark}

Concerning the quartic part $P$ of the Hamiltonian \eqref{HKG}, we have the following.
\begin{lemma}
  \label{decompP}
 We have
  \begin{equation*}
P=P^{\tt NLS}+P^{\mathcal{NG}}+P^{R},
  \end{equation*}
  where
  \begin{align}
\label{P:KG0}
    &P^{\tt NLS}(\psi, \bar{\psi})\vcentcolon= \int_{-\pi}^\pi\frac{3}{8}\psi^2 \bar \psi^2 dx\,,&
    \\
\label{P:NG0}
    &P^{\mathcal{NG}}(\psi, \bar{\psi})\vcentcolon= \frac{1}{16}\int_{-\pi}^\pi\left( \psi^4+4\psi^3\bar \psi+4\bar\psi^3\psi+\bar\psi^4 \right) dx&
  \end{align}
  and $P^R \vcentcolon= P-P^{\tt NLS}-P^{\mathcal{NG}}$. It holds that
  $\Pi_\cG P^{\ng}=0$, and for all $\tR >0$ we have
  \begin{align} 
\label{stima.0.0}
&\sup_{{(z,\bar z)} \in B_{a,q,0}(\tR)}\left\|X_{P}{(z,\bar z)}
  \right\|_{a,q,1}\lesssim \tR^3,&\\
\label{stima.1}
 &\sup_{{(z,\bar z)} \in B_{a,p}(\tR)}\left\|X_{P^{R}}{(z,\bar z)}
  \right\|_{a,p-2}\lesssim h \tR^3,& \\
\label{stima.2}
 &\sup_{{(z,\bar z)} \in B_{a,q}(\tR)}\left\|X_{P^{\mathcal{NG}}}{(z,\bar z)}
  \right\|_{a,q}\lesssim \tR^3,& \\
\label{stima.3}
 &\sup_{{(z,\bar z)} \in B_{a, q}(\tR)}\left\|X_{P^{\tt NLS}}{(z,\bar z)}
  \right\|_{a,q}\lesssim \tR^3&
  \end{align}
for both $q=p$ and $q=p-4$.
\end{lemma}
\begin{proof}
The proof is based on the application of Lemma \ref{sti.abs.2}, since
the Hamiltonians involved in the estimates have the structure
\eqref{sti.abs.1} with $n=4$ and suitable weights $b^{(t)}$. We now present the argument explicitly.

We start
by the estimate \eqref{stima.0.0}. In this case, by \eqref{G}, the Hamiltonian $P$
has the structure \eqref{sti.abs.1} with
$b_j^{(t)}={b_j}\vcentcolon= \tw_j^{-1/2}$ for any $t=1, \dots 4$ and coefficients
\[
F_{\vj, \vsigma}\vcentcolon=\frac{1}{32 \pi}\binom{4}{\hat{\sigma}},
\]
with $\hat{\sigma}$ defined as in \eqref{hatsig}.
So, defining $w^{(k)}$ as in \eqref{wk}, we have
$$
\left\|w^{(k)}\right\|_{a,q,1/2}\leq \left\|(z,\bar
z)\right\|_{a,q,0},
$$
and therefore
\begin{align*}
\left\|X_P(z,\bar z)\right\|_{a,q,1}\leq \left\|b
\left(w^{(1)}\star w^{(2)}\star w^{(3)}\right)\right\|_{a,q,1}\leq \left\|
w^{(1)} \star w^{(2)} \star w^{(3)}\right\|_{a,p,1/2}\sleq  \left\|(z,\bar
z)\right\|_{a,q,0}^3,
\end{align*}
for both $q=p$ and $q=p-4$. 

{We now prove \eqref{stima.1}. By the definition of $P^R$, \eqref{P:KG0} and \eqref{P:NG0}, it holds
\[
P^R=\Pi_{\mG}P-P^{\tt NLS}.
\]
Since $\Pi_{\mG}P$ is the Gauge invariant part of $P$, by \eqref{gauge.1} its coefficients are defined by \eqref{G} restricted to the case $\hat{\sigma}=2$. Then, $P^R$ can be written as}
\begin{align}
\nonumber
P^R(z,\bar{z})=\frac{3}{16\pi}\sum_{\vj,\vsigma}\left[\frac{1}{
 \sqrt{\tw_{j_1}\tw_{j_2}\tw_{j_3}\tw_{j_4}}  }-1\right]z^{\vsigma}_{\vj}
\end{align}
where the sum is always restricted to the indices
such that $\vsigma\cdot\vj=0$.
First, we can rewrite the square bracket as follows
\begin{align}
  \label{sti.PR.1}
\left[\frac{1}{
 \sqrt{\tw_{j_1}\tw_{j_2}\tw_{j_3}\tw_{j_4}}
  }-1\right]=\frac{1-\sqrt{\tw_{j_1}}}{\sqrt{\tw_{j_1}}}+\frac{1-\sqrt{\tw_{j_2}}}{\sqrt{\tw_{j_2}}}\frac{1}{\sqrt{\tw_{j_1}}}
+\frac{1-\sqrt{\tw_{j_3}}}{\sqrt{\tw_{j_3}}}\frac{1}{\sqrt{\tw_{j_1}\tw_{j_2}}}
\\
+\frac{1-\sqrt{\tw_{j_4}}}{\sqrt{\tw_{j_4}}}\frac{1}{\sqrt{\tw_{j_1}\tw_{j_2}\tw_{j_3}}},
\end{align}
which gives rise to four terms. We analyze the first one that we call
$P^{R_1}$, the others being similar. The Hamiltonian $P^{R_1}$ has the form \eqref{sti.abs.1}, with weights given by
\[
b_j^{(1)}=\frac{\sqrt{\tw_j}-1}{\sqrt{\tw_j}}\sleq h\, j^2, \quad
b_j^{(t)}=1,\ t=2,3,4,
\]
so that we have the estimates (recall \eqref{wk})
  \[
  \left\|w^{(1)}\right\|_{a,p-2}\sleq h\left\|(z,\bar
  z)\right\|_{a,p},\quad    \left\|w^{(t)}\right\|_{a,p}\sleq \left\|(z,\bar
  z)\right\|_{a,p} ,\ t=2,3,4.
  \]
Thus, by \eqref{sti.abs.6} 
\begin{align*}
\|X_{P^{R_1}}(z,\bar{z})\|_{a,p-2} &\lesssim \| b^{(1)} ( w^{(2)} \star w^{(3)} \star w^{(4)})\|_{a,p-2} 
+\|  w^{(1)} \star w^{(2)} \star w^{(3)}\|_{a,p-2}&
\\
&\sleq
h\left\|(z,\bar z)\right\|^3_{a,p},&
\end{align*}
where we also used the estimate \eqref{conv.n}.
The other terms are estimated similarly and give
\eqref{stima.1}. The estimates \eqref{stima.2} and \eqref{stima.3} follow again from Lemma \ref{sti.abs.2}, by similar (actually easier) computations.
\end{proof}

By collecting the above results we have
\begin{equation} \label{H.de}
H=H^{\tt NLS}+\Lambda^h+\Lambda^R+P^{\mathcal{NG}}+P^R,
\end{equation}
where the first term is the NLS Hamiltonian \eqref{HNLS.0} and the
other terms have been defined and estimated in Remark \ref{expaHlin}
and Lemma
\ref{decompP}.

\section{Birkhoff Normal Form}

Following \cite{poschel1996quasi}, \cite{kuksin1996invariant} and \cite{berti2013kam} we
perform the first step of normal form to extract from the nonlinearity
the parameters needed to develop KAM
theory.

\subsection{Normal Form step}

The main result of this subsection is the following theorem, which is a
version of Proposition 7.1 of \cite{berti2013kam},
which takes into account the dependence on $c$ of all the objects and their behavior as $c\to +\infty$.
\begin{theorem}\label{Teorema0}
There exist $h_* \in (0,1)$ and $\tR_*>0$  such that the following holds. For all $\tR\in(0, \tR_*)$ and $h \in (0,h_*)$, there exist two real
  analytic canonical transformations
\begin{align*}
\cT_0 (h; \cdot ): {B}_{a,q,{\beta}}(\tR) \to {B}_{a,
  q,{\beta}}(2\tR), \quad
\cT_0^{\tt NLS}: {B}_{a,q}(\tR) \to {B}_{a,
  q}(2\tR),
\end{align*}
for any $\beta \in [0,1]$ and for $q=p$, $q=p-4$, and \emph{translation invariant} Hamiltonians $\Lambda_+, \hat{P}, P_0, \Lambda_+^{\tt NLS}, \hat{P}^{\tt NLS}$ and $P_0^{\tt NLS},$ such that the following decompositions hold
  \begin{align}\label{NFS}
H \circ \cT_0&=\Lambda+\Lambda_+ + \hat{P}+ P_0, &\\
\label{NFS:NLS}
H^{\tt NLS}\circ\cT_0^{\tt NLS}&=\Lambda^{\tt NLS}+\Lambda_+^{\tt NLS}+\hat{P}^{\tt NLS}+P_0^{\tt NLS},&
\end{align}
with the Hamiltonians $H$ and $H^{\tt NLS}$ given in \eqref{HKG.0} and \eqref{HNLS.0}. The Hamiltonians in \eqref{NFS:NLS} are also Gauge invariant.

Furthermore, the following holds
\begin{enumerate}[label=(\alph*)] 
\item \label{item:flussi} We have
\begin{align}
  \label{sti.tra.1}
&\sup_{(z,\bar z) \in B_{a,q,0}(\tR)}\left\|\cT_0(h; z,\bar z)-(z,\bar z) \right\|_{a,q, 1}\lesssim
\tR^3,
\\
\label{res.tr.0}
& \sup_{(z,\bar z) \in B_{a,q}(\tR)}\left\|\cT_0^{\tt NLS}(h;
z,\bar z)-(z,\bar z)\right\|_{a,q}\lesssim
\tR^3,  \quad \sup_{(z,\bar z)\in B_{a,p}(\tR)}\left\|\cT_0^{R}(h;z,\bar z)\right\|_{a,p-4}\lesssim
h \tR^3, 
\end{align}
for $q=p$ and $q=p-4$, where 
\begin{equation}\label{dec:gamma}
\cT_0^R \vcentcolon=  \cT_0-\cT_0^{\tt NLS}.
\end{equation}
 \item \label{intem:normale}  $\Lambda_+$  and $\Lambda_+^{\tt NLS}$ are given by
   \begin{align}
     \label{res.nij}
&\Lambda_+\vcentcolon= \frac{1}{2} \sum_{i\,or\, j\in\cJ} \frac{N_{ij}}{(1+h\nu_j)(1+h\nu_i)}z_i \bar{z}_{i} z_j \bar{z}_{j}, \qquad N_{ij}\vcentcolon= \frac{3}{8 \pi}\,(2-\delta_{i,j}),&\\
\label{L+NLS}
&\Lambda_+^{\tt NLS}\vcentcolon= \frac{1}{2} \sum_{i\,or\, j\in\cJ} N_{ij} z_i \bar{z}_{i} z_j \bar{z}_{j}&
     \end{align}
and satisfy the estimates
\begin{equation}\label{res.0.2}
\sup_{(z,\bar z) \in B_{a,q}(\tR)}\left\|X_{\Lambda_+}(z,\bar z)\right\|_{a,q,1}\lesssim \tR^3, \qquad \sup_{(z,\bar z) \in B_{a,q}(\tR)}\left\|X_{\Lambda_+^{\tt NLS} }(z,\bar z)\right\|_{a,q}\lesssim \tR^3.
\end{equation}
Moreover, if we define $\Lambda^R_+ \vcentcolon= \Lambda_+-\Lambda_+^{\tt NLS}$, the following estimate holds true
\begin{equation}\label{stimaLR}
\sup_{(z,\bar z) \in B_{a,p}(\tR)}\left\|X_{\Lambda^R_+}(z,\bar z)\right\|_{a,p-2}\lesssim
h \tR^3,
\end{equation}
for $q=p$ and $q=p-4$.
\item \label{item:pertu} We have
\begin{equation*}
\hat{P}=\mathcal{O}(\| \Pi^{\perp}_\cJ (z,\bar z) \|_{a,q}^4), \qquad
P_0=\mathcal{O}(\| (z,\bar z) \|_{a,q}^6),
\end{equation*} 
with $\Pi^\perp_\cJ$ defined in \eqref{cut-off-doppio}. There exist Hamiltonians $\hat{P}^R$, $P_0^R$, $\hat{P}^{\mathcal{NG}}$ and $P^{\mathcal{NG}}_0$ such that 
\begin{align*}
&\hat{P}=\hat{P}^{\tt NLS}+\hat{P}^{\mathcal{NG}}+\hat{P}^{R},& \quad  &P_0=P_0^{\tt NLS}+P_0^{\mathcal{NG}}+P_0^{R},&\\
&\Pi_{\mathcal{G}} \hat{P}^{\mathcal{NG}}=0,& \qquad &\Pi_{\mathcal{G}} P_0^{\mathcal{NG}}=0,&
\end{align*}
fulfilling the following estimates
\begin{align}
\label{res.0.0}
&\sup_{(z,\bar z)\in B_{a,q,0}(\tR)}\left\|X_{P_0}(z,\bar z)\right\|_{a,q,1}\lesssim \tR^5,
&\qquad &
\sup_{(z,\bar z)\in B_{a,q}(\tR)}\left\|X_{P_0^{\mathcal{NG}}}(z,\bar z)\right\|_{a,q}\lesssim \tR^5,\\
\label{res.0.1}
&\sup_{(z,\bar z)\in B_{a,q,0}(\tR)}\left\|X_{\hat{P}}(z,\bar z)\right\|_{a,q,1}\lesssim \tR^3,
&\qquad &
\sup_{(z,\bar z)\in B_{a,q}(\tR)}\left\|X_{\hat{P}^{\mathcal{NG}}}(z,\bar z)\right\|_{a,q}\lesssim \tR^3,\\
\label{res.NLS}
&\sup_{(z,\bar z)\in B_{a,q}(\tR)}\left\|X_{P_0^{\tt NLS}}(z,\bar z)\right\|_{a,q}\lesssim \tR^5,
&\qquad &
\sup_{(z,\bar z)\in B_{a,q}(\tR)}\left\|X_{\hat{P}^{\tt NLS }}(z,\bar z)\right\|_{a,q}\lesssim \tR^3,\\
\label{res.0.3}
&\sup_{(z,\bar z)\in B_{a,p}(\tR)}\left\|X_{P_0^{R}}(z,\bar z)\right\|_{a,p-4}\lesssim h\tR^5,
&\qquad &
\sup_{(z,\bar z)\in B_{a,p}(\tR)}\left\|X_{\hat{P}^{R}}(z,\bar z)\right\|_{a,p-4}\lesssim h\tR^3,
\end{align}
for both $q=p$ and $q=p-4$.

\end{enumerate}

\end{theorem}

\begin{remark}
We point out that by \eqref{dec:gamma} and estimates \eqref{sti.tra.1} the map $\cT_0^{R}$ is bounded on ${\cP^{a,
  p}}$. Moreover, the bound \eqref{res.tr.0} shows that $\cT_0^{R}$ converges uniformly to zero as $h\to 0$ as a map from a ball of $\cP^{a,p}$ to
$\cP^{a,p-4}$. 
\end{remark}

{The following sections are devoted to the proof of the above theorem. In Section \ref{can:trasf} we introduce the canonical transformations 
$\cT_0$ and $\cT_0^{\tt NLS}$. 
They are the flows of auxiliary Hamiltonians $G$ and $G^{\tt NLS}$, 
which arise as solutions of two cohomological equations.
In Section \ref{fine}, we then prove Theorem \ref{Teorema0}. }

\subsubsection{Generator of the canonical transformations}\label{can:trasf}

As usual, the transformation $\cT_0$ will be constructed as the time
one flow map of an auxiliary Hamiltonian $G$. If $\phi^t_G$ is the
flow generated by $X_G$, then for any function $F$ defined on an open
set of $\cP^{a, p, \beta}$, expanding in Taylor series in $t$ and exploiting
\eqref{esisto.2},
we obtain
\begin{align*}
F \circ \phi^t_G \big|_{t=1}&= F+  \int_0^1  \{F,G\} \circ \phi_G^t\,  \diff t&\\
&= F+\{F,G\} + \int_0^1 (1-t) \{\{F,G\},G\} \circ \phi_G^t \, \diff t.&
\end{align*}
Therefore, if $G$ is a homogeneous polynomial of degree four, the KG
Hamiltonian \eqref{HKG} transforms into 
\begin{equation}
  \label{h_0}
H\circ\phi_G^1=\Lambda+\left\{\Lambda, G\right\}+P+P_0,
\end{equation}
with 
\begin{equation}\label{P0}
P_0=\int_0^1\{P,G\}\circ\phi_G^s\,ds+\int_0^1(1-s)\{\{\Lambda,G\},G\}\circ\phi_G^s\,ds,
\end{equation}
which has a zero of order six at the origin. 

\medskip

The auxiliary Hamiltonian $G$ is chosen in such a way to simplify the
term $\left\{\Lambda, G\right\}+P$ of \eqref{h_0}. So we look
for $G$ as a solution of the following
cohomological equation
\begin{equation}
  \label{homologicala}
\left\{\Lambda, G\right\}+P=\Lambda_+ + \hat{P}, \qquad
\hat{P}=\mathcal{O}(\|\Pi_\cJ^{\perp} (z,\bar z)\|_{a,p}^4),
\end{equation}
with $\Lambda_+$ defined in \eqref{res.nij}. {Similarly, $\cT_0^{\tt NLS}$ is constructed as a time-$1$ flow of $G^{\tt NLS}$, which is a solution of the cohomological equation
\begin{equation}
  \label{homoNLS}
\left\{\Lambda^{\tt NLS}, G^{\tt NLS}\right\}+P^{\tt NLS}=\Lambda_+^{\tt NLS}+\hat{P}^{\tt NLS}, \qquad \hat{P}^{\tt NLS}=\mathcal{O}(\|\Pi_\cJ^{\perp} (z,\bar z)\|_{a,p}^4),
\end{equation}
and $\Lambda_+^{\tt NLS} $ defined in \eqref{L+NLS}.}

To construct $G$ the following definition is useful.
{
\begin{definition}
  \label{resonant}
We define $\cIR\subset \Z^4\times\left\{\pm\right\}^4$ as the set of the
indices $(\vj,\vsigma)=((j_1, \dots, j_4), \allowbreak (\sigma_1, \dots, \sigma_4))$ such that there exists a permutation of order four $\pi \in \mathcal{S}_4$ for which
\begin{equation}
  \label{7.21}
j_{\pi(1)}=j_{\pi(2)},\quad j_{\pi(3)}=j_{\pi(4)},\quad \sigma_{\pi(1)}=-\sigma_{\pi(2)},\quad
\sigma_{\pi(3)}=-\sigma_{\pi(4)}.
  \end{equation}
\end{definition}}

\begin{lemma}
  \label{sol.homo}
  There exists $h_* \in (0,1)$ such that the following holds. For any $h \in (0,h_*)$ and $\tR >0$, the cohomological equation \eqref{homologicala} admits a solution  $G$ which fulfills 
\begin{equation}
\label{sti.g0}
\sup_{(z,\bar z) \in B_{a,q,0}(\tR)}\left\|X_G(z,\bar z)\right\|_{a,q,1}\lesssim
\tR^3.
\end{equation}
Moreover, it holds the decomposition
  $G=G^{\tt NLS}+G^R$, where
  \begin{itemize}
 \item[\textnormal{(i)}] $G^{\tt NLS}$ is a solution of the cohomological equation \eqref{homoNLS}, Gauge invariant (according to Definition \ref{def:HamGT}), independent of $h$ and
fulfills 
\begin{equation}\label{G.NLS}
\sup_{(z,\bar z)\in B_{a,q}(\tR)}\left\|X_{G^{\tt NLS}}(z,\bar z)\right\|_{a,q}\lesssim
\tR^3.
\end{equation}
\item[\textnormal{(ii)}] By defining $G^R:=G-G^{\tt NLS}$, the following estimate holds
\begin{align} 
\label{sti.g1}
&\sup_{(z,\bar z)\in B_{a, p}(\tR)}\left\|X_{G^{R}}(z,\bar z)\right\|_{a,p-4}\lesssim h
\tR^3, \quad  \forall \,  \tR >0.&
\end{align}
\end{itemize}
\end{lemma}
\begin{proof}
By recalling the notation \eqref{not:compatta}, we consider
\begin{align}
  \label{decoG}
G\,\,  =\sum_{ (\vj,\vsigma) \in \Z^4\times\left\{\pm\right\}^4 }G_{\vj, \vsigma} z_{\vj}^{\vsigma},
\end{align}
and we have that
\begin{equation}\label{perentesi:como}
\{\Lambda ,G \}=\iu\sum_{(\vj,\vsigma) \in \Z^4\times\left\{\pm\right\}^4 }\vsigma\cdot\vlambda\,  G_{\vj, \vsigma}\,z_{\vj}^{\vsigma},
\end{equation}
where we denoted $\vsigma\cdot\vlambda \vcentcolon= \sum_i\sigma_i\lambda_{j_i}$.
By the l.h.s. of \eqref{homologicala}, we would like to define the coefficients as
\begin{equation}
  \label{vorrei}
G_{\vj, \vsigma} \vcentcolon= \frac{P_{\vj, \vsigma}}{\im\vsigma\cdot\vlambda},
  \end{equation}
provided that the denominators do not vanish. This is the case for the indices which do not belong to the ``resonant set'' $\cIR$ given in Definition \ref{resonant}.  According to Lemma \ref{divisori} the denominators in \eqref{vorrei}
are uniformly bounded from below for all indices in $\cIR^c$ 
s.t. $\vsigma\cdot\vj=0$ and at least one index is in $\cJ$.

We thus define
\begin{equation}
  \label{LN}
\mathcal{L}_\cJ \vcentcolon=  \{(\vj,\vsigma)\in
\Z^4\times\left\{\pm\right\}^4\ :\ \vsigma\cdot\vj=0,\ \text{and}\ \exists
i\ :\ j_i\in\cJ\, \}.
  \end{equation}

Then, we can define $G$ as follows
\[
G \vcentcolon= \sum_{(\vj,\vsigma)\in\cL_{\cJ}\setminus\cIR}G_{\vj,\vsigma}\,z_{\vj}^{\vsigma},
\]
with the coefficients as in \eqref{vorrei}. From \eqref{perentesi:como} we have
\[
\{\Lambda,G\}+P=\Lambda_++\hat{P},
\]
with
{
\begin{equation}\label{P.capp}
\hat{P}(z,\bar
z)=\sum_{(\vj,\vsigma)\in(\cJ^c)^4 \times \{\pm\}^4}P_{\vj, \vsigma}\,z_{\vj}^{\vsigma}, \quad \Lambda_+=\sum_{(\vj,\vsigma)\in \cL_{\cJ} \cap \cIR }P_{\vj, \vsigma}\,z_{\vj}^{\vsigma}.
\end{equation}
From Definition \ref{resonant}, the definition of $\cL_{\cJ}$ in \eqref{LN} and the expression of the coefficients $P_{\vj, \vsigma}$ in \eqref{G}, we obtain 
\[
\Lambda_+=\sum_{(\vj,\vsigma)\in \cL_{\cJ} \cap \cIR }P_{\vj, \vsigma}\,z_{\vj}^{\vsigma}=\frac{1}{2} \sum_{i\,or\, j\in\cJ} \frac{3}{8 \pi}\frac{(2-\delta_{i,j})}{\tw_j \tw_i}\,z_i \bar{z}_{i} z_j \bar{z}_{j}, 
\]
which corresponds exactly to \eqref{res.nij}. }

We now prove 
\eqref{sti.g0}. Since, by Lemma \ref{divisori}, the denominators are uniformly
bounded from below, we have 
\[
G_{\vj,\vsigma}\sleq P_{\vj,\vsigma} \lesssim \frac{1}{(\tw_{j_1}\tw_{j_2}\tw_{j_3}\tw_{j_4})^{1/2}}.
\]
Using Lemma \ref{sti.abs.2} as in the proof of \eqref{stima.0.0} we
get 
\begin{equation}\label{StimaF0}
\sup_{(z,\bar{z}) \in B_{a,q}(\tR)}\| X_G (z,\bar{z})\|_{a,q,1} \lesssim \tR^3,
\end{equation}
for both $q=p$ and $q=p-4$.

{Following the same argument, we can define 
\[
G^{\tt NLS} \vcentcolon= \sum_{\substack{(\vj,\vsigma)\in\cL_\cJ\setminus\cIR\\ \sum_{i}\sigma_i=0}}
\frac{P^{\tt NLS}_{\vj, \vsigma}}{\im\vsigma\cdot\vlambda\null^{\tt NLS}}z_{\vj}^{\vsigma}.
\]
It is a solution of \eqref{homoNLS} with $\Lambda_+^{\tt NLS}$ as in \eqref{L+NLS} and 
\begin{equation}\label{P.capp.NLS}
\hat{P}^{\tt NLS} (z,\bar
z)=\sum_{(\vj,\vsigma)\in(\cJ^c)^4 \times \{\pm\}^4}P_{\vj, \vsigma}^{\tt NLS} z_{\vj}^{\vsigma},
\end{equation}
with the coefficients given \eqref{P:NLS}. From the same argument of \eqref{sti.g0} we obtain \eqref{G.NLS}.}

\medskip

Finally, we need to prove the estimate \eqref{sti.g1} for $G^R$. We expand the Hamiltonians $P^{\mathcal{NG}}$ and $P^R$, defined in Lemma \ref{decompP}, as follows
\begin{align}
\label{decoP.1}
&P^{\mathcal{NG}}=\sum_{\substack{(\vj,\vsigma) \in \Z^4\times\left\{\pm\right\}^4 \\ \sum_{i}\sigma_i\not=0}}P_{\vj,\vsigma}^{\mathcal{NG}}z_{\vj}^{\vsigma}, \qquad P^{R}\,\, =\sum_{\substack{(\vj,\vsigma) \in \Z^4\times\left\{\pm\right\}^4 \\ \sum_{i}\sigma_i=0}}
P_{\vj,\vsigma}^{R}z_{\vj}^{\vsigma}.&
\end{align}
Then we denote
\begin{align}
  \label{def.G.1}
  G^{\ng} \vcentcolon= \sum_{\substack{(\vj,\vsigma)\in\cL_\cJ\\ \sum_{i}\sigma_i\not=0}}\frac{P^{\ng}_{\vj,\vsigma}}
  {\im\vsigma\cdot\vlambda}z_{\vj}^{\vsigma}, \quad \gna \vcentcolon=  \sum_{\substack{(\vj,\vsigma)\in\cL_\cJ\setminus\cIR\\ \sum_{i}\sigma_i=0}}\frac{P^{\tt NLS}_{\vj,\vsigma}}{\im\vsigma\cdot\vlambda}z_{\vj}^{\vsigma}, \quad  G^{R1} \vcentcolon= \sum_{\substack{(\vj,\vsigma)\in\cL_\cJ\setminus\cIR\\ \sum_{i}\sigma_i=0}}\frac{P^{
    R}_{\vj,\vsigma}}{\im\vsigma\cdot\vlambda}z_{\vj}^{\vsigma}.
\end{align}

{Since $G^R \vcentcolon= G-G^{\tt NLS}$, we have the following decomposition 
\begin{equation}\label{decomp:G}
G^R= G^{R1}+G^{\mathcal{NG}}+(\gna- G^{\tt NLS}).
\end{equation}}

Concerning the non Gauge invariant terms, by the bound \eqref{divisori.2} in Lemma \ref{divisori}, the
denominators in \eqref{def.G.1} are of size $h^{-1}$, so that 
\begin{equation}\label{sti:sup}
\sup_{\vj,\vsigma }\left|G_{\vj,\vsigma}^{\mathcal{NG}}\right| \lesssim h.
\end{equation}
We can apply Lemma \ref{sti.abs.2}, with $b^{(t)}_j=1$ for any $t=1, \dots, 4$ and $j \in \mathbb{Z}$, to obtain
\begin{equation*}
\sup_{(z,\bar{z}) \in B_{a,q}(\tR)}\left\|X_{G^{\mathcal{NG}}}(z)\right\|_{a,q}\lesssim
h\tR^3,
  \end{equation*}
for both $q=p$ and $q=p-4$.

To estimate the vector field of $G^{R1}$ we use again Lemma \ref{sti.abs.2}, with coefficients
$$
F_{\vj,\vsigma}\vcentcolon= \frac{P^R_{\vj, \vsigma}}{\vsigma\cdot\vlambda}.
$$
Then we proceed exactly in the same way, with $b_j^{(t)}=1$ for any $j \in \mathbb{Z}$ and $t=1, \dots ,4$.

We are almost ready to prove the estimate of $G^R$. We already proved
the estimates for the first two terms of the decomposition \eqref{decomp:G}, so we have to estimate the
vector field of $G^{\tt app}-G^{\tt NLS}$. To this end, we remark that both the Hamiltonians are Gauge invariant, so that their coefficients are different from zero only when \eqref{gauge.1} is satisfied. Then, recalling \eqref{def:lambda2}, we have
\begin{align*}
\gna_{\vj,\vsigma}-G^{\tt NLS}_{\vj,\vsigma}&=
  \left(\frac{1}{\im \vsigma\cdot\vlambda}-
\frac{1}{\im \vsigma\cdot\vlambda\null^{\tt NLS}}
  \right)P_{\vj,\vsigma}^{\tt NLS}
\\
& = \left(  \frac{\sum_{m=1}^4 \sigma_m (\lambda^{\tt
    NLS}_{j_m}-\nu_{j_m})}{\im (\vsigma\cdot\vlambda)(\vsigma\cdot\vlambda\null^{\tt
  NLS})}\right)P_{\vj,\vsigma}^{\tt NLS}.
\end{align*}
Then $G^{\tt app}-G^{\tt NLS}$ is the sum of four terms with the structure \eqref{sti.abs.1},
each  with 
$$
F_{\vjs}= -\frac{P_{\vjs}^{\tt
    NLS}}{\im (\vsigma\cdot\vlambda)(\vsigma\cdot\vlambda\null^{\tt NLS}) }
$$
and the first term has
$$
b^{(1)}_{j_1}=\nu_{j_1}-\lambda_{j_1}^{\tt NLS} \sleq
h|j_1|^4,\qquad   b^{(t)}_{j_t}=1,\quad t=2,3,4.
$$
The other terms are obtained in the same way. The application of
Lemmas \ref{algebra} and \ref{sti.abs.2} gives the thesis and concludes the proof.
\end{proof}

We can now state the following result about the flows of $G$ and $G^{\tt NLS}$.

\begin{lemma}
  \label{flusso.0}
There exist $h_*$,
$\tR_{1*}>0$ such that the following holds. For any $\tR\in (0, \tR_{1*})$ and $h \in (0,h_*)$ the flows $\phi_G^t$ of
$G$ and $\phi_{G^{\tt NLS}}^t$ of
$G^{\tt NLS}$ are real analytic and fulfills
\[
\phi_G^t\colon \mathcal{B}_{a, q,\beta}(\tR) \to  \mathcal{B}_{a,
  q,\beta}(2\tR), \quad \phi_{G^{\tt NLS}}^t\colon \mathcal{B}_{a, q}(\tR) \to  \mathcal{B}_{a,
  q}(2\tR),
\]
for $q=p, p-4$, $\beta\in[0,1]$ and for all $t\in [-1, 1]$. Furthermore, the following holds.
\begin{itemize}
\item[\textnormal{(i)}]  For both $q=p,p-4$ the flows satisfy
 \begin{align}\label{stim:fluXKG}
  &\sup_{(z,\bar z)
    \in B_{a,q,0}(\tR)}\left\|\phi^t_G(z,\bar z)-(z,\bar z)\right\|_{a,q,1}\lesssim
  \tR^3& \qquad &\forall t\in [-1, 1],& \\
\label{fluNLS}
&\sup_{(z,\bar z)
    \in B_{a,q}(\tR)}\left\|\phi^t_{G^{\tt
      NLS}}(z,\bar z)-(z,\bar z)\right\|_{a,q}\lesssim \tR^3&  \qquad &\forall t\in [-1, 1].&
\end{align}
\item[\textnormal{(ii)}] If we define 
\begin{equation}\label{scomp.flow}
\Psi(h; t, \cdot) \vcentcolon= \phi_G^t-\phi_{G^{\tt NLS}} \qquad \forall t\in [-1, 1],
\end{equation} 
the following estimate holds
\begin{equation}
  \label{sti.flu.3}
\sup_{(z,\bar z)\in B_{a,p}(\tR)}\left\|\Psi(h;
 t, z,\bar z)\right\|_{a,p-4}\lesssim
  h{\tR^3}, \qquad \forall t \in [-1,1].
\end{equation}
\end{itemize}
\end{lemma}
\begin{proof}
The proof of \eqref{stim:fluXKG} and \eqref{fluNLS} follows from \eqref{esisto.3}, combined with  \eqref{sti.g0} and \eqref{G.NLS}. The proof of \eqref{sti.flu.3} follows from Lemma \ref{flowmap.perdo}, combined with \eqref{sti.g1}.
\end{proof}

%%%%%%%%%%%%%%%%%%%%%%%%%%%%%%%%%%%%

%%%%%%%%%%%%%%%%%%%%%%%%%
\subsubsection{Proof of Theorem \ref{Teorema0} }\label{fine}

The proof is divided into three parts, corresponding to the three items of the statement.

\textbf{Proof of Item \ref{item:flussi}:} By Lemma \ref{flusso.0}. we define $\cT_0 \vcentcolon = \phi_G^t \big|_{t=1}$. Then from \eqref{h_0}, combined with \eqref{homologicala}, we have 
\[
H \circ \cT_0=\Lambda+\Lambda_+ + \hat{P}+ P_0,
\]
which is exactly \eqref{NFS}, with $\Lambda_+$ and $\hat{P}$ homogeneous of degree four. Similarly we obtain \eqref{NFS:NLS} by defining $\cT_0^{\tt NLS} \vcentcolon = \phi_{G^{\tt NLS}}^t \big|_{t=1}$ introduced in Lemma \ref{flusso.0}. Then, the estimates \eqref{sti.tra.1}, \eqref{res.tr.0} are exactly \eqref{stim:fluXKG}, \eqref{fluNLS} and \eqref{sti.flu.3} for $t=1$, which hold true for any $\tR \in (0,\tR_{1*})$.

\medskip

\textbf{Proof of Item \ref{intem:normale}:}  the expressions of $\Lambda_+$ and $\Lambda_+^{\tt NLS}$ are obtained in Lemma \ref{sol.homo}. Then, the estimates \eqref{res.0.2} 
are obtained from Lemmata \ref{algebra},
  \ref{sti.abs.2}, in a way similar to the
one presented in the proof of Lemma \ref{sol.homo}.
To conclude Item \ref{intem:normale} we need to prove the second of \eqref{stimaLR} for the vector field of
\[
\Lambda_+^R\vcentcolon= \Lambda_+-\Lambda_+^{\tt NLS}.
\]
By the expressions \eqref{L+NLS} and \eqref{res.nij}, we have 
\begin{align}
\label{prima}
\Lambda_+^R&=\frac{1}{2} \sum_{i,j} N_{ij}\bigg(\frac{1}{(1+h\nu_j)(1+h\nu_i)}-1 \bigg) {z_i \bar{z}_i z_j \bar{z}_j}
\\
\label{seconda}
&=-\frac{h}{2}\sum_{\min(i,j) \leq N
}N_{ij}\frac{\nu_i+\nu_j+h\nu_i\nu_j}{(1+h\nu_i)(1+h\nu_j)} {z_i \bar{z}_i z_j \bar{z}_j}
\\
\label{sti.L.1}
&
=-\frac{h}{2}\left[\sum_{j_1...j_4}\frac{N_{j_1j_2}\delta_{j_1j_3}\delta_{j_2j_4}\nu_{j_1}}{(1+h\nu_{j_1})(1+h\nu_{j_2})}z_{j_1}\bar
  z_{j_3}z_{j_2}\bar z_{j_4}\right.
  \\
  \label{sti.L.2}
&  +\sum_{j_1...j_4}\frac{N_{j_1j_2}\delta_{j_1j_3}\delta_{j_2j_4}\nu_{j_2}}{(1+h\nu_{j_1})(1+h\nu_{j_2})}z_{j_1}\bar
  z_{j_3}z_{j_2}\bar z_{j_4}
  \\
  \label{sti.L.3}
& \left. +\sum_{j_1...j_4}\frac{N_{j_1j_2}\delta_{j_1j_3}\delta_{j_2j_4}h\nu_{j_1}\nu_{j_2}}{(1+h\nu_{j_1})(1+h\nu_{j_2})}z_{j_1}\bar
  z_{j_3}z_{j_2}\bar z_{j_4}\right]\,.
\end{align}
Actually in the sums of the lines \eqref{prima}-\eqref{sti.L.3} there are some
further limitations on the range of the indices, but they are
irrelevant for the subsequent estimates.
The term \eqref{sti.L.1} has the form \eqref{sti.abs.1} with
\begin{align}
\nonumber
  &F_{j_1...j_4}=  {h}N_{j_1j_2}\delta_{j_1j_3}\delta_{j_2j_4}\,,&
  \\
\label{stim:intermedia}
&b^{(1)}_{j_1}=  \frac{\nu_{j_1}}{1+h\nu_{j_1}}\leq
\frac{j_1^2}{2},\quad b^{(2)}_{j_2}=\frac{1}{1+h\nu_{j_2}}\leq
1,\quad b^{(3)}_{j_3}=b^{(4)}_{j_4}=1\,,&
\end{align}
thus from Lemmata \ref{algebra} and \ref{sti.abs.2} we have 
\begin{equation}  \label{sti.L.4}
\left\|X_{(\ref{sti.L.1})}(z,\bar z)\right\|_{a,p-2}\sleq
h\left\|(z,\bar z)\right\|_{a,p}^3 \qquad  \forall (z,\bar z) \in \cP^{a,p}.
\end{equation}
The vector field of \eqref{sti.L.2} is estimated exactly in the same
way. Also \eqref{sti.L.3} has the structure \eqref{sti.abs.1}, but the
weights are now
$$
b^{(1)}_{j_1}=\frac{\nu_{j_1}}{1+h\nu_{j_1}}\leq
\frac{j_1^2}{2},\quad
b^{(2)}_{j_2}=\frac{h\nu_{j_2}}{1+h\nu_{j_2}}<1,
$$
so again we get 
\begin{equation*}
\left\|X_{\eqref{sti.L.3}}(z,\bar z)\right\|_{a,p-2}\sleq
h\left\|(z,\bar z)\right\|_{a,p}^3 \qquad  \forall (z,\bar z) \in \cP^{a,p}.
\end{equation*}
Summarizing we get 

\[
\sup_{(z,\bar z) \in B_{a,p}(\tR)}\left\|X_{\Lambda_+^R}(z,\bar z)\right\|_{a,p-2}\lesssim h\tR^3.
\]

\medskip
\textbf{Proof of Item \ref{item:pertu}:} { first, we recall that $\hat{P}$ and $\hat{P}^{\tt NLS}$ are obtained in \eqref{P.capp} and \eqref{P.capp.NLS}. Then the non Gauge invariant part of $\hat{P}$ is defined as
\[
\hat{P}^{\mathcal{NG}}\vcentcolon=\sum_{\substack{(\vj,\vsigma)\in(\cJ^c)^4 \times \{\pm\}^4, \\ \sum_i \sigma_i \neq 0}}P_{\vj, \vsigma}z_{\vj}^{\vsigma}
\]
and thus we can define $\hat{P}^R\vcentcolon=\hat{P}-\hat{P}^{\mathcal{NG}}-\hat{P}^{\tt NLS}$. The estimates \eqref{res.0.1}, the second of \eqref{res.NLS} and the second of \eqref{res.0.3} are obtained from Lemmata \ref{algebra},
  \ref{sti.abs.2}, in a way similar to the
ones of Lemma \ref{sol.homo}. }

\smallskip

We need to prove the estimates of the vector fields of $P_0$ and $P_0^{\tt NLS}$. By \eqref{P0}, \eqref{homologicala} and \eqref{homoNLS} we have
\begin{align}
  \label{P0.1}
  P_0&=\int_0^1s \left\{P,G\right\}\circ \phi_G^s\,ds+\int_0^1(1-s)\left\{\Lambda_+,G\right\}\circ \phi_G^s\,ds+\int_0^1(1-s) \{ \hat P,G \} \circ \phi_G^s\,ds,& \\
  \label{P0.NLS}
 {P_0^{\tt NLS}} &={ \int_0^1 s \left\{ P^{\tt NLS}, G^{\tt NLS} \right\} \circ \phi_{G^{\tt NLS}}^s \, ds 
  + \int_0^1 (1-s) \left\{ \Lambda_+^{\tt NLS}, G^{\tt NLS} \right\} \circ \phi_{G^{\tt NLS}}^s \, ds }
  & \\
\nonumber
& {+ \int_0^1 (1-s) \left\{ \hat P^{\tt NLS}, G^{\tt NLS} \right\} \circ \phi_{G^{\tt NLS}}^s \, ds.}&
\end{align}
Then the first of \eqref{res.0.0} and the first of \eqref{res.NLS} are standard and are omitted (see for instance Section 3. of 
\cite{poschel1996quasi}).

Now, we are going to define the decomposition $P_0=P_0^{\tt NLS}+P_0^{\mathcal{NG}}+P_0^R$
and estimate the different terms. First, we define 
\begin{equation}
\label{P:0NG}
  P_0^{\mathcal{NG}}\vcentcolon= \int_0^1s\left\{P^{\mathcal{NG}},G^{\tt NLS} \right\}\circ \phi_{G^{\tt NLS} }^s\,ds +\int_0^1(1-s)\left\{\hat P^{\mathcal{NG}},G^{\tt NLS} \right\}\circ \phi_{G^{\tt NLS} }^s\,ds,
\end{equation}
so that $\Pi_{\mathcal{G}} P_0^{\mathcal{NG}}=0$ by Lemma \ref{lem:gauge2}. The estimate of $X_{P_0^{\mathcal{NG}}}$ follows from the same computation of those used for $X_{P_0}$ in \cite{poschel1996quasi} (see Section 3), combined with \eqref{stima.2}, the second of \eqref{res.0.1}, \eqref{G.NLS} and \eqref{fluNLS}.

Then we define 
\begin{equation}\label{P0R}
P_0^R\vcentcolon= P_0-P_0^{\mathcal{NG}}-P_0^{\tt NLS},
\end{equation}
 and the rest of the proof is devoted to estimate its vector field, as stated in \eqref{res.0.3}. {According to \eqref{P0.1}, \eqref{P0.NLS}, \eqref{P:0NG} and \eqref{P0R} we need to analyze several
contributions. We do this in detail only for 
\[
\int_0^1s \left\{P,G\right\}\circ \phi_G^s\,ds- \int_0^1s\left\{P^{\mathcal{NG}},G^{\tt NLS} \right\}\circ \phi_{G^{\tt NLS} }^s\,ds - \int_0^1 s \left\{ P^{\tt NLS}, G^{\tt NLS} \right\} \circ \phi_{G^{\tt NLS}}^s \, ds,
\]
the others being essentially equal. Up to an irrelevant factor $s$, this term involves the integral on $[0, 1]$ of the function}
\begin{align}
\label{termine:lungo}
\left\{P,G\right\}\circ \phi^s_G&-\left\{P^{\tt NLS} +P^{\mathcal{NG}},G^{\tt NLS} 
\right\}\circ \phi^s_{G^{\tt NLS} }
\\
&=\label{*.1}
\left\{P^{\tt NLS} +P^{\mathcal{NG}},G^{R}
\right\}\circ \phi^s_{G}
\\
\label{**.1}
&+\left\{P^{\tt NLS} +P^{\mathcal{NG}},G^{\tt NLS} 
\right\}\circ \phi^s_{G}
\\
\label{*.2}
&+\left\{P^{R},G
\right\}\circ \phi^s_{G}
\\
\label{**.2}
&-\left\{P^{\tt NLS} +P^{\mathcal{NG}},G^{\tt NLS} 
\right\}\circ \phi^s_{G^{\tt NLS} }\,.
\end{align}
We start by estimating \eqref{*.1}.

From  \eqref{esisto.4} and \eqref{prop:group} we have
\begin{equation}\label{stimaux}
X_{\left\{P^{\tt NLS} +P^{\mathcal{NG}},G^{R}
\right\}\circ \phi^s_{G}}=\left(d\phi_{G}^{-s} [X_{  P^{\tt NLS}
      +P^{\mathcal{NG}}},X_{ G^R}] \right)\circ\phi^s_G, 
\end{equation}
that we are going to estimate as a map from $B_{a,p}(\tR)$ to $\cP^{a,p-4}$. 

{For $\tR<\tR_{*1}/3$ introduced in Lemma \ref{flusso.0}, we have that $\phi_G^s$ is well defined on $B_{a,q}(3 \tR)$ for any $s \in [-1,1]$ and for both $q=p,p-4$. In particular 
 \begin{equation} 
\label{inclusion:flux}
\begin{aligned}
&\phi^s_G( B_{a,p}(\tR))\subset  B_{a,p}(2\tR)\subset
 B_{a,p-4}(2\tR),&\quad &\forall s \in [-1,1],& \\
&\phi^s_G( B_{a,p-4}(3\tR))\subset  B_{a,p-4}(6\tR),&\quad & \forall s \in [-1,1].&
\end{aligned}
\end{equation}}
Then for any $\tR<\tR_{*1}/3$, recalling the definition of the operator norm \eqref{norma:op}, we estimate \eqref{stimaux} as follows. For any $s \in [-1,1]$ we have
  \begin{align*}
\sup_{(z,\bar z)\in B_{a,p}(\tR)}&\left\|\left(d\phi_G^{-s} [X_{  P^{\tt NLS}
      +P^{\mathcal{NG}}},X_{ G^R}] \right)\left(\phi^s_G(z,\bar
z\right))  \right\|_{a,p-4}
\\
&\leq \sup_{(z,\bar z)\in B_{a,p}(2\tR)}\left\|\left(d\phi_G^{-s} [X_{  P^{\tt NLS}
      +P^{\mathcal{NG}}},X_{ G^R}] \right)(z,\bar
z)  \right\|_{a,p-4} 
\\
&\leq \sup_{(z,\bar z)\in B_{a,p}(2\tR)}\left\|d\phi_G^{-s}(z,\bar z)
\right\|_{\cP^{a,p-4};\cP^{a,p-4}}
\\
&\times \left(\sup_{(z,\bar
  z)\in B_{a,p}(2\tR)}\left\| dX_{  P^{\tt NLS}
      +P^{\mathcal{NG}}}(z,\bar z)
\right\|_{\cP^{a,p-4};\cP^{a,p-4}}\sup_{(z,\bar
  z)\in B_{a,p}(2\tR)}\left\|X_{G^R}(z,\bar z)\right\|_{a,p-4}\right.
\\
&+\left. \sup_{(z,\bar
  z)\in B_{a,p}(2\tR)}\left\| dX_{G^R}(z,\bar z)
\right\|_{\cP^{a,p};\cP^{a,p-4}}\sup_{(z,\bar
  z)\in B_{a,p}(2\tR)}\left\|X_{  P^{\tt NLS}
      +P^{\mathcal{NG}}}(z,\bar z)\right\|_{a,p}  \right)
\\
&\lesssim {\frac{1}{\tR}\sup_{(z,\bar z)\in B_{a,p-4}(3\tR)}\left\|\phi_G^{-s}(z,\bar z)
\right\|_{a,p-4}}
\\
&\times \left({\frac{1}{\tR}\sup_{(z,\bar
  z)\in B_{a,p-4}(3\tR)}\left\| X_{  P^{\tt NLS}
      +P^{\mathcal{NG}}}(z,\bar z)
\right\|_{a,p-4}}\sup_{(z,\bar
  z)\in B_{a,p}(2\tR)}\left\|X_{G^R}(z,\bar z)\right\|_{a,p-4}\right.
\\
&+\left. {\frac{1}{\tR}\sup_{(z,\bar
  z)\in B_{a,p}(3\tR)}\left\| X_{G^R}(z,\bar z)
\right\|_{a,p-4}}\sup_{(z,\bar
  z)\in B_{a,p}(2\tR)}\left\|X_{  P^{\tt NLS}
  +P^{\mathcal{NG}}}(z,\bar z)\right\|_{a,p}  \right) \lesssim h\tR^5,
  \end{align*}

where we estimated the norms of the differential of the flow combining \eqref{inclusion:flux} with \eqref{stim:fluXKG} and Cauchy estimates. The norm of the commutator is estimated recalling \eqref{reg.poisson}, \eqref{stima.2}, \eqref{stima.3} and \eqref{sti.g1}, combined with Cauchy estimates for the operator norms of the differentials.

\smallskip

The term \eqref{*.2} is estimated exactly in the same way. To conclude the estimate of the vector field of \eqref{termine:lungo}, we estimate
now $X_{\eqref{**.1}}-X_{\eqref{**.2}}$. First we analyze
$X_{\eqref{**.1}} $. To this end remark first that for any $s \in [-1,1]$, by \eqref{scomp.flow} and \eqref{prop:group}, we have
$
(\phi_G^s)^{-1}=\phi^{-s}_G=\phi^{-s}_{G^{\tt NLS}}+\Psi(h;-s).
$
Thus denoting
\begin{equation}\label{Z}
\text{Z}\vcentcolon= X_{\left\{P^{\mathcal{NG}}+P^{\tt NLS} ,G^{\tt NLS} \right\}},
\end{equation}
we have (c.f. \eqref{esisto.4})
\begin{align*}
X_{\eqref{**.1}}=\left[\diff (\phi_{G^{\tt NLS} }^s+\Psi(h; s))^{-1}\text{Z}\right]\circ\left(
\phi_{G^{\tt NLS} }^s+\Psi(h; s) \right)&= \left[\diff (\phi_{G^{\tt NLS} }^{-s})\text{Z}\right]\circ\left(
\phi_{G^{\tt NLS} }^s+\Psi(h; s) \right)&
\\
& + \left[\diff \Psi(h; \blu{-}s)\text{Z}\right]\circ\left(
\phi_{G^{\tt NLS} }^s+\Psi(h; s) \right).&
\end{align*}
The second term is easily estimated using \eqref{sti.flu.3} to
estimate the differential of $\Psi$. The difference
between the first one and the vector field  of \eqref{**.2}  is given by
\[
\text{W}\vcentcolon= (\diff \phi_{G^{\tt NLS} }^{-s}\text{Z})\circ (\phi_{G^{\tt NLS} }^{s}+\Psi(h; s))-
  (\diff \phi_{G^{\tt NLS} }^{-s}\text{Z})\circ \phi_{G^{\tt NLS} }^{s}.
\]
{For $\tR<\tR_{*1}/3$ introduced in Lemma \ref{flusso.0}, we have that $\phi_{G^{\tt NLS}}^s$ is well defined on $B_{a,q}(3 \tR)$ for any $s \in [-1,1]$ and for both $q=p,p-4$. In particular 
 \begin{equation} 
\label{inclusion:flux:NLS}
\begin{aligned}
&\phi^s_{G^{\tt NLS}}( B_{a,p}(\tR))\subset  B_{a,p}(2\tR)\subset
 B_{a,p-4}(2\tR),&\quad &\forall s \in [-1,1],& \\
&\phi^s_{G^{\tt NLS}}( B_{a,p-4}(3\tR))\subset  B_{a,p-4}(6\tR),& \quad & \forall s \in [-1,1].&
\end{aligned}
\end{equation}}
Then for any $\tR<\tR_{*1}/3$ and $ s \in [-1,1]$, by Lagrange theorem we have
\begin{align*}
\sup_{(z,\bar z)\in B_{a,p}(\tR)} \|\text{W}(z,\bar z)\|_{a,p-4} 
&\leq
\sup_{(z,\bar z)\in B_{a,p}(2\tR)}  
\oN{
\diff( \diff \phi_{G^{\tt NLS}}^{-s}\text{Z})(z,\bar z)
}_{\cP^{a,p-4};\cP^{a,p-4}}
\sup_{(z, \bar{z})\in B_{a,p}(\tR)} \|\Psi(h; s)\|_{a,p-4} 
\\
&\leq 
\Bigg(
\sup_{(z,\bar z)\in B_{a,p}(2\tR)}  
\oN{
\diff^2 \phi_{G^{\tt NLS}}^{-s}(z,\bar z)
}_{\big(\cP^{a,p-4}\big)^2;\cP^{a,p-4}}
\sup_{(z,\bar{z}) \in B_{a,p}(2\tR)} \| \text{Z}(z,\bar{z})\|  
\\
&\quad + 
\sup_{(z,\bar z)\in B_{a,p}(2\tR)} 
\oN{
\diff \phi_{G^{\tt NLS}}^{-s}(z,\bar z)
}_{\cP^{a,p-4};\cP^{a,p-4}} 
\oN{
\diff \text{Z} (z,\bar z)
}_{\cP^{a,p-4};\cP^{a,p-4}} 
\Bigg) h \tR^3 \\
&\lesssim h\tR^7,
\end{align*}
where we used \eqref{inclusion:flux:NLS}, \eqref{sti.flu.3}, \eqref{fluNLS}, and Cauchy estimates to bound the norms of the flows and their differentials. To bound the norm of $Z$ defined in \eqref{Z}, we use the fact that it is the commutator of two vector fields, and the estimates \eqref{stima.2}, \eqref{stima.3} and \eqref{GNLS}. 
This concludes the estimate for the vector field of \eqref{termine:lungo}.

{The remaining terms in \eqref{P0R} can be estimated similarly and are omitted. By choosing $\tR_*=\tR_{*1}/3$ introduced in Lemma \ref{flusso.0}, this completes the proof of the first inequality in \eqref{res.0.3}, and thus of the theorem.}

%%%%%%%%%% fin qui

\subsection{Expansion on the unperturbed invariant torus}

In a small neighborhood of the origin, we consider now an unperturbed torus, invariant by the flow of $\Lambda+\Lambda_+ + \hat{P}$, and we study its persistence in the full system.  Following \cite{poschel1996quasi} and \cite{kuksin1996invariant}, it is convenient to introduce the rescaled parameter
\begin{equation}
\label{def.di.r}
  r = \tR^{3/2}.
\end{equation}
{Then, for any $\xi = (\xi_j)_{j \in \cS} \in \Xi_0$, where the set $\Xi_0$ is defined in \eqref{xi0:def} and it can be written in terms of $r$ as
\begin{equation}\label{Xi0}
  \Xi_0 
          = \left[ \frac{1}{2} r^{4/3}, \frac{3}{2} r^{4/3} \right]^N,
\end{equation}}
we consider the torus
\begin{align*}
\T^N(\xi)\vcentcolon= \{(z, \bar{z}) \in \cP^{a,p}\,|\,
\overline{z_j}=\bar{z}_{j}\,\ \text{and}\  \,
|z_j|^2=\xi_j, \, \, j \in \cS , \quad |z_j|^2&=0,\quad j \not\in \cS \}\ 
\end{align*}
and we Taylor expand the Hamiltonian around it. To this end we
introduce action-angle coordinates $(x, y)$ for the
modes indexed on $\cS$. We consider the change of coordinates 
\begin{align}
  \label{def.t1}
(z,\bar z)=\cT_1\big(\left\{x_j\right\}_{j \in \cS},\left\{y_j\right\}_{j \in \cS} ,(\{
z_j\}_{j \in \cS^c},\{{\bar z}_{j}\}_{j \in \cS^c})\big)
\end{align}
  with 
  \begin{align*}
   & z_{j}=\sqrt{\xi_j+y_j}e^{\im x_j},   \, \, \,  \bar{z}_{j}=\sqrt{\xi_j+y_j}e^{-\im x_j},\quad j \in \cS,&
  \end{align*}
{while in the directions of $\cS^c$ we keep the original variable $z_j, \bar{z}_j$.}

Taylor expanding the Hamiltonian \eqref{NFS} at $(y,z,\bar z)=(0,0,0)$ we have
\begin{equation}\label{H}
H_0(\xi; x,y,z,\bar{z})=N_0+P_0,
\end{equation}
with
\begin{equation*}
N_0=N_0(\xi; x, y, z, \bar{z})\vcentcolon= \jbs{\omega_0(h,\xi),y}+
\jbs{\Omega_0(h,\xi),z\bar z},
\end{equation*}
where $\omega_0=(\omega_{0 \, j} )_{j \in \cS}$ and $\Omega_0=\{\Omega_{0 \, j}\}_{j \in \cS^c}$ are given below in \eqref{Oomega},
$z\bar{z}\vcentcolon= \{z_j \bar{z}_{j}\}_{j \in \cS^c}$, and  $\langle
w,w'\rangle\vcentcolon= \sum_{j }w_jw'_j$, with the sum that runs either on $\cS$ or on $\cS^c$, and $P_0 \vcentcolon = H_0-N_0$, while its
relevant properties will be specified later.  Similarly, for the NLS, we write
\begin{equation}\label{HNLS}
H_0^{\tt NLS}(\xi; x,y,z,\bar{z})=N_0^{\tt NLS}+P_0^{\tt NLS},
\end{equation}
with
\begin{equation}
  \label{N.NLS.1}
N_0^{\tt NLS}=N_0^{\tt NLS}(\xi; x, y, z, \bar{z})\vcentcolon= \jbs{\omega_0^{\tt NLS}(\xi),y}+
\jbs{\Omega_0^{\tt NLS}(\xi),z\bar z}\ 
\end{equation}
where $\omega^{\tt NLS}_0, \Omega_0^{\tt NLS}$ are given below in
\eqref{OomegaNLS} and the properties of $P^{\tt NLS}_0$ will be
specified later.  For the definition of the frequencies we have to introduce
the matrices $\matA=(\matA_{ij})_{i \in \cS}^{j \in \cS}$
 and  $\matB=(\matB_{ij})_{i \in \cS^c}^{j \in \cS}$
\begin{align}
\label{def:AKG}
&\matA_{ij}(h)\vcentcolon= \frac{N_{ij}}{(1+h\nu_j(h))(1+h\nu_i(h))},\ & & i,j \in \cS&
\\
\label{def:BKG}
&\matB_{ij}(h)\vcentcolon= \frac{N_{ij}}{(1+h\nu_j(h))(1+h\nu_i(h))},\qquad &
&i \in \cS^c, \, \, j \in \cS,&
\end{align}
with $N_{i, j}$ given in \eqref{res.nij}.

{It is useful to introduce a notation which distinguishes
between the frequencies of the modes in $\cJ$ and the remaining ones.}
The frequencies are given by

\begin{equation}\label{Oomega}
\omega_0(h,\xi)\vcentcolon= \lambda(h)+\matA(h)\xi, \qquad \Omega_0(h,\xi)\vcentcolon= \varLambda(h)+\matB(h)\xi,
\end{equation}
where 
\begin{equation}\label{lambda-Lambda}
\lambda=(\lambda_j)_{j \in \cS}, \qquad \varLambda=(\lambda_j)_{j\in \cS^c}, \qquad \lambda_j\vcentcolon= c\sqrt{c^2+j^2}.
\end{equation}
Similarly we define 
\begin{equation}\label{matNLS}
\matA^{\tt NLS}=(\matA^{\tt NLS}_{ij})_{i \in \cS}^{j \in \cS}\vcentcolon= (N_{ij})_{i \in \cS}^{j \in \cS},
\quad \matB^{\tt NLS}=(\matB^{\tt NLS}_{ij})_{j \in \cS}^{i \in \cS^c}\vcentcolon= (N_{{i}j})_{j \in \cS}^{i \in \cS^c},
\end{equation}
and
\begin{equation}\label{OomegaNLS}
\omega^{\tt NLS}_0(\xi)\vcentcolon= \lambda^{\tt
  NLS}+\matA^{\tt NLS}\xi, \qquad
\Omega_0^{\tt NLS}(\xi)\vcentcolon= \varLambda^{\tt
  NLS}+\matB^{\tt NLS}\xi,
\end{equation}
where  
\begin{equation}\label{linFreNLS}
\lambda^{\tt NLS}=(\lambda^{\tt NLS}_j)_{j \in \cS} \qquad \varLambda^{\tt NLS}=(\lambda^{\tt NLS}_j)_{j \in \cS^c}, \qquad \lambda_j^{\tt NLS}=\frac{j^2}{2}.
\end{equation}
Finally we define 
\begin{equation}\label{OomegaR}
\omega^{R}_0(h,\xi)\vcentcolon= \omega_0(h,\xi)-h^{-1}-\omega_0^{\tt NLS}, \qquad
\Omega_0^{R}(h,\xi)\vcentcolon= \Omega_0(h,\xi)-h^{-1}-\Omega_0^{\tt NLS}. 
\end{equation}
Whenever it is not useful, we will not specify the dependence on the parameters $\xi$ or $h$ of the frequencies.
\begin{remark}
We point out that $\lim_{h\to0^+}\matA_{i j}(h)=A_{i j}^{\tt NLS}$ and $\lim_{h\to0^+}\matB_{i j}(h)=B_{i j}^{\tt NLS}$.
\end{remark}

\smallskip

In terms of the action-angle variables we redefine the phase
  space as 
  \begin{equation}\label{phase:sptrasf}
\widetilde{\cP}^{a,p,\beta}\vcentcolon= \C^N/ (2\pi\Z)^N\times\C^N\times\ell^{a,p,\beta}\times\ell^{a,p,\beta}\ni(x,y,z,\bar
z)\ 
\end{equation}
and, if there is no risk of confusion we will omit the tilde from its symbol.
We define the complex
neighborhoods of the torus $\T^N(\xi)$ as 
\begin{equation}\label{Dq}
D_{q,\beta}(s,r)\vcentcolon= \{ |\Im x | < s, |y| < r^2, \|(z,\bar z)\|_{a,q,\beta} < r\}\subset
\widetilde{\mathcal{P}}^{a,q,\beta} \vcentcolon=  \C^N/ (2\pi\Z)^N \times \R^N \times \ell^{a,q,\beta} \times
\ell^{a,q,\beta}.
\end{equation}

 Again we will consider $q=p$ and $q=p-4$, $\beta \in [0,1]$. {When $\beta=0$, sometimes we will write $D_q(s,r)=D_{q,\beta}(s,r)$.}
\begin{remark}
  \label{cambio norme}
If $s\leq s_0$, and  $\tR<1/2$  and $\xi\in\Xi_0$, then, for
  $(x,y,z,\bar z)\in D_{q,\beta}(s, r)$, one has
  $$
\left\|\cT_1(x,y,z,\bar z)\right\|_{a,q,\beta}\leq C\tR,
$$
with a constant $C$ which in particular depends on $s_0$. To be
more concrete {\bf we set  $s_0=2$.}

This implies that the normal form transformation given in Theorem
 \ref{Teorema0} is well defined on the image of $D_{q,\beta}(s,r)$, provided $\tR$ is small enough. 
\end{remark}

For an element $W=(X,Y,U,V) \in \mathcal{P}^{a,q,\beta}$ we define the norm 
\begin{equation}\label{Norma1}
\|W\|_{r,q,\beta}\vcentcolon=  |X|+\frac{1}{r^2}|Y|+\frac{1}{r} \big(
\|U\|_{a,q,\beta}+\|V\|_{a,q,\beta}\big) ,
\end{equation}
and sometimes we will write $\|W\|_{r,q,0}=\|W\|_{r,q}$.

\smallskip

For a vector field $X=X(\xi)$ parameterized by $\xi \in \Xi$, where $\Xi$ is some subset of $ \Xi_0$, we introduce the norms
\begin{align}
\label{norma:normcv}
&\bn{X}_{r,q,\beta}^{D_q(s,r)}\vcentcolon= \sup_{\Xi \times D_q(s,r)}\|X\|_{r,q,\beta},& \\
\nonumber
&\bn{X}_{r,q,\beta}^{\mL, D_q(s,r)}\vcentcolon= \sup_{\xi \neq \eta}\sup_{w \in D_q(s,r)}
  \frac{\|\Delta_{\xi, \eta}X(w) \|_{r,q,\beta}}{|\xi-\eta|},
  \qquad  \Delta_{\xi, \eta} X(w)\vcentcolon= X(\xi,w)-X(\eta,w),& \\
\label{norma:lipcv}
&\bn{X}^{\mathcal{L}(\gamma),D_q(s,r)}_{r,q,\beta}\vcentcolon= 
\bn{X}^{D_q(s,r)}_{r,q,\beta}+\gamma \bn{X}^{\mL,
  D_q(s,r)}_{r,q,\beta}, \quad \textnormal{ for some} \, \, \gamma>0.&
\end{align}

{We introduce the space $\ell^{\infty,\beta}$ of sequences $\{ \ttw_j\}_{j \in \cS^c}$
satisfying 
\begin{equation}\label{infbeta}
\left\|\ttw\right\|_{\infty,\beta}\vcentcolon= \sup_{j \in \cS^c} |{\ttw}_j|\tw_j^\beta<\infty,
  \end{equation}}
with $\tw_j$ defined in \eqref{ellappiu}. Sometimes the dependency on $\beta$ is omitted when $\beta=0$. 
We will consider also functions $f: \Xi \to \ell^{\infty, \beta}$, where $\Xi$ is some subset of $ \Xi_0$, and the following norms
\begin{align}
\label{norma:Omega}
 &
\left|f\right|_{\infty,\beta}\vcentcolon= \sup_{\xi \in \Xi}\|f(\xi)\|_{\infty,\beta},
&
\\
\label{norma:OmegaL}
&\left|f\right|^{\mL}_{\infty,\beta}\vcentcolon= \sup_{\xi\not=\eta}
\frac{\left\|\Delta_{\xi,\eta}f\right\|_{\infty,\beta}}{\left|\xi-\eta\right|}, \quad |f|_{\infty,\beta}^{\mL(\gamma)}\vcentcolon= \left|f\right|_{\infty,\beta}+\gamma\left|f\right|_{\infty,\beta}^\mL \qquad \gamma>0.&
\end{align}
A similar notation is introduced for finite sequences. 

\medskip

Finally, we need to introduce the set
\begin{equation}\label{Z2}
 \mathcal{Z}_2 \vcentcolon= \{(k,\ell) \in  \Z^{\cS} \times \Z^{\cS^c} \mid {|k|_1+|\ell|_1 \neq 0}, \, \, 0 \leq |\ell|_1 \leq 2 \}.
\end{equation}
\begin{remark}
{For any $k \in \mathbb{Z}^\cS, m \in \mathbb{N}^\cS$ and $\mq, \bar{\mq} \in \mathbb{N}^{\cS^c}$, the Fourier Taylor monomial
\begin{equation}\label{FTmon}
e^{ \im \jbs{k,x}}y^m z^\mq \bar{z}^{\bar{\mq}} \vcentcolon = e^{ \im \jbs{k,x}} \bigg( \prod_{j \in \cS} y_j^{m_j}\bigg)\bigg( \prod_{j \in \cS^c} z_j^{\mq_j}\bigg)\bigg( \prod_{j \in \cS^c} \bar{z}_j^{\bar{\mq}_j}\bigg)
\end{equation}}
is
\begin{align}
\label{cond:Gauge:kq}
&\text{Gauge\ invariant}& \qquad &\textnormal{if}& \, \, \, &\sum_{j
  \in \cS}k_j +\sum_{j \in \cS^c} (\mq_j-\bar{\mq}_{j})= 0,&
\\
\label{cond:Moment:kq}
&\text{translation\ invariant}& \qquad &\textnormal{if}& \, \, \, &\sum_{j
  \in \cS}jk_j +\sum_{j \in \cS^c}j (\mq_j-\bar{\mq}_{j})= 0.&
\end{align}
\end{remark}
We define the \emph{translation invariant set of indices} $\mZ_\mM
\subseteq \mZ_2$ and the \emph{translation-Gauge invariant set of indices} $\mZ_\mG \subseteq \mZ_\mM$ as 
\begin{align}
\label{Indici:M}
\mZ_\mM& \vcentcolon= \{(k,\ell) \in \mZ_2 \, | \, \sum_{j \in \cS} j k_j +\sum_{j \notin \cS^c} j \ell_j=0\},\\
\label{Indici:G}  \mZ_\mG& \vcentcolon= \{(k,\ell) \in \mZ_\mM \, | \, \cL =0\}, \quad \cL \vcentcolon = \sum_{j \in \cS} k_j +\sum_{j \notin \cS^c} \ell_j.
\end{align}

The following remark will play a fundamental role for the measure
estimates.

  \begin{remark}
    \label{limiti.momento}
    Let $(k,\ell)\in\cZ_{\cM}$ and denote $J \vcentcolon= \max\{|j| \mid j \in \cS \}$ then for fixed $k$ the following holds
    \begin{itemize}
    \item[(i)] if $|\ell|_1=1$, and thus $\ell= \pm \delta_a$ (c.f. \eqref{Kronecker}) for some
      $a\in\cJ^c$, then \eqref{cond:Moment:kq} is fulfilled for at most two values of $a$ satisfying  $|a|\leq J |k|_1$ ;
\item[(ii)] if $|\ell|_1=2$ and $\ell=\pm(\delta_a+\delta_b)$ for some
  $a,b\in\cJ^c$, then \eqref{cond:Moment:kq} implies
      $|a+b|\leq J |k|_1
      $;
\item[(iii)] if $|\ell|_1=2$ and $\ell=\pm(\delta_a-\delta_b)$ for some
  $a,b\in\cJ^c$, then \eqref{cond:Moment:kq} implies
      $| |a|-|b|| \leq |a-b| \leq J |k|_1$.
   \end{itemize}
   \end{remark}

\begin{theorem}\label{Newstime}
Consider the Hamiltonians \eqref{H} and \eqref{HNLS}, there exist $h_* >0$ and $r_*>0$
    such that for $r_0 \in (0,r_*)$ and $h \in (0,h_*)$ the following holds.
\begin{itemize}
\item The frequency maps $\omega_0(h):\Xi_0\to \R^N$,
  $\omega_0^{\tt NLS}:\Xi_0\to \R^N$  in \eqref{Oomega}, \eqref{OomegaNLS}
are two Lipeomorphisms, with Lipschitz constants uniform in $h$. Moreover, the following inequalities hold true
\begin{align}
  \label{grande0}
|\matA k+\matB^t\ell |_1 &{\gtrsim} |k|_1& \qquad &\forall (k, \ell)\in {\mathcal{Z}_{\cM}}, \, \, \forall \, h \in (0,h_*)& \\
\label{grand0:NLS}
|\matA^{\tt NLS}k + (\matB^{\tt NLS})^t\ell|_1 &{\gtrsim}  |k|_1& \qquad  &\forall (k, \ell)\in {\mathcal{Z}_{\cG}}.&
\end{align}

\smallskip

For $i,j \in \mathbb{Z}$ such that $c^3<|i|<|j|$, we have
\begin{equation}\label{asymp}
\bigg|\frac{ \Omega_{0j}- \Omega_{0i}}{ c(|j|-|i|)}
-1\bigg|_\infty\sleq \frac{1}{\tw_i^2}.
\end{equation}
Finally,
\begin{equation}\label{stimeom0}
|\omega_0^R(h)|_{\infty} \lesssim
h, \qquad
\sup_{j\in \cS^c}\left|\Omega_{0j}^R(h)\right|_\infty j^{-4} \lesssim h.
\end{equation}

\item The
  perturbation $P_0$ is real analytic on $D_q(2,r_0)$, Lipschitz in
  the parameter $\xi$ and it is translational invariant. Its vector field
  $X_{P_0}$ satisfies 
\begin{equation}\label{Stima0}
 \epsilon_1\vcentcolon= \bn{X_{P_0}}_{r_0,p,1}^{\mL(r_0^{4/3}),D_p(2,r_0)}+\bn{X_{P_0}}_{r_0,p-4,1}^{\mL(r_0^{4/3}),D_{p-4}(2,r_0)}\lesssim
 r_0^2 .
\end{equation}
\item {By defining $P_0^{\mathcal{NG}} \vcentcolon=\Pi_{\mathcal{NG}}P_0$ and 
\begin{equation*}
P_0^R \vcentcolon=P_0-P_0^{\tt NLS}- P_0^{\mathcal{NG}}, 
\end{equation*}}
all the Hamiltonians and their vector fields are real analytic in $D_q(2,r_0)$ and Lipschitz in the
parameter $\xi$ and are translational invariant. $P_0^{\tt NLS}$ {and $P^R$ are}
also Gauge invariant. We have the following estimates
\begin{align}
\label{1Stima}
& \epsilon_2\vcentcolon= {\bn{X_{P_0^{\tt NLS}}}_{r_0,p}^{\mL(r_0^{4/3}),D_p(2,r_0)}+\bn{X_{P_0^{\tt NLS}}}_{r_0,p-4}^{\mL(r_0^{4/3}),D_{p-4}(2,r_0)}}\lesssim r_0^2,& \\
\label{2Stima}
&\epsilon_3\vcentcolon=h^{-1}\bn{X_{P_0^{R}}}_{r_0,p-4}^{D_p(2,r_0)}\lesssim r_0^2 ,& \\
\label{3Stima}
&\epsilon_4\vcentcolon= {\bn{X_{P_0^{\ng}}}_{r_0,p}^{\mL(r^{4/3}), D_p(2,r_0)}+\bn{X_{P_0^{\ng}}}_{r_0,p-4}^{\mL(r^{4/3}), D_{p-4}(2,r_0)}}\lesssim r_0^2.& 
\end{align}
\end{itemize}
\end{theorem}
\begin{proof}
The proof of the first item is postponed in Appendix
\ref{Smalldivisors}. The estimates
\eqref{Stima0}, \eqref{1Stima}, \eqref{2Stima}, \eqref{3Stima} follow
from Theorem \ref{Teorema0}(for more details see \cite{kuksin1996invariant} Section 7).
\end{proof}
 We define 
\begin{equation}\label{def:e}
\e_0\vcentcolon= \max\{\epsilon_1, \epsilon_2, \epsilon_3, \epsilon_4 \}.
\end{equation}

%%%%%%%%%%fin qui

\section{The linearized equation}

The standard KAM scheme uses an infinite sequence of coordinate changes. Each change is generated by the time-$1$ flow of a Hamiltonian $G$ obtained as a solution of a cohomological equation.

Here we formulate and solve an abstract version of the cohomological equation, that we are going to use in the KAM iteration. We first need to introduce some notation.

{For any $h \in (0,1)$}, on $\Xi_0$ (c.f. \eqref{Xi0}) we consider frequency
vectors of the form
\begin{equation}  \label{omegaprima}
\begin{aligned}
& \omega(h;\xi)\vcentcolon= \lambda+\matA(h)\xi+\delta(h; \xi),\\
& \Omega(h; \xi)\vcentcolon= \varLambda+\matB(h)\xi+\Delta(h; \xi),
\end{aligned}
\end{equation}
{with the matrices $\matA$ and $\matB$ given in \eqref{def:AKG} and
\eqref{def:BKG}, the linear frequency $\lambda, \varLambda$ defined in \eqref{lambda-Lambda}, and for any $h \in [0,1)$
\[
\delta(h; \cdot) :\Xi_0\to\mathbb{R}^N,\qquad \Delta(h; \cdot) :\Xi_0\to \ell^{\infty,1-\theta},\quad \theta\in[0,1],
\]
are Lipschitz functions. We recall that $\ell^{\infty, 1-\theta}$ is defined by \eqref{infbeta}. We also assume that for any $h \in [0,1)$ the functions $\delta(h; \cdot), \Delta(h; \cdot)$ satisfy the bounds
\begin{equation}
  \label{lipfre}
|\delta(h; \cdot)|_{\infty}^{\mL(\alpha)}\lesssim \e,\qquad |\Delta(h; \cdot)|_{\infty,1-\theta}^{\mL(\alpha)}\lesssim \e,
\end{equation}
with the norm defined in \eqref{norma:OmegaL}, and $\e,\alpha, \e \alpha>0$ sufficiently small.} The results of this section hold uniformly on $\theta \in [0,1]$. 

Whenever it does not create
confusion, we will omit to specify the dependence on the parameters $(\xi,h)$ of
the frequencies. We emphasize that the frequencies will always be considered as
  maps defined on the whole $\Xi_0$, even if in the KAM iteration they will be constructed
  as functions defined only on a Cantor subset of $\Xi_0$. The idea is to use the
  extension Lemma \ref{estensione} which allows to extend them on the
  whole  $\Xi_0$.

\smallskip

We also define $\delta^{\tt NLS}(\xi)\vcentcolon= \delta(0;\xi)$ and $\Delta^{\tt NLS}(\xi)\vcentcolon= \Delta(0;\xi)$ and we set
  \begin{equation}
\label{omegaprima:NLS}
\begin{aligned}
&\omega^{\tt NLS}(\xi)\vcentcolon= \lambda^{\tt NLS} +\matA^{\tt NLS}\xi+\delta^{\tt
  NLS}(\xi),&
\\
&
\Omega^{\tt NLS}(\xi)\vcentcolon= \varLambda^{\tt NLS} +\matB^{\tt NLS}\xi+\Delta^{\tt
  NLS}(\xi),&
  \end{aligned}
\end{equation}
where $\lambda^{\tt NLS}, \varLambda^{\tt NLS}, \matA^{\tt NLS}$ and $ \matB^{\tt NLS}$ defined in \eqref{matNLS}, \eqref{linFreNLS}. 

\smallskip

We also assume that 

\begin{equation}\label{stimadifdelta}
|\delta(h)-\delta^{\tt NLS}|_{\infty} \sleq h, \qquad |\Delta_j(h)-\Delta_j^{\tt NLS}|_{\infty} \sleq h j^4, \qquad \forall \, j \in \cS^c.
\end{equation}
{with the norms defined in \eqref{norma:Omega} and we denote $\delta(h) = \delta(h; \cdot)$ and $\Delta(h) = \Delta(h;\cdot)$. Then we define 
\begin{equation}\label{OomegaR:hom}
\omega^R(h) \vcentcolon = \omega_0^R(h)+ \delta(h)-\delta^{\tt NLS}, \quad \Omega^R(h) \vcentcolon = \Omega_0^R(h)+ \Delta(h)-\Delta^{\tt NLS},
\end{equation}
with $\omega_0^R(h), \Omega_0^R(h)$ defined in \eqref{OomegaR}.
}

For this section, {we denote
\begin{equation}
  \label{Nqui}
N\vcentcolon= \langle\omega(h),y\rangle+\langle\Omega(h),z\bar z \rangle, \qquad N^{\tt NLS}\vcentcolon= \langle\omega^{\tt NLS},y\rangle+\langle\Omega^{\tt NLS},z\bar z \rangle. 
\end{equation}}

Let $F$ be an analytic and \emph{translation invariant} Hamiltonian, according to Definition \ref{def:HamGT}. Assume that it admits a decomposition
\[
F=F^{\tt NLS}+F^\ng+F^R,
\]
where each term is still analytic and \emph{translation invariant}. Furthermore, let $F^{\tt NLS}$ be \emph{Gauge invariant according} to Definition \ref{def:HamGT}, and independent of $c$, while $\Pi_{\mathcal{G}} F^{\ng}=0$ (c.f. Definition \eqref{gauge}). Finally, the vector fields of the different terms are assumed to be analytic. {In particular, for  $\theta_0 \in [0,1)$ and some positive constants $s$, $r$, $\alpha$,  we assume that there exist $\te_1, \te_2, \te_3, \te_4>0$, uniform in $h$, such that
\begin{equation}\label{bound:pertub}
\begin{aligned}
&\bn{X_F}^{\mL({\alpha}),D_p(s,r)}_{r,p,1-\theta_0}+\bn{X_F}^{\mL({\alpha}),D_{p-4}(s,r)}_{r,p-4,1-\theta_0}< \te_1,& \\   
&\bn{X_{F^{\ng}}}^{\mL({\alpha}),D_p(s,r)}_{r,p} +\bn{X_{F^{\ng}}}^{\mL({\alpha}),D_{p-4}(s,r)}_{r,{p-4}} < \te_2, & \\
&\bn{X_{F^{\tt NLS}}}^{\mL({\alpha}),D_p(s,r)}_{r,p}+\bn{X_{F^{\tt NLS}}}^{\mL({\alpha}),D_{p-4}(s,r)}_{r,p-4}< \te_3,& 
\\
&\bn{X_{F^R}}^{D_p(s,r)}_{r,p-4}< h \te_4,&
\end{aligned}
\end{equation}
 with the norms defined in \eqref{norma:normcv} and \eqref{norma:lipcv}.}

\smallskip

Recalling the notation \eqref{FTmon}, we consider the Fourier Taylor expansion 
\begin{equation}\label{PolR}
F= \sum_{\substack{k \in \mathbb{Z}^{\cJ} \, m \in \mathbb{N}^{\cJ}, \\ \mq, \bar \mq \in \mathbb{N}^{\cJ^c}}} F_{km \mq\bar{\mq}} e^{\iu \langle k,x \rangle} y^m
z^\mq \bar{z}^{\bar{\mq}},
\end{equation}
where we did not specify the limitation \eqref{cond:Moment:kq} on the range of the indices due to the translation invariance.

We define the average $[F]$ and the truncation of order two $\{ F \}$ of the Hamiltonian $F$ by
\begin{align}
\label{media}
 & [F]\vcentcolon=  \sum_{{|m|+|\mq| =1}}  F_{0m \mq \mq }  y^m z^\mq
  \bar{z}^\mq,&
  \\
\label{taglio}
&\{F\}\vcentcolon=  \sum_{{k \in \Z^\cJ}}
\sum_{2|m|+|\mq+\bar{\mq}| \leq 2} F_{km \mq \bar{\mq}}
e^{\iu \langle k,x \rangle} y^m z^\mq \bar{z}^{\bar{\mq}}.&
\end{align}
\begin{remark}
Exploiting the analyticity of $F$ and its components, one proves
exactly as in \cite{poschel1996kam} (see equation (7)), the following estimates
\begin{align}
\label{restoF}
&\bn{X_{F-\{F\}}}^{\mL({\alpha}), D_q(s,4\eta r)}_{\eta r,q,1-\theta_0} \lesssim \eta \bn{X_{F}}^{\mL({\alpha}), D_q(s,r)}_{r,q,1-\theta_0},& \\
\nonumber
&\bn{X_{F^\ng-\{F^\ng\}}}^{\mL({\alpha}), D_q(s,4\eta r)}_{\eta r,q} \lesssim \eta \bn{X_{F^{\ng}}}^{\mL({\alpha}), D_q(s,r)}_{r,q},& \\
\nonumber
&\bn{X_{F^{\tt NLS}-\{F^{\tt NLS}\}}}^{\mL({\alpha}), D_q(s,4\eta r)}_{\eta r,q} \lesssim \eta \bn{X_{F^{\tt NLS}}}^{\mL({\alpha}), D_q(s,r)}_{r,q},& \\
\nonumber
&\bn{X_{F^R-\{F^R\}}}^{\mL({\alpha}), D_p(s,4\eta r)}_{\eta r,p-4} \lesssim \eta \bn{X_{F^R}}^{\mL({\alpha}), D_p(s,r)}_{r,p-4},&
\end{align}
{for $q=p$ and $q=p-4$, $\theta_0 \in [0,1)$ }and $0<\eta< \frac{1}{8}$.
\end{remark}

\smallskip

We are now going to state and prove the result about the solution $G$ of the cohomological equation
\begin{equation}\label{homological}
\{N,G\}+\{ F\} =[F],
\end{equation}
with $N$ given by \eqref{Nqui}. 
We will also compare $G$ to the solution $G^{\tt NLS}$ of the cohomological equation 
\begin{equation}\label{HomoNLS}
\{N^{\tt NLS},G^{\tt
  NLS}\}+\{F^{\tt NLS}\} =[F^{\tt NLS}],
\end{equation}
with $N^{\tt NLS}$ given by \eqref{N.NLS.1}.

\begin{lemma}
  \label{solhomo}
  Assume that there exist $h_*,\e_*,r_*, \alpha_*, \kappa_*>0$ such that for any $h \in (0,h_*)$,  $\e \in
(0,\e_*)$, $r\in(0,r_*)$, $\alpha \in (0,\alpha_*)$, and for any fixed parameters {$\tau\geq 1$}, $\theta_0 \in [0,1)$ and $\theta_1 \in [0,\tfrac{1-\theta_0}{2})$ the following holds.
\begin{enumerate}[label=(\alph*)] 
\item \label{ip:nongauge} {For any $(k, \ell ) \in \mZ_\mM$ such that $\cL \neq 0$ (c.f. \eqref{Indici:G}), $|k|_1 \leq \kappa_* c^{1/2}$ and $\supp (\ell) \subset [-c/2,c/2]$, we have the following bound
\begin{equation}\label{bound:NGhom:0}
\left|\langle\omega,k\rangle+ \jbs{\Omega, \ell} \right|\gtrsim \alpha c^2\frac{ \jbs{\ell}_1}{\jbs{k}}, \qquad \forall \xi \in \Xi_0.
\end{equation}}
\item There
exists a subset $\Xi\subseteq \Xi_0$ such that for any
$\xi\in\Xi$ the following inequalities are fulfilled
\begin{align}
  \label{dioph}
&\left|\langle\omega,k\rangle+\langle\Omega, \ell \rangle\right|\gtrsim\frac{\alpha\langle
  \ell \rangle_1}{{\langle
  k\rangle^{\tau+1}}(\min_{i\in \textnormal{supp}(l)}\tw_i)^{\theta_1}},&\quad &\forall
(k, \ell)\in \mZ_\mM, \, \, h \in (0,h_*)&\\
  \label{dioph:NLS}
&\left|\langle\omega^{\tt NLS},k\rangle+\langle\Omega^{\tt NLS}, \ell \rangle\right|\gtrsim \frac{\alpha\langle
   \ell \rangle_1}{{\langle
  k\rangle^{\tau+1}}},&\quad &\forall
(k, \ell)\in {\mZ_\mG},&
\end{align}
with $\jbs{\ell}_1$ defined in \eqref{jb}.
\end{enumerate}
Then the cohomological equation \eqref{homological} has a solution of the form 
$G=G^{\tt NLS}+G^R$, defined for all $\xi\in \Xi$, where $G^{\tt NLS}$ is a solution of \eqref{HomoNLS}
and the following estimates hold on $\Xi$ {
\begin{align}
  \label{GG}
&\bn{X_G}^{\mL({\alpha}), D_p(s-\sigma,r)}_{r,p,1-\theta_0-2\theta_1}+\bn{X_G}^{\mL({\alpha}), D_{p-4}(s-\sigma,r)}_{r,p-4,1-\theta_0-2\theta_1}
  \lesssim \frac{1
  }{\sigma^\mu\alpha} \te_1,
\\
  \label{GNLS}
&\bn{X_{G^{\tt NLS}}}^{\mL({\alpha}), D_p(s-\sigma,r)}_{r,p}+\bn{X_{G^{\tt NLS}}}^{\mL({\alpha}), D_{p-4}(s-\sigma,r)}_{r,p-4}
  \lesssim \frac{1
  }{\sigma^\mu\alpha} \te_3 ,
  \\
\label{4.13a}
&\bn{X_{G^R}}^{ D_p(s-\sigma,r)}_{r,p-4} \lesssim \frac{h}{\sigma^\mu\alpha^{2}}\bigg(\alpha \te_4+ \te_3+  \te_2 \bigg),&
\end{align}}
with $s \in (0,2]$, $\sigma \in (0, s/2)$, 
$\te_1, \te_2, \te_3, \te_4>0$ defined as in \eqref{bound:pertub} and
\begin{equation}\label{mu}
{ \mu\vcentcolon= 2\tau+N+3.}
\end{equation}
\end{lemma}
\begin{proof}
First, we have that the solution $G$ of \eqref{homological} is defined by 
\begin{equation}\label{def:G}
 G_{km\mq\bar{\mq}}=
\begin{cases}
\iu\displaystyle\frac{F_{km\mq\bar{\mq}}}{\langle \omega,k \rangle + \langle \Omega, \mq-\bar{\mq} \rangle} \hspace{1cm} |k|_1+|{\mq-\bar{\mq}}| \neq 0, \\
\hspace{1,5cm} 0 \hspace{2,5cm} \qquad \text{otherwise},
\end{cases}
\end{equation}
On the subset $\Xi$, the estimates \eqref{GG} with $\theta_1=0$ and
  \eqref{GNLS} are provided by Lemma 1 Section 2 of
\cite{poschel1996kam}, while \eqref{GG} for the case $\theta_1>0$ follows from Section 6
of \cite{poschel1996kam}.

\smallskip

We remark that, since the cohomological equation is linear in $F$, its
solution $G$ is given by the sum of the solutions of the following equations
\begin{align}
  \nonumber
&\{N,G^{R1}\}+\{ F^R\} =[F^R],&
\\
  \label{HomoNG}
&  \{N,G^{\ng}\}+\{ F^{\ng}\} =[F^{\ng}],&
\\
  \label{NLS.app}
&\left\{N,\gna \right\}+\left\{F^{\tt NLS}\right\}=[F^{\tt NLS}].&
\end{align}
Since \(\Pi_{\mathcal G} F^{\mathcal{NG}}=0\), its averaged part \([F^{\mathcal{NG}}]\) vanishes (see \eqref{media}). Consequently the right-hand side of \eqref{HomoNG} is zero. By defining $G^R\vcentcolon= G-G^{\tt NLS}$ we have
\begin{equation}\label{decomp:GR}
G^R= G^{R1}+G^{\ng}+\left(\gna-G^{\tt NLS}\right).
\end{equation}

The estimate of $G^{R1}$ follows from Lemma 1 Section 2 of
\cite{poschel1996kam}, {obtaining in particular 
\begin{equation}\label{GR1}
\bn{X_{G^{R_1}}}^{ D_p(s-\sigma,r)}_{r,p-4}  \lesssim \frac{h
  }{\sigma^\mu\alpha} \te_4,
\end{equation}}
We come to the estimate of $G^{\ng}$. To simplify the
notation we denote
$$
\tilde F \vcentcolon= F^{\ng},\quad \tilde G\vcentcolon= G^{\ng}.
$$
Following \cite{poschel1996kam}, we write
$\{\tilde F\}=\tilde
F^0+\tilde F^1+\tilde F^2$, where $\tilde F^j$ contains the terms
of $F$ with degree $j=|\rho+\bar \rho|$ in $(z,\bar z)$ (c.f. \eqref{PolR}). We have 
\begin{equation}\label{aux:decomp}
\begin{aligned}
\tilde F^0= \tilde F^{0,0}(x,y) \quad \tilde F^1=\vcentcolon  \langle \tilde F^{1,0}(x),z \rangle + \langle
\tilde F^{0,1}(x),\bar{z}\rangle,
\\
\tilde F^2=  \vcentcolon\langle \tilde F^{2,0}(x)z,z \rangle +
\langle \tilde F^{1,1}(x)z,\bar{z} \rangle + \langle \tilde F^{0,2}(x)\bar{z},
\bar{z} \rangle .
\end{aligned}
\end{equation}
Correspondingly we write $\tilde G=\tilde G^0+\tilde G^1+\tilde G^2$. We are going to estimate
explicitly only $\tilde G^{0,2}$,  the other terms
 could be estimated in a similar way.

First, we remark that the operator $\tilde F^{0,2}(x)$ is given by $ - \im 2\tilde F^{0,2}(x)=\left.d_{\bar z} [X_{\tilde F}]_{z}\right|_{z=\bar
z=0}.$
Then, by a Cauchy estimate we obtain
$$
\sup_{|\Im x|<s}\left\| \tilde F^{0,2}(x)\right\|_{\ell^{a,q};\ell^{a,q}} \leq
\bn{X_{\tilde F}}^{D_q(s,r)}_{r,q},
$$
for $q=p$ and $q=p-4$. Expanding in Fourier series we have 
\begin{equation}\label{Fexp}
\tilde F^{0,2}(x)=\sum_{k}\tilde F^{0,2}_ke^{i \jbs{k,x}},\quad \left\|\tilde{F}^{0,2}_k\right\|_{\ell^{a,q};\ell^{a,q}}
\leq e^{-s|k|_1} \bn{X_{\tilde F}}^{D_q(s,r)}_{r,q},
\end{equation}
{ and for any $i,j \in \cS^c$, we define the matrix elements $\tilde F^{0,2}_{k,ij} \vcentcolon = \tilde F_{k00\bar{\mq}}$ (c.f. \eqref{PolR}) with $\supp(\bar{\mq})=\{i,j\}$ and $\bar{\mq}_i=1=\bar{\mq}_j$.}
Writing the corresponding expansion for $G$ we have
\begin{equation}\label{stim:esempio}
\tilde G_{k,ij}^{0,2}=\iu \frac{\tilde{F}_{k,ij}^{0,2}}{ \jbs{\omega,k}-(\Omega_i+\Omega_j)}, \qquad
|k|_1+|i+j| \neq 0.
\end{equation}

We consider the set $ \cC \vcentcolon= \{ j \in \mathbb{Z} \mid |j|<c/2 \}$, then
\begin{equation}
\label{decomp:GNG}
\tilde G^{0,2}_k= \Pi_\cC \tilde G^{0,2}_k \Pi_\cC+ \Pi_\cC^{\perp} \tilde G^{0,2}_k \Pi_{\cC}+ \tilde G^{0,2}_k \Pi_\cC^{\perp},
\end{equation}
with the projectors defined as in \eqref{cut-off}.
We estimate one by one the three operators. 
\begin{itemize}
\item[1.] If $|i|,|j|< c/2$ and $|k|_1 \leq \kappa_*c^{1/2}$ then \eqref{bound:NGhom:0} is satisfied and from Lemma \ref{modop} and Remark \ref{Stima:pesata} we obtain
\[
\left\| \Pi_\cC \tilde G^{0,2}_k \Pi_\cC\right\|_{\ell^{a,p};\ell^{a,p}}\sleq \frac{\jbs{k}}{\alpha c^2}\left\| \tilde{F}_k^{0,2}\right\|_{\ell^{a,p};\ell^{a,p}}.
\]
Otherwise, for $|i|,|j|< c/2$ and $|k|_1 \geq \kappa_*c^{1/2}$ we apply \eqref{dioph} and from Lemma \ref{modop} and Remark \ref{Stima:pesata} we obtain 
\[
\left\| \Pi_\cC \tilde G^{0,2}_k \Pi_\cC\right\|_{\ell^{a,p};\ell^{a,p}}\sleq \frac{{\jbs{k}^{\tau+1}}}{\alpha}\left\| \tilde{F}_k^{0,2}\right\|_{\ell^{a,p};\ell^{a,p}}.
\]
Therefore by \eqref{Fexp} we have
\begin{align*}
\sum_{k \in \mathbb{Z}^\cS}e^{|k|_1(s-\sigma)}\left\| \Pi_\cC \tilde G^{0,2}_k \Pi_\cC\right\|_{\ell^{a,p};\ell^{a,p}}&\sleq \bn{ X_{\tilde
  F}}_{r,p}^{D_p(s,r)} \bigg( \sum_{
  |k|_1 \lesssim c^{1/2}} \frac{e^{-\sigma|k|_1}\langle k\rangle}{\alpha c^2}+\sum_{|k|_1 \gtrsim  c^{1/2}} \frac{e^{-\sigma|k|_1}{\langle k\rangle^{\tau+1}}}{\alpha} \bigg)&\\
&\sleq \bn{X_{\tilde F}}_{r,p}^{D_p(s,r)} 
\bigg(
  \frac{1}{\alpha\, c^2\, \sigma^{N+1}}
  + \frac{e^{-\frac{\sigma}{2} c^{1/2}}}{\alpha\, {\sigma^{\tau+N+1}}}
\bigg)
\lesssim
\frac{\bn{X_{\tilde F}}_{r,p}^{D_p(s,r)}}{\alpha\, c^2\, {\sigma^{\tau+N+5}}}
  \end{align*}
where the last inequality is satisfied for any $c \geq c_*$, for some $c_*>0$ independent of $\sigma$.
\item[2.] Now we study the second operator in \eqref{decomp:GNG}. By Lemmata \ref{modop}, \ref{lemma:cutoff-multiplication}, Remark \ref{Stima:pesata} and the assumption \eqref{dioph} we have 
\begin{equation}\label{aux:esempio}
\left\| \Pi_\cC^{\perp}  \tilde G_k^{0,2}  \Pi_{\cC}\right\|_{\ell^{a,p};\ell^{a,p-4}} \lesssim
\| \Pi^{\perp}_{\cC} \|_{\ell^{a,p};\ell^{a,p-4}} \left\| \tilde{G}_k^{0,2}\right\|_{\ell^{a,p};\ell^{a,p}} \lesssim \frac{{\jbs{k}^{\tau+1}}}{\alpha c^4}\left\| \tilde{F}_k^{0,2}\right\|_{\ell^{a,p};\ell^{a,p}},
\end{equation}
therefore
\begin{equation}\label{aux:stima}
\sum_{k \in \mathbb{Z}^\cS}e^{|k|_1(s-\sigma)}\left\|\Pi_\cC^{\perp} \tilde G_k^{0,2}\right\|_{\ell^{a,p};\ell^{a,p-4}}\lesssim \frac{\bn{ X_{\tilde
  F}}_{r,p}^{D_p(s,r)}}{c^4 \alpha {\sigma^{\tau+N+1}}}.
  \end{equation}
\item[3.] For the third operator in \eqref{decomp:GNG}, following the strategy of item $2)$, we have
\[
\left\|  \tilde G_k^{0,2} \Pi_\cC^{\perp}  \right\|_{\ell^{a,p};\ell^{a,p-4}} \lesssim
\left\| \tilde{G}_k^{0,2}\right\|_{\ell^{a,p-4};\ell^{a,p-4}} \| \Pi^{\perp}_{c/2} \|_{\ell^{a,p};\ell^{a,p-4}} \lesssim \frac{{\jbs{k}^{\tau+1}}}{\alpha c^4}\left\| \tilde{F}_k^{0,2}\right\|_{\ell^{a,p-4};\ell^{a,p-4}},
\]
which gives 
\begin{equation}\label{aux:stima1}
\sum_{k \in \mathbb{Z}^\cS }e^{|k|_1(s-\sigma)}\left\| \tilde G_k^{0,2}\Pi_\cC^{\perp}\right\|_{\ell^{a,p};\ell^{a,p-4}}\lesssim \frac{\bn{ X_{\tilde
  F}}_{r,p-4}^{D_{p-4}(s,r)}}{c^4 \alpha {\sigma^{\tau+N+1}}}.
  \end{equation}
\end{itemize}
 Recalling the definition of $\te_2$ in \eqref{bound:pertub}, combining the three items above we obtain
\begin{equation}\label{aux:finalee}
\frac{1}{r}\sup_{D(s-\sigma,r)} \|\tilde G^{0,2}\bar{z}\|_{r,p-4}\leq
\sup_{|\Im(x)|<(s-\sigma)}\left\|\tilde G^{0,2}\right\|_{\ell^{a,p};\ell^{a,p-4}}
\sleq  \frac{ h }{\alpha \sigma^{\mu}} \te_2 ,
\end{equation}
with ${\mu=2\tau+N+3}$. The estimate for the other terms of $\Pi_{\ng}G$ follows by similar, but simpler computations.

\medskip

Finally, we estimate
$\gna-G^{\tt NLS}$. First we remark that from \eqref{HomoNLS}, \eqref{NLS.app} and Lemma \ref{lem:gauge2}, $\gna$ and $G^{\tt NLS}$ are Gauge invariant.

 Subtracting \eqref{HomoNLS} from \eqref{NLS.app} we obtain
\begin{equation}
  \label{app.G}
\left\{N^{\tt NLS},\gna-G^{\tt NLS}\right\}=\left\{N^{\tt NLS}-N,\gna\right\}.
\end{equation}
{For any $(k,\mq-\bar{\mq}) \in \mZ_\mG$ and $|\rho+\bar{\rho}| \neq 0$, the Fourier-Taylor coefficients of the r.h.s. of \eqref{app.G}, are given by 
\begin{equation}\label{op:dastim}
\im \left[\jbs{\omega^R(h),
  k}+\jbs{\Omega^R(h),\mq-\bar{\mq}} \right] \gna_{k0\mq \bar{\mq}},
\end{equation}
with $\omega^R, \Omega^R$ defined in \eqref{OomegaR:hom}. Then, by \eqref{app.G} and \eqref{NLS.app} we have 
\begin{equation}\label{decomp:diff}
(\gna-G^{\tt NLS})_{k0 \mq \bar{\mq}}= \frac{\im \left[\jbs{\omega^R(h),
  k}+\jbs{\Omega^R(h),\mq-\bar{\mq}} \right]}{(\jbs{\omega,
  k}+\jbs{\Omega,\mq-\bar{\mq}})(\jbs{\omega^{\tt NLS},
  k}+\jbs{\Omega^{\tt NLS},\mq-\bar{\mq}})} F^{\tt NLS}_{k0\mq \bar{\mq}}.
\end{equation}
By definition \eqref{OomegaR:hom} and the estimates \eqref{stimadifdelta} and \eqref{stimeom0}, we have that 
\begin{equation}\label{bound:perdita}
\left|\jbs{\omega^R(h),
  k}+\jbs{\Omega^R(h),\mq-\bar{\mq}} \right| \lesssim  h \bigg(|k|_1+\sum_{j \in \supp(\mq-\bar{\mq})}j^4\bigg).
\end{equation}
Then, by the assumptions \eqref{dioph},   \eqref{dioph:NLS} and the rescaling argument in \cite{poschel1996kam} (Section 2, pag 10),  we can argue as above to obtain 
\begin{equation}\label{finale:NG}
\bn{X_{\gna-G^{\tt NLS}}}_{r,p-4}^{D_p(s,r)} \sleq
\frac{ h }{\alpha^2\sigma^{2\tau+N+3}} \te_3.
\end{equation}}
{Considering the decomposition \eqref{decomp:GR}, by the estimates of $G^{\ng}$ (see \eqref{finale:NG}), of $G^{R1}$ in \eqref{GR1} and the bound \eqref{finale:NG}, we prove \eqref{4.13a}, which concludes the proof.}
\end{proof}

%%%%%%%%%%%%%%%%%%%%

\section{Measure estimates}\label{meas.est}
For any $\ell \in \mathbb{Z}^{\cS^c}$ ($\cJ^c=\mathbb{Z} \setminus \cS$), with $|\ell|_1 \leq 2$ we define
\begin{equation}\label{supp:ell}
\tw(\ell) \vcentcolon= 
\begin{cases} 
\qquad 1 & \ell=0, \\ 
\displaystyle\min_{i\in \supp(\ell) }\tw_i &  \ell \neq 0.
\end{cases}
\end{equation}
Let $\omega,\Omega$ be the frequencies \eqref{omegaprima} defined on $\Xi_0$ (see \eqref{Xi0}). Given $\alpha>0$, $\tau >0 $, and $\theta\in[0,1)$, for each $(k,\ell)\in\mZ_\mM$ (see \eqref{Indici:M}) we define
\begin{align}  \label{res1}
  &\cR^{\theta}_{k \ell}(h;\alpha, \delta, \Delta)\vcentcolon= 
\left\{\xi\in\Xi_0\ :
\left|\langle\omega ,k\rangle+\langle\Omega, \ell \rangle\right|<
\frac{\alpha}{\langle 
  k\rangle^\tau \tw(\ell) ^{\theta}}
\right\}, \\
\label{res2}
&\cR^{\theta}_k(h;\alpha, \delta, \Delta)\vcentcolon= \bigcup_{\substack{\ell \in \mathbb{Z}^{\cS^c} \\ (k, \ell) \in \mZ_\mM}}\cR^{\theta}_{k \ell}(h;\alpha, \delta, \Delta),
\end{align}
with $\jbs{k}$ introduced in \eqref{jb0}.

Moreover, considering the NLS frequencies \eqref{omegaprima:NLS}, for each $(k,\ell)\in\mZ_\mG$ (see \eqref{Indici:G}) we set
\begin{equation}\label{R:NLS}
  \cR^{\tt NLS}_{k \ell}(0;\alpha, \delta^{\tt NLS}, \Delta^{\tt NLS})\vcentcolon= 
  \left\{\xi\in\Xi_0\ :
\left|\langle\omega^{\tt NLS},k\rangle+\langle\Omega^{\tt NLS},\ell \rangle\right|<
\frac{\alpha}{\langle 
  k\rangle^{\tau}}
\right\},
 \end{equation}
and we define the union set as
\begin{equation}\label{res:setNLS}
\cR^{\tt NLS}_{k}(0;\alpha, \delta^{\tt NLS}, \Delta^{\tt NLS})\vcentcolon= \bigcup_{\substack{\ell \in \mathbb{Z}^{\cS^c} \\ (k, \ell) \in \mZ_\mG}}\cR^{\tt  NLS}_{k \ell}(0;\alpha, \delta^{\tt NLS}, \Delta^{\tt NLS}).
\end{equation}
{\begin{remark}
We clarify the relation between the Melnikov conditions \eqref{dioph}, \eqref{dioph:NLS} and the definitions of the resonant sets \eqref{res1}, \eqref{R:NLS}. By the zero momentum condition \eqref{Indici:M} we have 
\[
\bigg| \sum_{n \in \cS^c} n \ell_n \bigg|=\bigg| \sum_{n \in \cS} n k_n \bigg| \lesssim |k|_1,
\]
which implies $\jbs{\ell}_1 \lesssim \jbs{k}$ (c.f. \eqref{jb}, \eqref{jb0}). Therefore, for any $\xi \in \Xi_0 \setminus \cR^{\theta}_{k \ell}(h)$ we have 
\[
| \jbs{\omega,k}+\jbs{\Omega, \ell}| \gtrsim \frac{\alpha \jbs{\ell}_1}{ \jbs{k}^{\tau+1} \tw(\ell)^{\theta}},
\]
which corresponds to \eqref{dioph}. In the same way, from the zero momentum condition, the Melnikov condition \eqref{dioph:NLS} follows for any $\xi \in \Xi_0 \setminus \cR^{\tt NLS}_{k \ell}$.
\end{remark}}

\medskip 
The main purpose of this section is to provide measure estimates for the sets \eqref{res2} and \eqref{res:setNLS}, and to provide a lower bound of the non-Gauge divisors. 

This section is organized as follows:
\begin{itemize}
\item In Section \ref{sec:triv}, we provide the measure estimates for the single resonant sets \eqref{res1} and \eqref{R:NLS}. We then discuss the measure estimate of the union of the resonant sets of NLS  \eqref{res:setNLS}, following the arguments in \cite{poschel1996kam} (Section 5, Lemma 7).

\item The remaining sections are devoted to prove measure estimates on the union of KG resonant sets \eqref{res2}, together with the validation {of the hypothesis \ref{ip:nongauge} of Lemma \ref{solhomo}}. We begin with trivial cases in Section \ref{simo:cas}. Sections \ref{nontrivial} and \ref{nontrivial:1} address the analysis of Second Melnikov conditions involving the difference and the sum of the normal frequencies. In particular, we first present the situations where the momentum condition \eqref{Indici:M} suffices to establish the estimates (see Sections \ref{simo:cas}, \ref{menodiversi} and \ref{fac:2} ). We then turn to the genuinely non-trivial cases, where the main difficulty is to study the dependence on $c$. 
\end{itemize}

\subsection{{ Measure estimates of the single resonant sets $\cR_{k \ell}^{\theta}, \cR_{k \ell}^{\tt NLS}$, $\cR_k^{\tt NLS}$. }}\label{sec:triv}
First, we consider the case $k=0$.
\begin{lemma}
\label{remark:trivial}
There exist $h_*,\varepsilon_*,\alpha_*,r_*>0$ such that for any $h\in(0,h_*)$, $\varepsilon\in(0,\varepsilon_*)$, $\alpha\in(0,\alpha_*)$, $r\in(0,r_*)$ and $\theta \in [0,1)$ the following holds. {For any $\ell \in \mathbb{Z}^{\cS^c}$ such that $(0,\ell) \in \mZ_\mM$ we have
\begin{equation}\label{div:k0}
|\jbs{\Omega, \ell}|  \geq \alpha c^2,
\end{equation}
and therefore $\cR^{\theta}_{k\ell}(\alpha) = \emptyset$. Furthermore, if $(0, \ell) \in \mZ_\mM$ then $(0, \ell) \not \in \mZ_\mG$.}
\end{lemma}

\begin{proof}
We first prove \eqref{div:k0}. Let $(k,\ell) \in \mZ_\mM$. For $k=0$ the zero momentum condition \eqref{Indici:M} reduces to
\begin{equation}\label{mom:k0}
\sum_{n \in \cS^c} n \, \ell_n = 0.
\end{equation}
{We remark that if $\# \supp (\ell)=1$ then \eqref{mom:k0} implies $\supp(\ell)=\{0\}$. Otherwise, if $\supp(\ell)=\{i,j\}$ for some $i, j \in \cS^c$ with $i \neq j$, then by \eqref{mom:k0} we have $i= -j$ and then $\ell_i \ell_j=1$. 
Thus, for any $(0,\ell) \in \mZ_\mM$ we have
\[
|\jbs{\varLambda, \ell}|=\sum_{n \in \supp(\ell)} |\ell_n| \lambda_n,
\]
with $\varLambda$ defined in \eqref{lambda-Lambda}.  By \eqref{omegaprima}, the estimates \eqref{Omega:lambda}, \eqref{lipfre} and $\lambda_n \geq c^2$, there exists $C>0$ such that
\[
\begin{aligned}
| \jbs{\Omega, \ell}| &\geq | \jbs{ \varLambda, \ell } | - C(\e+r^{4/3})& \\
&\geq |\ell|_1c^2 - C(\e+r^{4/3}).& 
\end{aligned}
\]
Then, provided that $\e, r$ and $\alpha$ are small enough, and $c$ large enough, we obtain \eqref{div:k0}.}

\medskip
\noindent
To prove the second claim, we observe that for $k=0$ and $\ell \neq 0$, the Gauge invariance condition \eqref{Indici:G} reduces to $
\sum_{j\in\cS^c}\ell_j=0$, which implies $\supp(\ell)=\{i,j\}$ for some $i, j \in \cS^c$ with $i \neq j$ and $\ell_i\ell_j=-1$. However, this contradicts \eqref{mom:k0}.
\end{proof}

\begin{lemma}\label{singole:stim}
There exist $h_*,\e_*,\varsigma_*,r_*, \alpha_*>0$ such that the following holds. For any $h \in (0,h_*)$, $\e \in
(0,\e_*)$, $r\in(0,r_*)$, $ \alpha \in (0,\alpha_*)$, and $\theta \in [0,1)$, if  
\begin{equation}\label{cond:zeta}
\e \alpha^{-1}<\varsigma_*, 
\end{equation}
 then
  \begin{align}\label{singl:set:mes1}
&\left|\mathcal{R}^{\theta}_{k\ell}(\alpha)\right|\sleq
\frac{\alpha}{\jbs{k}^{\tau+1}}\, r^{\frac{4}{3}(N-1)},\qquad \forall \, (k,\ell) \in \mZ_\mM,\\
\label{singl:set:mes2}
&\left|\mathcal{R}^{\tt NLS}_{k\ell}(\alpha)\right|\sleq
\frac{\alpha}{\jbs{k}^{\tau+1}}\, r^{\frac{4}{3}(N-1)},\qquad \forall \, (k, \ell) \in \mZ_\mG.
  \end{align}
\end{lemma}

\begin{proof}
We start by \eqref{singl:set:mes1}. For $(k, \ell)\in \mathcal{Z}_{\cM}$ and $\xi \in \Xi_0$, we set
$$
\Phi_{k\ell}(\xi)\vcentcolon= \langle\omega(\xi),k\rangle+\langle\Omega(\xi),\ell \rangle, 
$$
so that for all $\xi,\eta\in\Xi_0$
\begin{align*}
\Phi_{k\ell}(\xi)-\Phi_{k\ell}(\eta)&=\langle \matA k+\matB^t\ell,\xi-\eta\rangle+
  \langle(\delta(\xi)-\delta(\eta)),k\rangle+\langle(\Delta(\xi)-\Delta(\eta)),\ell \rangle.
\end{align*}
The case $k=0$ is covered by Lemma~\ref{remark:trivial}, so we assume $k\neq0$. By the estimates \eqref{lipfre}, recalling the definition of the norms \eqref{norma:OmegaL}, we have $|\delta|_\infty^{\mL} \lesssim \e \alpha^{-1}$ and $|\Delta|_\infty^{\mL} \lesssim \e \alpha^{-1}$. Hence, since $|\ell|_1 \leq 2$, we have
\begin{equation}\label{stim-rest}
 |\langle(\delta(\xi)-\delta(\eta)),k\rangle+\langle(\Delta(\xi)-\Delta(\eta)),\ell \rangle | \leq (|\delta|_\infty^{\mL}|k|_1 + 2|\Delta|_\infty^{\mL}) {|\xi-\eta|} \leq C_1 \e \alpha^{-1} |k|_1 {|\xi-\eta|},
\end{equation}
form some constant $C_1>0$. Moreover, for any $c$ large enough we can define
$$
b\vcentcolon=  \frac{1}{|\matA k+\matB^t\ell |_2}(\matA k+\matB^t\ell),
$$
since by \eqref{grande0} the denominator is different from zero.

For any $\xi\in\Xi_0$, there exist $s_1 \in \mathbb{R}$ and $v \in b^{\perp}$ such that $\xi = s_1 b + v$.
For $s\in\R$, we define
\begin{equation}\label{psi}
\psi(s) \vcentcolon= \Phi_{k\ell}((s+s_1)b + v),
\qquad \textnormal{whenever } (s+s_1)b + v \in \Xi_0.
\end{equation}
Combining \eqref{grande0} and \eqref{stim-rest}, there exists a constant $C_2>0$ such that for any $s,t$, for which \eqref{psi} is well defined, we have 
\begin{equation}\label{stim:lipaux}
|\psi(s)-\psi(t)| 
  \ge \bigl(|\matA k+\matB^t \ell|_2 -C_1\e \alpha^{-1} |k|_1\bigr) |s-t|
  \ge (C_2 - C_1\e\alpha^{-1}) |k|_1|s-t|
  \ge \frac{C_2}{2}|k|_1|s-t|,
\end{equation}
where we used the smallness condition \eqref{cond:zeta}, choosing $\varsigma_*<C_2/(2C_1)$.

Applying Fubini's theorem we obtain
\[
\left|\cR^{\theta}_{k \ell}(\alpha)\right| 
   \lesssim  r^{\frac{4}{3}(N-1)} \frac{\alpha}{ \tw(\ell)^{\theta} \jbs{k}^{\tau+1}}.
\]
%with $\Pi_{b^{\perp}}(\Xi_0)$ the projection of $\Xi_0$ on $b^{\perp}$. 

By the definition \eqref{supp:ell} and the first estimate in \eqref{bound:peso}, we obtain \eqref{singl:set:mes1}.

\smallskip

The estimate \eqref{singl:set:mes2} follows by the same argument, using Lemma~\ref{remark:trivial} and the bound \eqref{grand0:NLS}.

\end{proof}

We end this section with the estimate of the resonant sets \eqref{res:setNLS} for NLS.
\begin{theorem}
There exist $\e_*,\varsigma_*,r_*, \alpha_*>0$ such that for any $\e \in
(0,\e_*)$, $r\in(0,r_*)$ and $\alpha \in (0,\alpha_*)$ the following holds. 
By assuming that $\e \alpha^{-1} < \varsigma_*$, we have 
\begin{equation}
    \label{meas.finalNLS}
\left|\cR_k^{\tt NLS}(0; \alpha, \delta^{\tt NLS}, \Delta^{\tt NLS}) \right|
\sleq
\frac{\alpha}{\langle k\rangle^{\tau-1}} \,
r^{\frac{4}{3}(N-1)} \qquad \forall k \in \mathbb{Z}^\cS.
\end{equation}
\end{theorem}
\begin{proof}
By the argument in Section 5 of \cite{poschel1996kam} (see Lemma~7 therein), we have that for any $k \in \mathbb{Z}^{\cS}$, with $k \neq 0$, it holds
\[
{ \# \{ \ell \in \mathbb{Z}^{\cS^c} \mid (k, \ell) \in \mZ_\mG, \, \cR^{\tt NLS}_{k\ell}(\alpha) \neq \emptyset, \, \, \supp(\ell) \neq \{i,-i\} \, \, \forall \, i \in \cS^c\setminus \{0\} \} \lesssim |k|_1^2,}
\]
with $\mZ_\mG$ defined in \eqref{Indici:G}. 

{Otherwise, when $\supp(\ell)=\{i,-i\}$ for some $i \in \cS^c \setminus \{0\}$, we need to deal with the multiplicity of the linear frequencies.  In particular, we need to consider two cases. For $\ell_i \ell_{-i}=1$, by the same argument of P\"oschel, for a fixed $k$, the cardinality of $i$ such that $\cR_{k \ell}^{\theta}  \neq \emptyset $ is of order $|k|_1$. On the other hand, for $\ell_i\ell_i=-1$ we need to apply the zero momentum condition \eqref{Indici:M}, obtaining 
\[
2|i|=\bigg| \sum_{n \in \cS^c} n \ell_n \bigg|=\bigg|\sum_{n \in \cS} n k_n \bigg| \lesssim |k|_1.
\]}

Recalling the measure estimate \eqref{singl:set:mes2} we obtain \eqref{meas.finalNLS}. 
The case $k=0$ is excluded by Lemma \ref{remark:trivial}.
\end{proof}

\subsection{{Estimate of the Klein-Gordon resonant set $\cR_k^{\theta}$}}

For $(k,\ell)\in\cZ_{\cM}$ (see \eqref{Indici:M}) we remark that if $\cL\neq0$ (see \eqref{Indici:G}) the corresponding Fourier Taylor monomial \eqref{FTmon} is not gauge invariant.
The main result of the upcoming sections is the following.
\begin{theorem}
  \label{measure_estimate}
 Let  $\theta_0 \in [0,1)$ and $\theta_1 \in [0,1-\theta_0)$. There exist $h_*,\e_*,\varsigma_*,r_*, \alpha_*>0$ such that for any $h \in (0,h_*)$,  $\e \in
(0,\e_*)$, $r\in(0,r_*)$ and $\alpha \in (0,\alpha_*)$  the following holds. If $|\Delta|_{\infty,1-\theta_0} \lesssim \e$ (recall \eqref{norma:Omega}) and $ \e \alpha^{-1}<\varsigma_*$, then
  \begin{equation}
    \label{meas.final}
\left|\cR^{\theta_1}_k(h; \alpha, \delta, \Delta) \right|\sleq
\frac{\alpha^{\frac{2}{3-\theta_1}}}{\langle k\rangle^{\tau_{(\theta_1,\theta_0)}}}\,r^{\frac{4}{3}(N-1)}, \qquad \forall \, k \in \mathbb{Z}^{\cS},
  \end{equation}
with $\tau_{(\theta_1,\theta_0)}\vcentcolon=   \tau\frac{(1-\theta_0)}{(3-\theta_1)}-5$.

Moreover, there exists $\kappa_*$ such that for any $(k, \ell ) \in \mZ_\mM$ with $\cL \neq 0$, $|k|_1 \leq \kappa_* c^{1/2}$ and $\supp(\ell) \subset [-c/2,c/2]$, the following lower bound holds
\begin{equation}\label{bound:NGhom}
\left|\jbs{\omega,k}+ \jbs{\Omega, \ell} \right|\geq \alpha c^2.
\end{equation}
\end{theorem}

The following sections are devoted to the proof of Theorem~\ref{measure_estimate}. We split the analysis into several cases. A key role is played by the zero momentum condition \eqref{Indici:M}, which gives bounds on the indices in $\supp(\ell)$.

Considering that the frequencies are of the form \eqref{omegaprima} and by the estimates \eqref{Omega:lambda}, \eqref{lipfre} we have
\begin{align}\label{lambda.circa2}
{\omega=\lambda+\cO(\e+r^{4/3}), \quad \Omega=\varLambda +\cO(\e+r^{4/3}).}
\end{align}
Then, recalling the definition of $\cL$ in  \eqref{Indici:G}, the divisor can be written as 
\begin{align}
  \label{elle.c2}
\jbs{\omega, k} + \jbs{\Omega, \ell}=\cL c^2+\langle\nu,k\rangle+ {\sum_{n \in \supp(\ell)} \ell_n \nu_n} +\cO(|k|_1(\e+r^{4/3})),
\end{align}
where $\nu_j$ is defined in \eqref{def:h} and $\nu \vcentcolon =(\nu_j)_{j\in\cS}$.

\medskip

{We need to estimate $\cR_{k \ell}^{\theta}$, which, for any fixed $k \in \mathbb{Z}^{\cS^c}$, is defined as the union of $\cR_{k \ell}^{\theta}$ over any $\ell$ such that $(k,\ell) \in \mZ_\mM$. We will study separately the union over the following subsets
\begin{align}
\label{S0}
&\ccS_k^{(0)} \vcentcolon= \{ \ell \in \mathbb{Z}^{\cS^c} \mid (k, \ell) \in \mZ_\mM, \, \exists \, i \in \cS^c \, \, \textnormal{s.t. } \, \supp(\ell)=\{i\} \, \, \textnormal{or } \, \supp(\ell)=\{i, 0\} \}, &\\
\label{SA}
&\ccS_k^{-} \vcentcolon= \{ \ell \in \mathbb{Z}^{\cS^c} \mid  \, (k,\ell) \in \mZ_\mM, \, \exists \, i,j\in\cS^c\setminus\{0\} 
 \textnormal{ s.t. } i \neq j,  \supp(\ell) =  \{i,j\},  \ell_i\ell_j= \textnormal{-}1\},&\\
\label{SA1}
&\ccS_k^{+} \vcentcolon= \{ \ell \in \mathbb{Z}^{\cS^c} \mid  \, (k,\ell) \in \mZ_\mM, \, \exists i,j\in\cS^c\setminus\{0\} 
 \textnormal{ s.t. } i \neq j,  \, \supp(\ell) =  \{i,j\},\,  \ell_i\ell_j=1\},&
\end{align}
noting that  $\{ \ell \in \mathbb{Z}^{\cS^c} \setminus \{0\} \mid (k, \ell) \in \mZ_\mM \}= \ccS_k^{(0)} \cup \ccS_k^+ \cup \ccS_k^-$.}  

\begin{remark}\label{sup:1}
From now on for any $\ell \in \ccS_k^{\pm }$, we will denote $ \{i,j\} \vcentcolon= \supp(\ell)$, assuming $|i| \leq |j|$.
\end{remark}
\subsection{{Estimates of the union over $\ccS_k^{(0)}$}}\label{simo:cas}
{We discuss here the union of $\cR_{k \ell}^{\theta}$ over $\ccS_k^{(0)}$ defined in \eqref{S0}. In this case, we are going to see that the momentum condition \eqref{Indici:M} suffices to obtain for any $\ell \in \cS_k^{(0)}$ a bound on the possible values of $i$ such that $i\in \supp(\ell)$.
\begin{lemma} \label{S0} 
Let $\theta \in [0,1)$. There exist $h_*,\e_*,\varsigma_*,r_*, \alpha_*>0$ such that for any $h \in (0,h_*)$, $\e \in (0,\e_*)$, $r\in(0,r_*)$ and $ \alpha \in (0,\alpha_*)$ such that $ \e \alpha^{-1}<\varsigma_*$, the following estimate holds
\begin{equation} \label{unione:S0} 
    \left|\bigcup_{\ell\in\ccS_k^{(0)}}\cR^{\theta}_{k\ell}(\alpha )\right|\leq \frac{\alpha}{\langle k\rangle^{\tau+1}}\,r^{\frac{4}{3}(N-1)}, \qquad \forall \, k \in \mathbb{Z}^{\cS}.
\end{equation} 
\end{lemma} 
\begin{proof}
Let $k \in \mathbb{Z}^{\cS} \setminus \{0\}$. Considering $\ell \in \mathcal{S}_k^{(0)}$, we need to study two cases. If $\# \supp(\ell)=1$, by \eqref{Indici:M}, for any fixed $k \in \mathbb{Z}^{\cS}$ there exist at most four possible values of $i \in \supp(\ell)$ (corresponding to $\ell_i\in\{\pm1,\pm2\}$) such that the zero momentum condition is fulfilled. Hence, using the bound \eqref{singl:set:mes1} we obtain \eqref{unione:S0}  for this sub-case. 
The case with $ 0 \in \supp(\ell)$ and $|\ell|_1=2$, follows similarly from \eqref{Indici:M}.
\end{proof}}

\medskip

Regarding the bound \eqref{bound:NGhom}, we have the following Lemma.
\begin{lemma}\label{caso:trivial}
There exist $h_*,\e_*,r_*, \alpha_*, \kappa_*>0$ such that for any 
$h \in (0,h_*)$, $\e \in (0,\e_*)$, $r\in(0,r_*)$, and $\alpha \in (0,\alpha_*)$, the following holds. 
{For any $k \in \mathbb{Z}^\cS$ and $\ell \in \ccS_k^{(0)} \cup \{0\}$ such that $\cL \neq 0$ (see \eqref{Indici:G}) and $|k|_1 \leq \kappa_* c$, the estimate \eqref{bound:NGhom} is satisfied.}
\end{lemma}
\begin{proof}
The case $k=0$ is covered by Lemma \ref{remark:trivial}. Consider $(k, \ell) \in \mZ_\mM$ such that $k \neq 0$, $\cL \neq 0$ and $\supp(\ell)=\{i\}$, for some $i \in \cS^c \setminus \{0\}$. Since $|\ell_i \nu_i| \lesssim |k|_1^2$, by the fact that $|\nu_i|\lesssim i^2$ and \eqref{Indici:M}, using \eqref{elle.c2} we obtain the estimate
\begin{align}
\label{auxx:remark}
| \jbs{\omega,k}+ \jbs{\Omega, \ell} | &\geq |\cL|  c^2 - |\jbs{\nu,k}|-|\ell_i \nu_i| - C_0(r^{4/3}+\e)|k|_1& \\
\nonumber
&\geq c^2-C_0|k|_1-C_1|k|_1^2-C_2(r^{4/3}+\e)|k|_1,&
\end{align}
for some constants $C_2, C_1,C_0>0$. Then for $\e$ and $r$ small enough and $c$ large enough, there exists $C_3>0$ such that 
\[
| \jbs{\omega,k}+ \jbs{\Omega, \ell}|  \geq c^2- C_3 |k|_1^2.
\]
For any $|k|_1 \leq \kappa_* c $, with $\kappa_* \leq 1/\sqrt{(2C_3)}$, this implies \eqref{bound:NGhom}, provided that $\alpha$ is small enough. The cases $ 0 \in \supp(\ell)$ and $\supp(\ell) = \emptyset $ follow by similar, but simpler, computations. 
\end{proof}
\subsection{Estimate of the {union over $\ccS_k^-$} }\label{nontrivial}

We start by proving the bound \eqref{bound:NGhom} on $\ccS_k^-$ defined in \eqref{SA}.

\begin{lemma} \label{non.gauge.1} 
There exist $h_*,\e_*,r_*, \alpha_*, \kappa_*>0$ such that for any 
$h \in (0,h_*)$, $\e \in (0,\e_*)$, $r\in(0,r_*)$, and $\alpha \in (0,\alpha_*)$, the following holds. 
For any $(k , \ell) \in \mZ_\mM$, such that $\ell \in \ccS_k^{-}$ (see \eqref{SA}), $|k|_1 \leq \kappa_* c$ and $\cL \neq0$ (see \eqref{Indici:G}), then 
    \[ 
        \left|\jbs{\omega,k}+\jbs{\Omega,\ell}\right| \geq \alpha c^2. 
    \]
\end{lemma}
\begin{proof}
The case $k=0$ is covered by Lemma \ref{remark:trivial}. Let $k \in \mathbb{Z}^{\cS} \setminus \{0\}$ and $\ell \in \ccS_k^{-}$. By \eqref{B.2} and the momentum conservation \eqref{Indici:M}, we obtain
\begin{equation}\label{eq:lambda-diff} 
    |\nu_j - \nu_i|= |\lambda_j-\lambda_i| \leq c ||j|-|i|| \leq c |j-i| \lesssim c|k|_1.
\end{equation} 
If  $\cL \neq 0$, by using \eqref{elle.c2} we have
\begin{align} 
    \nonumber 
    \big|\jbs{\omega,k} + \jbs{\Omega,\ell}\big| 
        &\ge c^2|\cL| - |\nu_j-\nu_i| - |\jbs{\nu,k}| - C_1( \e + r^{4/3})|k|_1 \\ 
    \label{eq:bound-Lneq0} 
        &\ge c^2 - c \, C_0|k|_1 - C_2|k|_1 - C_1( \e + r^{4/3})|k|_1,
\end{align} 
for some positive constants $C_0,C_1, C_2$. For $\e$, $r$ small enough and $c$ large enough, there exists $C_3>0$ such that 
\begin{equation}\label{aux:C3}
\big|\jbs{\omega,k} + \jbs{\Omega,\ell}\big|  \geq c(c-C_3|k|_1).
\end{equation}
By taking $|k|_1 \le \kappa_* c$ in \eqref{aux:C3}, with $\kappa_* \leq 1/(2C_3)$, 
provided $\alpha$ is small enough, we have the thesis.
\end{proof} 

{ To study the measure estimate, we are going to further decompose the set $\ccS_k^-$ introducing the subsets (recall Remark \ref{sup:1})
\begin{align}
\label{SK1}
&\ccS_k^{(1)} \vcentcolon= \{ \ell \in \ccS_k^{-} \mid \sgn(i) \neq \sgn(j) \},& \\
\label{SK2}
&\ccS_k^{(2)} \vcentcolon= \{ \ell \in \ccS_k^{-} \mid  \sgn(i) = \sgn(j), |i| \leq |j|/2 \}, & \\
\label{SK3}
&\ccS_k^{(3)} \vcentcolon = \{ \ell \in \ccS_k^{-} \mid  \sgn(i) = \sgn(j), |j|/2 < |i| \leq |j| \}&,
\end{align}
and remark that $\ccS_k^-= \ccS_k^{(1)} \cup \ccS_k^{(2)} \cup \ccS_k^{(3)}$.}

\subsubsection{{ Estimate of the union over $\ccS_k^{(1)} \cup \ccS_k^{(2)}$}}\label{menodiversi}
{Let $k \in \mathbb{Z}^{\cS}$, we have the following Lemma. 
\begin{lemma} \label{S12} 
Let $\theta \in [0,1)$. There exist $h_*,\e_*,\varsigma_*,r_*, \alpha_*>0$ such that for any $h \in (0,h_*)$, $\e \in (0,\e_*)$, $r\in(0,r_*)$ and $ \alpha \in (0,\alpha_*)$ such that $ \e \alpha^{-1}<\varsigma_*$, the following estimate holds
\begin{equation} \label{elle.2.eq:1} 
    \left|\bigcup_{\ell\in\ccS_k^{(1)}\cup \ccS_k^{(2)} }\cR^{\theta}_{k\ell}(\alpha )\right|\leq \frac{\alpha}{\langle k\rangle^{\tau-1}}\,r^{\frac{4}{3}(N-1)}, \qquad \forall \, k \in \mathbb{Z}^{\cS}.
\end{equation} 
\end{lemma} 
\begin{proof}
First, we remark that if $\ell \in \ccS_k^{(1)}$, the zero momentum condition \eqref{Indici:M} implies $|i|\lesssim|k|_1$ and $|j|\lesssim|k|_1$. Otherwise, if $\ell \in \ccS_k^{(2)}$, from the zero momentum condition \eqref{Indici:M} we have $
\tfrac{|j|}{2} \lesssim |j|-|i| \lesssim |k|_1$, which implies again $|j| \lesssim |k|_1$, and thus $|i| \lesssim |k|_1$. Combining this bound with \eqref{singl:set:mes1}, we obtain \eqref{elle.2.eq:1}.
\end{proof}}

\subsubsection{{Estimate of the union over $\ccS_k^{(3)}$}}\label{meno:triv}

{In this section we study the union of the resonant sets over $S_k^{(3)}$ introduced in \eqref{SK3}. We consider the following sets
\begin{align*}
I_1 &\vcentcolon= \{(i,j) \in \mathbb{Z}^2 \mid |i|<|j|\leq 2c, \, \sgn(i)=\sgn(j)\}, \\
I_2 &\vcentcolon= \{ (i,j) \in \mathbb{Z}^2  \mid c \leq |i|< |j| \leq 2c^3, \, \sgn(i)=\sgn(j)\}, \\
I_3 &\vcentcolon= \{(i,j) \in \mathbb{Z}^2 \mid c^3 \leq |i| < |j|, \, \sgn(i)=\sgn(j) \}.
\end{align*}
Then, by defining 
\begin{equation}
\label{SK3}
\ccS^{(4)}_k \vcentcolon= \left\{\ell \in \ccS_k^{-} \mid (i,j)\in I_1\cup I_2\right\}, \quad  \ccS^{(5)}_k \vcentcolon= \left\{\ell \in \ccS_k^{-} \mid (i,j)\in I_3\right\},
\end{equation}
we have $\ccS_k^{(3)} \subseteq \ccS_k^{(4)} \cup \ccS_k^{(5)}.$} Then we estimate this union.

{\begin{lemma} \label{I1e2} 
There exist $h_*,\e_*,r_*, \alpha_*, \kappa_*, \varsigma_*>0$ such that for any 
$h \in (0,h_*)$, $\e \in (0,\e_*)$, $r\in(0,r_*)$, $\alpha \in (0,\alpha_*)$ and $k \in \mathbb{Z}^{\cS}$ the following holds
\begin{enumerate}[label=(\alph*)] 
    \item \label{non.gauge.2} 
   If $\ell \in \ccS_k^{(5)}$, $|k|_1 \leq \kappa_* c$ and $\cL = 0$ then 
    \[ 
        \left|\jbs{\omega,k}+\jbs{\Omega,\ell}\right| \geq \alpha c. 
    \]
    \item \label{I1e2.tesi} 
    If the smallness condition $\e \alpha^{-1} < \varsigma_*$ is satisfied, we have 
\begin{equation} \label{I1e2.meas.1} 
    \left|\bigcup_{\ell\in\ccS_k^{(4)}}\cR^{\theta}_{k\ell}(\alpha )\right|\leq \frac{\alpha}{\langle k\rangle^{\tau-5}}\,r^{\frac{4}{3}(N-1)}.
\end{equation} 
\end{enumerate} 
\end{lemma} }

\begin{proof} 
The case $k=0$ is covered by Lemma \ref{remark:trivial}. For any $k \in \mathbb{Z}^{\cS} \setminus \{0\}$ we split the proof into two sub-cases.

\noindent\textbf{Case \(\cL = 0\).} Let $\cL=0$, we start proving \ref{non.gauge.2}. If $(i,j) \in I_1 \cup I_2$ we have
\begin{equation}\label{aux:I1I2}
  {  \frac{|i|+|j|}{\sqrt{c^2+i^2} + \sqrt{c^2+j^2}} 
        = \frac{|i|+|j|}{|i|\sqrt{1 + \frac{c^2}{i^2}} + |j|\sqrt{1 + \frac{c^2}{j^2}}} 
        \ge \frac{|i|+|j|}{\sqrt2 |i| + \sqrt2|j|} = \frac{1}{\sqrt2}, }
\end{equation}
which combined with the second equality of \eqref{C.56}, gives 
\begin{equation}\label{diff:aux1} 
  \nu_j-\nu_i=\lambda_j - \lambda_i \ge {\frac{c}{\sqrt2}(|j|-|i|)} \ge \frac{c}{\sqrt2}. 
\end{equation} 

From \eqref{elle.c2} and \eqref{diff:aux1},  for $\e$ and $r$ small enough and $c$ large enough, we have
\begin{align} 
    \nonumber 
    \big|\jbs{\omega,k} + \jbs{\Omega,\ell}\big| 
        &\geq |\lambda_j-\lambda_i|- |\jbs{\nu,k}| - C_1 ( \e + r^{4/3})|k|_1 \\ 
    \label{eq:L0-I2I3} 
        &\ge \frac{c}{\sqrt{2}} - C_2|k|_1 - C_1 ( \e + r^{4/3})|k|_1 \geq \left( \frac{c}{2}-C_3 |k|_1 \right),
\end{align} 
for some constant $C_1,C_2,C_3>0$. Then, for any $\ell \in \ccS_k^{(5)}$, by taking $|k|_1 \leq \kappa_* c$, with $\kappa_* \leq 1/(4C_3)$, we obtain Item \ref{non.gauge.2} provided that $\alpha$ is small enough. 

\medskip 

{Consider now $\ell \in \ccS_k^{(4)}$, we prove Item \ref{I1e2.tesi} with $\cL=0$.}  In the case \((i,j) \in I_2\) we have $|j| \leq 2c^3$, so that $c \geq |j|^{1/3}/2^{1/3}$. Then the estimate \eqref{eq:L0-I2I3} holds true and gives
\begin{equation}\label{aux:3item} 
    \big|\jbs{\omega,k} + \jbs{\Omega,\ell}\big| \geq \frac{|j|^{1/3}}{2^{4/3}}-C_3|k|_1.
\end{equation} 
Hence, for any $|k|_1< |j|^{1/3}/(4C_3)$ the resonant set $\cR^\theta_{k \ell}(\alpha)$ is empty, provided that $\alpha$ is small enough. It follows that, if $\cR^\theta_{k \ell}(\alpha)\neq \emptyset$ then $|j|  \lesssim |k|_1^3$, and therefore $|i| \lesssim |k|_1^3.$

\smallskip 

Consider now $(i,j) \in I_1$. By the first of \eqref{C.56} we have
\begin{equation}\label{diff:aux2} 
    \lambda_j - \lambda_i \geq \frac{j^2 - i^2}{2\sqrt5} \geq { \frac{|i|+ |j|}{2\sqrt5}}. 
\end{equation} 

Thus, by \eqref{elle.c2} and \eqref{diff:aux2} for $c$ large enough, and $\e$, $r$ small enough, we obtain
\[
    \big|\jbs{\omega,k} + \jbs{\Omega,\ell}\big| \geq {\frac{|i|+|j|}{2\sqrt5}}- C_3 |k|_1,
\] 
with $C_3>0$ obtained as in \eqref{eq:L0-I2I3}. Therefore, provided $\alpha$ is small enough, if $\cR^\theta_{k \ell}(\alpha)\neq\emptyset$ then $|i| \lesssim |k|_1$ and $|j| \lesssim |k|_1$. 

\medskip 
\noindent\textbf{Case \(\cL \neq 0\).} 
{From \eqref{aux:C3}, we can directly prove Item \ref{I1e2.tesi} for $\cL\neq 0$. 
For every ${(i,j)} \in I_1 \cup I_2$, we have again $|j| \leq 2c^3$, and the estimate \eqref{aux:C3} takes the form
\begin{equation}\label{aux:3item} 
    \big|\jbs{\omega,k} + \jbs{\Omega,\ell}\big| \geq \frac{|j|^{1/3}}{2^{1/3}}-C_3|k|_1.
\end{equation} 
Reasoning as we did for \eqref{aux:3item}, if $\cR^{\theta}_{k \ell} \neq \emptyset$ then $|i| \lesssim |k|_1^3$ and $|j| \lesssim |k|_1^3$. Thus, we have $\# \{ \ell \in \ccS_k^{(4)} \mid \cR_{k \ell}^{\theta} \neq \emptyset \} \lesssim |k|_1^6$,  which combined with \eqref{singl:set:mes1} gives \eqref{I1e2.meas.1}.}

\end{proof}

To conclude this section, in the following lemma we state the measure estimate of the union of the resonant set over $\ccS_k^{(5)}$ (see \eqref{SK3}). Essentially, we follow the line of the proof of measure estimates in \cite{poschel1996kam} (Section 5 Lemma 8) {and Kuksin \cite{Kuksin2006} (Section 4)}, combined with Lemma \ref{non.gauge.1} and Item \ref{non.gauge.2} of Lemma \ref{I1e2}.
\begin{lemma}\label{casoI3}
  Let $\theta_0\in [0, 1)$ and $\theta_1 \in [0,1-\theta_0)$. There
      exist $h_*,\e_*,\varsigma_*, r_*, \alpha_*>0$ such that for any $h \in (0,h_*)$,
      $\e \in (0,\e_*)$, $\alpha \in (0,\alpha_*)$ and $r\in(0,r_*)$, the following holds. If
      $|\Delta|_{\infty,1-\theta_0} \lesssim \e$ (see \eqref{norma:Omega}) and $\e
      \alpha^{-1}<\varsigma_*$,  then
  \begin{equation}\label{est:SK4}
\left|\bigcup_{ \ell \in \ccS_k^{(5)}}\cR^{\theta_1}_{k\ell}(h;\alpha, \Delta,\delta)\right|\sleq
\frac{\alpha^{\frac{2}{3-\theta_1}}}{ \langle k\rangle^{\tau_{(\theta_1,\theta_0)}}}r^{\frac{4}{3}(N-1)}, \qquad \forall \, k \in \mathbb{Z}^{\cS},
  \end{equation}
with $\tau_{(\theta_1,\theta_0)} \vcentcolon=\frac{ \tau(1-\theta_0)-5(3-\theta_1)}{(3-\theta_1)}$ and $\ccS_k^{(5)}$ defined in \eqref{SK3}.
  \end{lemma}
\begin{proof}
The case $k=0$ is covered by Lemma \ref{remark:trivial}, then let $k \in \mathbb{Z}^{\cS} \setminus \{0\}$. Recalling Remark \ref{sup:1}, without loss of generality, we prove the statement only for $\ell \in \ccS_k^{(5)}$ satisfying $j>i>0$, $\ell_j=1$, and $\ell_i=-1$.

Denote $m\vcentcolon= j-i$. Recalling the estimate \eqref{asymp}, if $(i,j) \in I_3$ we have
\begin{equation}\label{asy:aux}
\bigg|\frac{ \Omega_{0i}- \Omega_{0j}}{ cm} -1\bigg|_{\infty}= \mathcal{O}(1/\tw_i^2).
\end{equation}
By the definition of $\Omega$ in \eqref{omegaprima}, together with the assumption $|\Delta|_{\infty,1-\theta_0} \lesssim \e$, we have 
\begin{equation*}
\begin{aligned}
\mR_{k\ell}^{\theta_1}(\alpha)&=\bigg\{ \xi \in \Xi_0 \mid  \left| \jbs{k,\omega}+ \Omega_j-\Omega_i \right| <  \frac{ \alpha}{\jbs{k}^\tau \tw_i^{\theta_1}} \bigg\}&\\
&\subseteq Q_{kmi}\vcentcolon=  \bigg\{  \xi \in \Xi_0 \mid \left| \jbs{k,\omega}+ cm \right| \lesssim \frac{ \alpha }{\jbs{k}^\tau \tw_i^{\theta_1}}+ \frac{cm}{\tw_i^2}+\frac{\e}{\tw_i^{1-\theta_0}} \bigg\}.&
\end{aligned}
\end{equation*}
Moreover, by definition we obtain
\[
Q_{kmi} \subseteq Q_{kms}, \qquad \forall \, |i|>|s|.
\]
Hence, for any fixed $k \in \mathbb{Z}^\cS$, $m\in \Z\setminus\{0\}$ and $i_0>c^3$ we have
\begin{align}
\nonumber
\bigg| \bigcup_{\substack{i,j \in I_3, \\ j-i=m}}\mR_{k\ell}^{\theta_1}(\alpha) \bigg| &\leq \bigg| \bigcup_{c^3<i<i_0} \mR_{k\ell}^{\theta_1}(\alpha) \bigg|+ \big| Q_{kmi_0} \big|& \\
\nonumber
&\lesssim r^{\frac{4}{3}(N-1)} \bigg[\left(\frac{\alpha m}{\jbs{k}^\tau} \sum_{0 < i \leq i_0} \frac{1}{\tw_{i}^{\theta_1}}\right)+\frac{cm}{\tw_{i_0}^2}+\frac{\e}{\tw_{i_0}^{1-\theta_0}}\bigg]& \\
\label{sommamis}
& \lesssim c^{3} m r^{\frac{4}{3}(N-1)} \bigg[ \frac{\alpha\, i_0^{1-\theta_1}}{\jbs{k}^\tau} +\frac{1}{i_0^2}+\frac{\e}{i_0^{1-\theta_0}} \bigg], &
\end{align}
where we used Fubini theorem on $Q_{kmi}$, exactly as in Lemma \ref{singole:stim}, and the second bound of \eqref{bound:peso}.

\smallskip

Now we choose $i_0$ in order to minimize \eqref{sommamis}. We start by looking for $i_1$ such that in the r.h.s. of \eqref{sommamis} the first and the second terms are equal, 
which is $i_1= \jbs{k}^{\frac{\tau}{3-\theta_1}}/\alpha^{\frac{1}{3-\theta_1}}$ in place of $i_0$. On the other hand, in the r.h.s of \eqref{sommamis} we can impose that the first term is equal to the third one by using the hypothesis $\e \alpha^{-1} < \varsigma_*$, with $\varsigma_*$ small enough, and replacing $i_0$ with 
$i_2=\jbs{k}^{\frac{\tau}{2-\theta_1-\theta_0}}$.
{By taking $i_0=\max(i_1,i_2)$, the estimate \eqref{sommamis} becomes 
\begin{equation*}
\bigg| \bigcup_{\substack{i,j \in I_3, \\ j-i=m}}\mR_{k\ell}^{\theta_1}(\alpha) \bigg| \lesssim c^3 m \bigg( \frac{\alpha\, i_1^{1-\theta_1}}{\jbs{k}^\tau} +\frac{\alpha\, i_2^{1-\theta_1}}{\jbs{k}^\tau} \bigg) r^{\frac{4}{3}(N-1)} \lesssim c^3 m \bigg( \frac{\alpha^{\frac{2}{3-\theta_1}}}{\jbs{k}^{\frac{2\tau}{3-\theta_1}}} + \frac{\alpha}{\jbs{k}^{\frac{\tau(1-\theta_0)}{2-\theta_1-\theta_0}}} \bigg) r^{\frac{4}{3}(N-1)}.
\end{equation*}
By taking $\theta_0 \in [0,1)$ and $\theta_1 \in [0,1-\theta_0)$ we have $\tfrac{2}{(3-\theta_1)}< 1$, so that the biggest numerator is the one of the first term, provided that $\alpha$ is small enough. Concerning the denominator, for $\tau >0$ we have that
\[
\frac{\tau(1-\theta_0)}{3-\theta_1}< \frac{2\tau}{3-\theta_1}, \qquad \frac{\tau(1-\theta_0)}{3-\theta_1}< \frac{\tau(1-\theta_0)}{2-\theta_1-\theta_0}.
\]
Then, we can replace both exponents in the denominators by $\tfrac{\tau(1-\theta_0)}{(3-\theta_1)}$, obtaining 
\begin{equation*}
\bigg| \bigcup_{\substack{i,j \in I_3, \\ j-i=m}}\mR_{k\ell}^{\theta_1}(\alpha) \bigg| \lesssim r^{\frac{4}{3}(N-1)} c^3 m\frac{\alpha^{\frac{2}{3-\theta_1}}}{ \langle k\rangle^{\frac{\tau(1-\theta_0)}{(3-\theta_1)}}}.
\end{equation*}
We restrict our estimate to the case $c \lesssim |k|_1$, since for $|k|_1 \lesssim c$ the resonant sets are empty by Lemma \eqref{non.gauge.1} (case $\mathcal{L}\neq 0$) and Item \eqref{non.gauge.2} of Lemma \ref{I1e2} (case $\mathcal{L}= 0$)}.

 Finally, from the zero momentum condition \eqref{Indici:M} we obtain $j-i=m \lesssim
|k|_1$ so that summing over $m$ and remarking that $c^3 \lesssim |k|_1^3$, we have the thesis with
\begin{equation}\label{def:tau}
\tau_{(\theta_0,\theta_1)}= \frac{\tau(1-\theta_0)-5(3-\theta_1)}{(3-\theta_1)}.
\end{equation}
\end{proof}

We can summarize the main result of this section \ref{nontrivial} in the following Lemma, which follows directly from Lemmata \ref{S12}, \ref{I1e2} and \ref{casoI3}.
\begin{lemma}\label{Smenofin}
  Let $\theta_0\in [0, 1)$ and $\theta_1 \in [0,1-\theta_0)$. There
      exist $h_*,\e_*,\varsigma_*, r_*, \alpha_*>0$ such that for any
      $\e \in (0,\e_*)$, $r\in(0,r_*)$, $\alpha \in (0,\alpha_*)$ and 
      $h \in (0,h_*)$, the following holds. If
      $|\Delta|_{\infty,1-\theta_0} \lesssim \e$ and $\e
      \alpha^{-1}<\varsigma_*$, then 
 \begin{equation}\label{est:Caso1}
\left|\bigcup_{ \ell \in \ccS_k^{-}}\cR^{\theta_1}_{k\ell}(h;\alpha, \Delta,\delta)\right|\sleq
\frac{\alpha^{\frac{2}{3-\theta_1}}}{ \langle k\rangle^{\tau_{(\theta_1,\theta_0)}}}r^{\frac{4}{3}(N-1)}, \qquad \forall \, k \in \mathbb{Z}^{\cS},
  \end{equation}
with $\tau_{(\theta_1,\theta_0)}=\frac{ \tau(1-\theta_0)-5(3-\theta_1)}{(3-\theta_1)}$ and $\ccS_k^{-}$ defined in \eqref{SA}.
\end{lemma}

\subsection{Estimate of the {union over $\ccS_k^+$}}\label{nontrivial:1} 

{As we did in Section \ref{nontrivial}, we are going to study $\ccS_k^+$, introduced in \eqref{SA1}, by considering different subcases. We denote 
\begin{align}
\label{SA5}
&\ccS_k^{(6)} \vcentcolon = \left\{\ell \in \ccS_k^{+} \mid  \cL=0 \, \textnormal{ or } \, \sgn(\cL)= \sgn(\ell_i)= \sgn(\ell_j) \right\},& \\
\label{SA6}
&\ccS_k^{(7)} \vcentcolon= \{ \ell \in \ccS_k^{+} \mid \sgn(i)=\sgn(j), \quad \sgn(\cL)= -\sgn(\ell_i)=-\sgn(\ell_j) \},&  \\
\label{SA7}
&\ccS_k^{(8)} \vcentcolon= \{ \ell \in \ccS_k^{+} \mid \sgn(i)=-\sgn(j), \quad \sgn(\cL)= -\sgn(\ell_i)=-\sgn(\ell_j) \},&
\end{align}
 which satisfy $\ccS_k^+= \ccS_k^{(6)} \cup \ccS_k^{(7)} \cup \ccS_k^{(8)}.$}

\subsubsection{Estimate of the union over $\ccS_k^{(6)}$ and $\ccS_k^{(7)}$ }\label{fac:2}

{We begin by studying $\ccS_k^{(6)}$ (see \eqref{SA5}) which corresponds to the situation where the first and third terms in \eqref{elle.c2} do not cancel each other.
\begin{lemma}\label{Gauge:2piu}
There exist $h_*,\e_*,r_*, \alpha_*>0$ such that for any 
$h \in (0,h_*)$, $\e \in (0,\e_*)$, $r\in(0,r_*)$, and $\alpha \in (0,\alpha_*)$ the following holds. 
For any $\theta \in [0,1)$, and $k \in \mathbb{Z}^{\cS}$ we have
\begin{align}
  \label{stimS6}
\bigg| \bigcup_{ \ell \in \ccS_k^{(6)}} \cR_{k \ell}^{\theta}(\alpha) \bigg| \leq \frac{\alpha}{\jbs{k}^{\tau-1}} r^{\frac{4}{3} (N-1)}.
\end{align}
Furthermore, if $\cL \neq 0$ there exists $\kappa_*>0$ such that 
\begin{equation}
  \label{2piu.res.41}
 \left|\langle\omega,k\rangle+ \jbs{\Omega, \ell} \right|\geq \alpha c^2, \qquad  \forall \, |k|_1 \leq \kappa_* c^{2}.
\end{equation}
\end{lemma}}

\begin{proof}
The case $k=0$ is covered by Lemma \ref{remark:trivial}, then let $k \in \mathbb{Z}^{\cS} \setminus \{0\}$. From \eqref{elle.c2} we have
  \begin{align}
\nonumber
    \left|\jbs{\omega,k}+\jbs{\Omega,\ell} \right| &\geq ( |\cL| c^2+\nu_i+\nu_j )-|\nu\cdot k|-C_1 (\e+r^{4/3}) |k|_1 & \\
    \label{2piu.facile}
&\geq  (|\cL| c^2+\nu_i+\nu_j )-(C_2+C_1 (\e+r^{4/3}))|k|_1,& \\
    \label{2piu.facile.00}
&\geq  (\nu_i+\nu_j)-(C_2+C_1 (\e+r^{4/3}))|k|_1, &
    \end{align}
for some positive constants $C_1,C_2$.

It follows easily from \eqref{def:h} that
\begin{equation}\label{nuj3}
\nu_n \geq \frac{1}{3} \, |n|, \qquad \forall \, c \geq 1, \, \, \forall n \in \cS^c.
\end{equation}
Then, from \eqref{2piu.facile.00}, provided $\e$ and $r$ are small enough and $c$ large enough, there exists $C_3>0$ such that
\begin{equation}\label{C4}
\left|\langle\omega,k\rangle+\jbs{\Omega, \ell} \right|  \geq \frac{1}{3}(|i|+|j|)-C_3|k|_1.
\end{equation}
It follows that if $\cR_{k \ell}^{\theta}(\alpha)\not=\emptyset$ then $|i| \lesssim |k|_1$ and $|j| \lesssim |k|_1$, which combined with \eqref{singl:set:mes1} gives \eqref{stimS6}, provided that $\alpha$ is small enough.

\medskip

We now prove the second claim. Since $\cL>0$ and $\nu_j+\nu_i \geq 0$, from the estimate \eqref{2piu.facile} we have
\[
\left|\langle\omega,k\rangle+ \jbs{\Omega, \ell} \right|  \geq c^2-C_4|k|_1,
\]
with $C_3$ as in \eqref{C4}. By taking $|k|_1 \leq \kappa_* c^2$, with $\kappa_*=1/(2C_3)$, we obtain \eqref{2piu.res.41}, provided that $\alpha$ is small enough.
\end{proof}

{We consider now the union over $\ccS_k^{(7)}$ (see \eqref{SA6}), for which the zero momentum condition suffices to estimate the measure and the divisors.}

{\begin{lemma}\label{Basta:mom}
There exist $h_*,\e_*,r_*, \alpha_*>0$ such that for any 
$h \in (0,h_*)$, $\e \in (0,\e_*)$, $r\in(0,r_*)$, and $\alpha \in (0,\alpha_*)$ the following holds. 
For any $\theta \in [0,1)$, and $k \in \mathbb{Z}^{\cS}$ we have
\begin{align}
  \label{2piu.res.4.1.0}
\bigg| \bigcup_{ \ell \in \ccS_k^{(7)}} \cR_{k \ell}^{\theta}(\alpha) \bigg| \leq \frac{\alpha}{\jbs{k}^{\tau-1}} r^{\frac{4}{3} (N-1)}.
\end{align}
Furthermore, there exists $\kappa_*>0$ such that 
\begin{equation}
  \label{2piu.res.41.0}
 \left|\langle\omega,k\rangle+ \jbs{\Omega, \ell} \right|\geq \alpha c^2, \qquad  \forall \, |k|_1 \leq \kappa_* c^{2}.
\end{equation}
\end{lemma}
\begin{proof}
The estimate \eqref{2piu.res.4.1.0} follows by the zero momentum condition \eqref{Indici:M} which implies $|i| \lesssim |k|_1$ and $|j| \lesssim |k|_1$.
Now we prove \eqref{2piu.res.41.0}. First we remark that $|\nu_j+ \nu_i| \lesssim |k|_1^2$ by the fact that $(\nu_i+ \nu_j) \lesssim i^2+j^2$ and by the zero momentum condition. Then, applying \eqref{elle.c2} we obtain the estimate
\[
\begin{aligned}
| \jbs{\omega,k}+\jbs{\Omega, \ell}|  &\geq c^2 |\cL| -(\nu_i+\nu_j) -| \jbs{\nu,k}|- C_1(\e+r^{4/3})|k|_1&\\
& \geq c^2 -C_2 |k|_1^2-C_3(\e+r^{4/3})|k|_1&
\end{aligned}
\]
for some positive constants $C_1, C_2, C_3$. Then, by taking $\e, r$ small enough and $c$ large enough, there exists $C_4>0$ such that 
\[
| \jbs{\omega,k}+\jbs{\Omega, \ell}| \geq c^2- C_4|k|_1^2.
\]
The bound \eqref{2piu.res.41.0} follows by taking $|k|_1 \leq c/(\sqrt{2C_4})$, provided $\alpha$ is small enough.
\end{proof}}

\subsubsection{Estimate of the union over $\ccS_k^{(8)}$} \label{non:fac}
We come to $\ell \in \ccS_k^{(8)}$, defined in \eqref{SA7}. First we point out that the momentum condition \eqref{Indici:M} gives 
\begin{equation}\label{momnto:aux}
| |j|-|i| |= \bigg| \sum_{n \in \cS} n k_n \bigg| \lesssim |k|_1,
\end{equation}
We have the following lemma.
\begin{lemma}\label{NonGauge:2piu}
{There exist $h_*,\e_*,r_*, \alpha_*>0$ such that for any 
$h \in (0,h_*)$, $\e \in (0,\e_*)$, $r\in(0,r_*)$, and $\alpha \in (0,\alpha_*)$ the following holds. 
For any $\theta \in [0,1)$, and $k \in \mathbb{Z}^{\cS}$ we have
\begin{align}
  \label{2piu.res.1}
\bigg| \bigcup_{ \ell \in \ccS_k^{(8)}} \cR_{k \ell}^{\theta}(\alpha) \bigg| \leq \frac{\alpha}{\jbs{k}^{\tau-5}} r^{\frac{4}{3} (N-1)}.
\end{align}
Furthermore, there exists a constant $\kappa_{*}>0$, such that if $\supp(\ell) \subseteq [-c/2,c/2]$ and $|k|_1 \leq \kappa_* c^{1/2}$, then the following estimate holds
\begin{equation}
  \label{2piu.res.2}
 \left|\langle\omega,k\rangle+ \jbs{\Omega, \ell} \right|\geq \alpha c^2.
\end{equation}}
\end{lemma}
\begin{proof}
The case $k=0$ is already covered by Lemma~\ref{remark:trivial}, then we consider $k \in \mathbb{Z}^{\cS} \setminus \{0\}$. 
It suffices to prove \eqref{2piu.res.1} under the restriction
\begin{equation}\label{restriz:1}
\ell \in \ccS_k^{(8)}, \quad \cL > 0, \quad i>0, \, j<0, \quad \ell_i=\ell_j=-1,
\end{equation}
and we set $l \vcentcolon=-j$. We may further assume $l \neq i$, while the case $i=l$ can be handled in a similar, and simpler, way.

 We point out that, for 
  $a>0$, by \eqref{elle.c2}, for $\e$ and $r$
  small enough, we have 
\begin{equation}\label{implication}
|\jbs{ \omega, k}- (\Omega_i+\Omega_l)| \leq a,\ \Longrightarrow
\left| c^2 \cL -(\nu_i+\nu_l) \right|\sleq a+|k|_1.
\end{equation}
Dividing by $c^2$ we obtain 
\begin{equation}
  \label{aux:11}
\left|\cL-\left(\frac{\nu_i}{c^2}+\frac{\nu_l}{c^2}\right)\right|\sleq
\frac{a}{c^2}+\frac{|k|_1}{c^2}.
\end{equation}
We will first consider $a$ of order one and then of order $c^2$ to prove the two claims \eqref{2piu.res.1} and \eqref{2piu.res.2}. Therefore, in any case we assume that
\begin{equation}\label{cond:a}
a \leq c^2.
\end{equation}

\medskip 

The advantage of expression \eqref{aux:11} is that $\frac{\nu_i}{c^2}$ has an
explicit expression that allows to study the inequality \eqref{aux:11}. Indeed,
defining 
\begin{equation}\label{def:f}
f(x) \vcentcolon= \frac{x}{1+\sqrt{1+x}}= \sqrt{x+1}-1,
\end{equation}
we have $\frac{\nu_i}{c^2}=f(i^2/c^2)$. For future reference we note a couple
of properties of the function $f$. First, it is a monotonically
increasing function and, on $\R^+$, its Lipschitz constant is globally bounded by
$1$. Secondly, its inverse $x=f^{-1}(y)$ is given by
\begin{equation}
  \label{inversa}
x=f^{-1}(y)=y(y+2), 
\end{equation}
Furthermore, for any $\bar x,\tilde x>0$, denoting $\bar{y}=f(\bar{x})$ and $\tilde{y}=f(\tilde{x})$ we have
\begin{equation}
  \label{lip}
\bar x-\tilde x=(\bar y-\tilde y)(\bar y+\tilde y+2).
\end{equation}

Denote $x_i \vcentcolon= i/c$ and $x_l \vcentcolon= l/c$. Then
\begin{align}
  \label{eq.2piu}
\left|\cL-\left(f(x_i^2)+f(x_l^2)\right)\right|\sleq
\frac{a}{c^2}+\frac{|k|_1}{c^2},\quad \left|x_l-x_i\right|\sleq
\frac{|k|_1}{c},
\end{align}
where the first equation follows from \eqref{aux:11}, while the second inequality is the momentum conservation \eqref{momnto:aux} divided by a factor $c$ . We are
going to use that to approximate $f(x_l^2)$ with
$f(x_i^2)$. To this end we need first an a priori estimate on
$f(x_i^2)$, $f(x_l^2)$, 
$x_i$ and $x_l$. This is obtained by remarking that, since $f$ has
positive values, the first inequality of \eqref{eq.2piu} implies 
\begin{equation}
  \label{apriori.1}
\max\left\{f(x_i^2),f(x_l^2)\right\}\leq f(x_i^2)+f(x_l^2)\sleq \cL
+\frac{a}{c^2}+\frac{|k|_1}{c^2}\sleq\cL+\frac{|k|_1}{c^2}\sleq |k|_1,
  \end{equation}
where we use \eqref{cond:a} to obtain $a/c^2 \lesssim \cL$. According to \eqref{inversa}, by \eqref{apriori.1} we obtain
\begin{equation}
  \label{apriori}
x_i,x_l\sleq |k|_1.
\end{equation}
By using \eqref{apriori}, the bound on the Lipschitz constant of $f$ and the second of  \eqref{eq.2piu}, we obtain 
$$
|f(x_l^2)-f(x_i^2)|\leq |x_l^2-x_i^2|= |x_l-x_i|(x_l+x_i) \sleq
\frac{|k|_1^2}{c}.
$$
The previous estimate, combined with the first of \eqref{eq.2piu}, implies
\begin{equation}
  \label{2piu.1}
\left|\frac{\cL}{2}-f(x_i^2)\right|\sleq
\frac{a}{c^2}+\frac{|k|_1}{c^2}+ \frac{|k|_1^2}{c}\sleq
\frac{a}{c^2}+\frac{|k|_1^2}{c}.
\end{equation}

 We define $x_{\cL} \vcentcolon= \sqrt{\frac{\cL}{2} \left(\frac{\cL}{2}+2
   \right)}$, that satisfies $f(x_{\cL}^2)=\cL/2$ by \eqref{inversa}. Then, the estimate \eqref{2piu.1} reads
 $$
\left|f(x_{\cL}^2)-f(x_i^2)\right|\sleq \frac{a}{c^2}+
\frac{|k|_1^2}{c}.
$$
We apply \eqref{lip} to $x_{\cL}$ and $x_i$, combined with the second-last estimate in \eqref{apriori.1}, obtaining
$$
\left|x_\cL^2-x_i^2\right|\sleq \left(\frac{a}{c^2}+
\frac{|k|_1^2}{c} \right)\left(\cL+\frac{|k|_1}{c^2}\right).
$$
Multiplying by $c^2$ we have
\begin{align}
  \label{finale.2piu.10}
\left|c^2x_\cL^2-i^2\right|\sleq
\left(a+c|k|_1^2\right)\left(\cL+\frac{|k|_1}{c^2}\right).
\end{align}
To obtain a bound on $i$ we can rewrite \eqref{finale.2piu.10} as
\begin{align}
  \label{finale.2piu.1}
cx_{\cL} \left|cx_{\cL}-i\right|\leq (cx_{\cL}+i )
\left|cx_{\cL}-i\right|=\left|c^2x_\cL^2-i^2\right|\sleq
\left(a+c|k|_1^2\right)\left(\cL+\frac{|k|_1}{c^2}\right).
\end{align}
 Finally, dividing by $c x_{\cL}$ we have
\begin{align}
  \label{finale.2piu}
\left|cx_{\cL}-i\right|\sleq
\left(\frac{a}{c}+|k|_1^2\right)\left(\frac{\cL}{x_{\cL}}+\frac{|k|_1}{c^2x_{\cL}}\right)
\sleq\left(\frac{a}{c}+|k|_1^2\right) \left(1+\frac{|k|_1}{c^2}\right).
\end{align}
The same estimate for $l$ follows by doing all the computation replacing $i$ with $l$.

\smallskip

{We can now prove \eqref{2piu.res.1}. By recalling the definition of $\cR_{k \ell}^{\theta}(\alpha)$ in \eqref{res1}, for $\tau>0$ e $\theta \in [0,1)$ we set
\[
a=\frac{\alpha}{\langle 
  k\rangle^\tau \tw(\ell) ^{\theta}}.
\]
We point out that $a<\alpha$ by the first estimate in \eqref{bound:peso}.
Then \eqref{finale.2piu} becomes
$$
\left|cx_{\cL}-s\right|\sleq \left|k\right|^3, \quad s=i, l,
$$
provided that $\alpha$ is small enough. Thus, the cardinality of the $i$'s and the $j$'s in the above interval is bounded by $|k|_1^3$. By \eqref{implication} we obtain \eqref{2piu.res.1}.}

\medskip

{Finally, we prove \eqref{2piu.res.2}, under the assumptions \eqref{restriz:1}, by using a contradiction agument. Let $a=\alpha c^2$ and $|k|_1 \leq \kappa_* c^{1/2}$, with $\alpha, \kappa_* \in (0,1)$. Combining this bounds with \eqref{implication} and \eqref{finale.2piu}, we obtain
$$
|\jbs{ \omega, k}- (\Omega_l+\Omega_i)| \leq \alpha c^2 \implies \left|cx_{\cL}-s\right| \lesssim (\alpha+\kappa_*^2)c, \quad s=i, l.
$$
Then there exists $C>0$ such that
$$
s \geq  c\left(x_{\cL}- C(\alpha+\kappa_*^2)\right)> \frac{c}{2}, \quad s=i, l, 
$$
where we use that $x_{\cL} \geq 1$, and we impose $ (\alpha+\kappa_*^2) < 1/(2C)$. This is a contradiction because $\supp(\ell)\subset [-c/2, c/2]$.
It follows that, by taking $\alpha$ and $\kappa_*$ small enough, uniformly in $c$, we obtain \eqref{2piu.res.2}. }
\end{proof}

We can state the general estimate for the union resonant set on $\ccS_k^+$, which follows directly from Lemmata \ref{Gauge:2piu}, \ref{Basta:mom} and \ref{NonGauge:2piu}.
\begin{lemma}\label{Spiu:fine}
  Let $\theta_0\in [0, 1)$ and $\theta_1 \in [0,1-\theta_0)$. There
      exist $h_*,\e_*,\varsigma_*, r_*, \alpha_*>0$ such that for any
      $\e \in (0,\e_*)$, $r\in(0,r_*)$, $\alpha \in (0,\alpha_*)$ and 
      $h \in (0,h_*)$, the following holds. If
      $|\Delta|_{\infty,1-\theta_0} \lesssim \e$ and $\e
      \alpha^{-1}<\varsigma_*$, then we have 
 \begin{equation}\label{est:Caso2}
\left|\bigcup_{ \ell \in \ccS_k^{+}}\cR^{\theta_1}_{k\ell}(h; \alpha, \Delta,\delta)\right|\sleq
\frac{\alpha^{\frac{2}{3-\theta_1}}}{ \langle k\rangle^{\tau_{(\theta_1,\theta_0)}}}r^{\frac{4}{3}(N-1)}, \qquad \forall k \in \mathbb{Z}^{\cS},
  \end{equation}
with $\tau_{(\theta_1,\theta_0)}=\frac{ \tau(1-\theta_0)-5(3-\theta_1)}{(3-\theta_1)}$ and $\ccS_k^{+}$ defined in \eqref{SA1}.

{Furthermore, there exists a constant $\kappa_{*}>0$, such that if $\supp(\ell) \subseteq [-c/2,c/2]$ and $|k|_1 \leq \kappa_* c^{1/2}$, then the following estimate holds
\begin{equation*}
 \left|\langle\omega,k\rangle+ \jbs{\Omega, \ell} \right|\geq \alpha c^2.
\end{equation*}}
\end{lemma}
{We conclude Section \ref{meas.est} remarking that the Main Theorem \ref{measure_estimate} follows by the decomposition given in \eqref{S0}-\eqref{SA1} combined with Lemmata \ref{S0}, \ref{caso:trivial}, \ref{non.gauge.1} , \ref{Smenofin} and \ref{Spiu:fine}.}

\section{ KAM theorem}

\subsection{The KAM algorithm}

Following 
\cite{poschel1996kam} (see Section 6), we first perform a preliminary KAM step by imposing the
Melnikov conditions \eqref{dioph:NLS} and \eqref{dioph} with a positive parameter $\theta$ and apply Theorem \ref{measure_estimate} with
$\theta_0=0$ and $\theta_1= \theta \in(0,1)$. As a result we will solve the cohomological
equation losing regularity with respect to $\beta$ in $\mathcal{P}^{a, p, \beta}$ (recall $\cP^{a,p,\beta}$ in \eqref{phase:sptrasf}). Then we will make infinitely many steps of
normalization and at each step we will impose the Melnikov conditions
 \eqref{dioph:NLS} and \eqref{dioph} with $\theta=0$, so that we will not loose further regularity on $\beta$ . We will thus use the
estimate of Theorem \ref{measure_estimate} with $\theta_1=0$ and
$\theta_0=2\theta$. This
is needed in order to obtain an asymptotically full measure set of invariant tori.

\smallskip

We now fix the following parameters
\begin{equation}
  \label{scelte.theta}
s_0\vcentcolon= 2, \quad \sigma_0\vcentcolon= \frac{1}{40}s_0,\quad \theta\in(0,1/2), \quad r_0 \in (0,r_*),
  \end{equation}
 with $r_*$ introduced
in Lemma \ref{Newstime}. Then, we choose a
positive parameter $\alpha_0$ fulfilling
\begin{equation}\label{alpha0}
 0<\alpha_0<r_0^{4/3},
\end{equation}
which at the end of the procedure
  will be taken to be a function of $r_0$. {We will keep track of the dependence of the constants on these parameters}.

{From the definitions \eqref{res2} and \eqref{res:setNLS}, we denote 
\[
 \mR^\theta(h;\alpha_0,0,0)\vcentcolon=  \bigcup_{k \in \mathbb{Z}^{\cJ}} \mR_k^{\theta}(h; \alpha_0,0,0), \qquad  \mR^{\tt NLS}(0;\alpha_0,0,0) \vcentcolon=  \bigcup_{k \in \mathbb{Z}^{\cJ}} \mR_k^{\tt NLS}(0;\alpha_0,0,0),
\]
{and then the initial non resonant set is denoted as}
\begin{equation}\label{Xi1}
\Xi_1\vcentcolon= \Xi_0 \setminus \big( \mR^\theta(h;\alpha_0,0,0) \cup
\mR^{\tt NLS}(0;\alpha_0,0,0)\big),
\end{equation}
(c.f. \eqref{Xi0}).}
For $0<\eta_0<1/9$, we also define two complex neighborhoods (recall \eqref{Dq})
\[
D_{q}^1\vcentcolon= D_{q}(s_0-5\sigma_0, \eta_0 r_0), \qquad D_{q}^0\vcentcolon= D_{q}(s_0,r_0),
\] 
{for both $q=p$ and $q=p-4$-}
 We have the following lemma .
\begin{lemma}\label{Lemma:step1}
  %By considering the two Hamiltonians $\eqref{H}$ and $\eqref{HNLS}$
   There
    exist $h_*>0$ and $\varrho_0>0$ s.t., if
    \begin{align*}
\frac{\e_0}{\alpha_0}<\varrho_0,
    \end{align*}
  the following holds. For any $h \in (0,h_*)$
there exist
two canonical changes of coordinates 
\[
\Phi_0, \Phi_0^{\tt NLS}:  D_{q}^1 \times \Xi_1 \to D_{q}^0,
\]
which are Lipschitz as functions of $\xi$ and analytic as functions of the
variables $(x, y, z, \bar{z})$ and they are such that 
\begin{align*}
  &(N_0+P_0)\circ\Phi_0=N_1+P_1,&
  \\
  &(N_0^{\tt NLS}+P_0^{\tt NLS})\circ\Phi_0^{\tt NLS}=N_1^{\tt NLS}+P_1^{\tt NLS},&
\end{align*}
where
\begin{align*}
N_1=\langle\omega_1,y\rangle+\langle\Omega_1,z\bar{z}\rangle , \qquad N_1^{\tt NLS}=\langle\omega_1^{\tt NLS},y\rangle+\langle\Omega_1^{\tt NLS},z\bar{z}\rangle,
\end{align*}
and the frequency maps are lipeomorphisms on $\Xi_0$ of the form  (recall \eqref{Oomega}, \eqref{OomegaNLS}, \eqref{OomegaR})
\[
\begin{aligned}
&\omega_1(\xi,h)\vcentcolon= \omega_0(\xi,h)+\delta_1(\xi,h),& \qquad & \omega_1^{\tt NLS} (\xi)\vcentcolon= \omega_0^{\tt NLS}+\delta_1^{\tt NLS}(\xi),&
  \\
&\Omega_1(\xi,h)\vcentcolon= \Omega_0(\xi,h)+\Delta_1(\xi,h),& \qquad &\Omega^{\tt NLS}_1(\xi)\vcentcolon= \Omega_0^{\tt NLS}+\Delta_1^{\tt NLS}(\xi),& \\
&\omega_1^R(\xi,h)\vcentcolon= \omega_0^R(\xi,h)+\delta_1(\xi,h)-\delta_1^{\tt NLS}(\xi),& \qquad &\Omega_1^R(\xi,h)\vcentcolon= \Omega_0^R(\xi,h)+\Delta_1(\xi,h)-\Delta_1^{\tt NLS}(\xi),&
\end{aligned}
\]
 satisfying the following bounds
\begin{align}
\label{smallcond2.0}
&|\delta_1|_{\infty}^{\mL(\alpha_0)} \sleq \e_0,& \qquad &|\delta_1^{\tt  NLS}|_{\infty}^{\mL(\alpha_0)}  \sleq \e_0,&
\\
\label{smallcond2.1}
&|\Delta_1|_{\infty,1-2\theta}^{\mL(\alpha_0)} \sleq \e_0,& \qquad
&|\Delta_1^{\tt  NLS}|_{\infty}^{\mL(\alpha_0)} \sleq \e_0,& \\
\label{cond:OmegaR.1}
&|\delta_1-\delta_1^{\tt NLS}|_{\infty} \sleq h\e_0, & \quad &|(\Delta_1)_j-(\Delta_1^{\tt NLS})_j|_\infty \sleq  h j^4\e_0, \quad \forall j \in \cS^c.
\end{align}
Furthermore,
the new perturbation and the change of coordinates { admit a decomposition of the form}
\begin{align*}
 P_1=P_1^{\tt NLS}+P_1^{\mathcal{NG}}+P_1^R, \qquad\Phi_{0}=\Phi_{0}^{\tt NLS}+\Psi_{0}^R,&
\end{align*}
{with $P_1^{\ng}=\Pi_{\ng} P_1^{\ng}$ and the remainders defined by the above identities.} 
{By defining
$$
\eta_0\vcentcolon= \left(\frac{\e_0}{\alpha_0}\right)^{1/3},
$$}
the following estimates hold in $\Xi_1$ for $q=p$ and $q=p-4$
\begin{align}
  \label{lambro.1.0}
&\bn{X_{P_1}}^{\mL(\alpha_0),D_q^1}_{\eta_0 r_0,q, 1-2\theta}, \,
\bn{X_{P_1^{\ng}}}^{\mL(\alpha_0),D_q^1}_{\eta_0 r_0,q} , \, \bn{X_{P^{\tt NLS}_1}}^{\mL(\alpha),D_q^1}_{\eta_0 r_0,q}\sleq \frac{\e_0^{4/3}}{(\alpha_0)^{1/3}}, &
\\
\label{resti.6.0}
&
\bn{\uno-\Phi_0  }^{\mL(\alpha_0), D_q^1}_{r_0,q}\sleq
\frac{\e_0}{\alpha_0},\quad \bn{\uno-\Phi_0^{\tt NLS}  }^{\mL(\alpha_0), D_q^1}_{r_0,q}\sleq
\frac{\e_0}{\alpha_0},&
\\
\label{resti.6.1}
& \bn{X_{P_1^R}}^{D_p^1}_{\eta_0 r_0,p-4}\sleq
h\frac{\e_0^{4/3}}{(\alpha_0 )^{4/3}}, \, \quad 
\bn{\Psi_{0}^R}^{D_p^1}_{\eta_0 r_0,p-4} \sleq h \frac{\e_0}{\alpha_0^2}.&
\end{align}
 Finally, if $\tau$ in \eqref{res1} and \eqref{R:NLS}
   fulfills
   \begin{equation}
     \label{tau.KAM.1}
\frac{\tau}{3-\theta}-5>N,
     \end{equation}
then we have the estimate
\begin{equation}\label{first:mestim}
\frac{ \big| \Xi_0 \setminus \Xi_1 \big| }{ \big| \Xi_0 \big| } \leq C(\theta)\, \frac{ \alpha_0^{\frac{2}{3-\theta}}}{r_0^{4/3}},
\end{equation}
with $C(\theta)>0$ diverging as $\theta\to0$.
\end{lemma}

\begin{proof}
  Except for \eqref{cond:OmegaR.1} and \eqref{resti.6.1}
    this is the ``KAM step'' of Sections 3 and 6 of \cite{poschel1996kam} (see also \cite{berti2013kam} Section 5) . We
    start by recalling its scheme of proof.
We consider the Hamiltonian of the Klein-Gordon equation $H_0=N_0+P_0$
introduced in \eqref{H}. By Lemma \ref{solhomo} there exists a
Hamiltonian $G_1$, which is the solution of the cohomological equation (recall \eqref{media}, \eqref{taglio})
\[
\{N_0,G_1\}+\{P_0\}=[P_0],
\]
for $\xi\in \Xi_1$ (see \eqref{Xi1}).
By considering the change of coordinates $\Phi_{0}\vcentcolon= \phi_{G_1}^t \big|_{t=1}$, the new Hamiltonian is defined as
\[
H_1=H_0 \circ \Phi_{0}=N_1+P_1,
\]
where
\begin{equation}
  \label{P0quadro}
N_1\vcentcolon= N_0+[P_0]
  \end{equation}
and 
\begin{align}
\label{int}
P_1&= \int_0^1 \{ (1-t)[P_0]+t\{P_0\},G_1\} \circ \phi_{G_1}^t \diff t& \\
\label{diff}
&+ (P_0-\{P_0\})\circ \Phi_{0}.&
\end{align}
By using that
$
X_{P_0}: D_q(1,r_0) \to \cP^{a,q,1}
$ and the definition \eqref{Xi1} of $\Xi_1$,
by \eqref{GG} with $\sigma=5\sigma_0=\frac{5}{40}$ we have
\begin{equation}
\bn{X_{G_1}}^{\mL(\alpha_0),D_q^1}_{r,q, 1-2\theta} \lesssim
\frac{\e_0}{\alpha_0},\quad \bn{X_{G_1^R}}^{D_q^1}_{r,p-4} \lesssim
h\frac{\e_0}{\alpha_0^2}.
\end{equation}
In particular by \eqref{esisto.3} and \eqref{Lip:semnorm} we get \eqref{resti.6.0}.
We now remark that, by Lemma \ref{flowmap.perdo}, with $E=\cP^{a,p}$,
$\tilde E=\cP^{a,p-4}$, $X_1\vcentcolon= X_{G_1}$, $X_2\vcentcolon= X_{G^{\tt NLS}_1}$ we
have the second of \eqref{resti.6.1}. 
By proceeding essentially as in the proof of Theorem
\ref{Teorema0} (see also the proof of Equation (13) in
  \cite{poschel1996kam}) we have
\begin{equation}
\bn{X_{(\ref{int})}}^{\mL(\alpha_0),D_q^1}_{\eta_0 r_0,q, 1-2\theta} \lesssim \frac{\e_0^2}{\alpha_0 \eta_0^2}, 
\end{equation}
while \eqref{restoF} yields
\begin{equation}
\bn{X_{(\ref{diff})}}^{\mL(\alpha_0),D_q^1}_{\eta_0 r_0,q, 1-2\theta} \lesssim \eta_0 \e_0.
\end{equation}
By choosing $\eta_0=\e_0^{1/3}/\alpha_0^{1/3}$ the two estimates above have the same
order. Then we have the first estimate in \eqref{lambro.1.0}.

We prove now the second of \eqref{cond:OmegaR.1}, the first being
  simpler. {From equation \eqref{P0quadro} we have} (recall \eqref{PolR})
  \[
(\Delta_{1})_j=(P_{0})_{00\delta_j\delta_{j}},\quad (\Delta_{1}^{\tt
    NLS})_j=(P_{0}^{\tt NLS})_{00\delta_j\delta_{j}},
  \]
  so that
  \begin{equation}\label{ident:diff}
(\Delta_{1})_j-(\Delta_{1}^{\tt
    NLS})_j
  =(P_{0}^R)_{00\delta_j\delta_{j}}=\frac{1}{2\pi}\int_{\T^N}\frac{\partial}{\partial
  z_j}(X_{P^R_0})_{z_j}(x,0,0,0)dx.
  \end{equation}
  On the other hand, recalling the definition of $\delta_j$ in \eqref{Kronecker}, we have 
\begin{equation}\label{ident:1}
\frac{\partial}{\partial
  z_j}(X_{P^R_0})_{z_j}= d(X_{P^R_0})_{z_j}\delta_j.
\end{equation}
To estimate this quantity we remark that for all $v\in\cP^{a,p}$, we have
\begin{align*}
\left\|d(X_{P^R_0})_{z}(x,0,0,0)v\right\|_{a,p-4}&=
\left[\sum_{j}\jbs{j}^{2p-8}e^{2a|j|}\left|d(X_{P^R_0})_{z_j}(x,0,0,0)v\right|^2
  \right]^{1/2}&
\\
&\leq
\left\| d(X_{P^R_0})_{z}(x,0,0,0)
\right\|_{\cP^{a,p},\cP^{a,p-4}}\left\|v\right\|_{a,p}& \\
 &\leq
\frac{1}{r_0}\sup_{D_p(2,r)}\left\|(X_{P^R_0})_{z}(x,y,z,\bar z)
\right\|_{a,p-4} \left\|v\right\|_{a,p}&
\\
&\leq \left\|X_{P^R_0}
\right\|_{r_0,p-4}^{D_p(2,r)} \left\|v\right\|_{a,p}\sleq \e_0h
\left\|v\right\|_{a,p},&
\end{align*}
{where we used \eqref{2Stima} and \eqref{def:e}. This gives
\[
\jbs{j}^{p-4}e^{a|j|} \left|d(X_{P^R_0})_{z_j}\delta_j\right| \lesssim \left\|d(X_{P^R_0})_{z}(x,0,0,0)\delta_j \right\|_{a,p-4} \lesssim  \e_0 h
\left\|\delta_j\right\|_{a,p}\sleq \jbs{j}^p e^{|j|a} \e_0h.
\]
Then, the second of \eqref{cond:OmegaR.1} follows from the above estimate combined with \eqref{ident:diff} and \eqref{ident:1}.}
The estimates \eqref{smallcond2.1} and \eqref{smallcond2.0} are obtained similarly, by \eqref{Stima0} and \eqref{1Stima}. Finally we have to prove the first estimate in \eqref{resti.6.1}. { We point out that the function
$P_1^R$ to be estimated, mutatis mutandis is given by the sum of terms similar to
\eqref{*.1}-\eqref{**.2}. We now bound the vector
field of one of its terms, which is $\left\{P_0^{\tt
  NLS}+P_0^{\ng},G_1^R\right\}$. We have}
\begin{align*}
\bn{X_{\left\{P_0^{\tt
    NLS}+P_0^{\ng},G_1^R\right\}}}^{D^1_p}_{r_0,p-4} & \leq
 \bn{dX_{P_0^{\tt
  NLS}+P_0^{\ng}
  }}^{D_{p-4}^1}_{\cP^{a,p-4},\cP^{a,p-4}}\bn{X_{G^R_1}}^{D^1_p}_{r_0,p-4}&
   \\
&\quad +
  \bn{dX_{G_1^R
  }}^{D_p^1}_{\cP^{a,p},\cP^{a,p-4}} \bn{X_{P_0^{\tt
      NLS}+P_0^{\ng}}}^{D^1_p}_{r_0,p}&
  \\
 & \sleq \frac{1}{\sigma_0}\bn{X_{P_0^{\tt
  NLS}+P_0^{\ng}
  }}^{D^0_{p-4}}_{r_0,p-4}
  \bn{X_{G^R_1}}^{D^1_p}_{r_0,p-4}& \\
& \quad+\frac{1}{\sigma_0}
 \bn{X_{G_1^R
  }}^{D^0_p}_{r_0,p-4}\bn{X_{P_0^{\tt
      NLS}+P_0^{\ng}}}^{D^1_p}_{r_0,p}& 
\\
& \sleq
 \frac{1}{\sigma_0}\e_0 h\frac{\e_0}{\alpha_0},&
\end{align*}
where we used the Cauchy estimate to estimate the differential and the
fact that any point in $D^1_q$ is at a distance of order $\sigma_0$,
in the norm of the phase space (with parameter $r_0$), from the boundary of $D^0_q$. Finally
to get the desired estimate just use the fact that, for any vector
field $Y$, $\left\|Y\right\|_{\eta_0 r_0}\sleq \left\|Y\right\|_{r_0}/\eta_0^2$ holds pointwisely.

Finally, the measure estimate \eqref{first:mestim} follows from \eqref{meas.finalNLS} and \eqref{meas.final}.
\end{proof}
We now start the KAM iteration. In order to state the main result,
  we define
\begin{equation}
  \label{eps1}
\e_1\vcentcolon= \left(\frac{\e_0}{\alpha_0}\right)^{\frac{1}{3}}\e_0.
  \end{equation}
We take a parameter 
\begin{equation}\label{alpha1}
0<\alpha_1<\alpha_0,
\end{equation}
{which at the end of the procedure
  will be taken to be a function of $r_0$. Given $C_1>0$ and
  $K_1>0$ we define,
\begin{align}
  \label{cond:param}
    \sigma_{{\nu}+1}\vcentcolon= \frac{\sigma_{\nu}}{2},\quad \nu\geq 0\,,
    \\
  \alpha_{\nu}\vcentcolon= \frac{\alpha_1}{2}(1+2^{-{\nu}+1}),\quad
  K_{\nu}\vcentcolon= 2^{\nu-1}K_1,\quad \e_{{\nu}+1}\vcentcolon= \frac{C_1\e_{\nu}^{4/3}}{(\alpha_{\nu}\sigma_{\nu}^{\mu})^{1/3}},    
 \quad \nu\geq 1\,,
  \\
  \eta_{\nu}^3\vcentcolon= \frac{\e_{\nu}}{\alpha_{\nu}
    \sigma_{\nu}^{\mu}},\quad    s_{{\nu}+1}\vcentcolon= s_{\nu}-5\sigma_{\nu},\quad
        \nu\geq 1\,,
 \\
 r_{{\nu}+1}\vcentcolon= \eta_{\nu} r_{\nu},\quad \nu\geq 0\,,
 \\
    D^{\nu}_{(q,\beta)}\vcentcolon= D_{(q,\beta)}(s_{\nu},r_{\nu}),\quad
    \nu\geq 1\,.     
\end{align}

\medskip

We have the following.

\begin{lemma}\label{iterativo}
There exist $h_*>0$, $C_*>0$ and $\varrho_1>0$ such that if 
\begin{equation}
  \label{piccolezza}
C_1>C_*,\quad\frac{\e_1}{\alpha_1}<\varrho_1,\quad K_1^{\tau+1}>\frac{1}{\varrho_1}
  \end{equation}
the following holds. 

 For any $h \in (0,h_*)$ and for all $\nu\geq 1$,
  there exists a measurable set $\Xi_{\nu+1}\subset\Xi_0$
and two symplectic maps $\Phi_{\nu}, \,\Phi_{\nu}^{\tt NLS}: D_{q}^{\nu+1}
\times \Xi_{\nu+1} \to D_{q}^1$, for $q=p$ and $q=p-4$, Lipschitz in the variable $\xi$ and analytic in the
variables $(x, y, z, \bar{z})$
such that 
\begin{align*}
  (N_1^{\tt NLS}+P_1^{\tt NLS})\circ\Phi_{\nu}^{\tt NLS}=N_{\nu+1}^{\tt NLS}+P_{\nu+1}^{\tt NLS},
\end{align*}
where for $\nu \geq 1$ 
\begin{align*}
&N_{\nu}=\Lambda^h+N_{\nu}^{\tt NLS}+N_{\nu}^R,& \\
&{N}_{\nu}=\langle\omega_{\nu},y\rangle+\langle\Omega_{\nu},
z\bar z\rangle, \qquad N_{\nu}^{\tt NLS}=\langle\omega_{\nu}^{\tt
  NLS},y\rangle+\langle\Omega_{\nu}^{\tt NLS},z\bar z\rangle.&
\end{align*}
The frequency maps are lipeomorphisms on $\Xi_0$ of the form 
\begin{align}
\label{lip.fre}
&\omega_\nu(\xi,h)\vcentcolon= \omega_0(\xi,h)+\delta_\nu(\xi,h),& \quad & \omega_\nu^{\tt NLS} (\xi)\vcentcolon= \omega_0^{\tt NLS}+\delta_\nu^{\tt NLS}(\xi),&
  \\
\label{lip.fre.2}
&\Omega_\nu(\xi,h)\vcentcolon= \Omega_0(\xi,h)+\Delta_\nu(\xi,h),& \quad
&\Omega^{\tt NLS}_\nu(\xi)\vcentcolon= \Omega_0^{\tt NLS}+\Delta_\nu^{\tt
  NLS}(\xi),&
\\
\label{omega:R:it}
&\omega_\nu^R(\xi,h)\vcentcolon= \omega_0^R(\xi,h)+\delta_\nu(\xi,h)-\delta_\nu^{\tt NLS}(\xi),& \quad &\Omega_\nu^R(\xi,h)\vcentcolon= \Omega_0^R(\xi,h)+\Delta_\nu(\xi,h)-\Delta_\nu^{\tt NLS}(\xi)&
\end{align}
and the following estimates hold
\begin{align}
\label{6.27}
&|\delta_{\nu}|_{\infty} \sleq
\sum_{n=0}^{\nu-1}\e_n\sleq 
\e_0+\e_1, \qquad |\delta_{\nu}^{\tt
  NLS}|_{\infty}  \sleq
\sum_{n=0}^{\nu-1}\e_n\sleq \e_0+\e_1,
\\
\label{6.28}
& |\delta_{\nu}|_{\infty}^{\cL} \sleq 
\frac{ \e_0}{\alpha_0}+\frac{\e_1}{\alpha_1},
  \qquad |\delta_{\nu}^{\tt
  NLS}|_{\infty}^{\cL}  \sleq
\frac{ \e_0}{\alpha_0}+\frac{\e_1}{\alpha_1}, 
\\
\label{smallcond2.2}
&|\Delta_{\nu}|_{\infty, 1-2 \theta} \sleq \sum_{n=0}^{\nu-1}\e_n\sleq \e_0+\e_1,
\qquad |\Delta_{\nu}^{\tt  NLS}|_{\infty}
\sleq \sum_{n=0}^{\nu-1}\e_n\sleq \e_0+\e_1,
\\
\label{smallcond2.2.1}
& |\Delta_{\nu}|_{\infty, 1-2 \theta}^{\cL} \sleq 
\frac{ \e_0}{\alpha_0}+\frac{\e_1}{\alpha_1},
  \qquad |\Delta_{\nu}^{\tt
  NLS}|_{\infty}^{\cL}  \sleq
\frac{ \e_0}{\alpha_0}+\frac{\e_1}{\alpha_1}, 
\\
\label{cond:OmegaR.2}
&|\delta_{\nu}-\delta_{\nu}^{\tt NLS}|_{\infty} \sleq h
  \left(\e_0+\sum_{n=1}^{\nu-1}\frac{\e_n}{\alpha_{n-1}} \right) \sleq
h\bigg(\e_0+\frac{\e_1}{\alpha_0}\bigg), \quad 
\\
\label{6.31}
&|(\Delta_{\nu})_j-(\Delta_{\nu}^{\tt NLS})_j| \sleq h \, j^4  
  \left(\e_0+\sum_{n=1}^{\nu-1}\frac{\e_n}{\alpha_{n-1}} \right) \sleq
hj^4
\bigg({\e_0}+\frac{\e_1}{\alpha_0}\bigg),\quad \forall
j\in \cS^c,
\end{align}
where the omitted constants are independent of  $\nu$. The sets $\Xi_{\nu}$ are defined as
\begin{equation}\label{defxin}
\Xi_{{\nu}+1}\vcentcolon= \Xi_{\nu} \setminus \bigg(\bigcup_{|k|_1 >K_{\nu}} \mR^0_k(h;\alpha_{\nu},\delta_{\nu},\Delta_{\nu}) \cup \mR^0_k(0;\alpha_{\nu},\delta_{\nu}^{\tt NLS},\Delta_{\nu}^{\tt NLS}) \bigg).
\end{equation}
If the parameter $\tau$ of \eqref{res1} and \eqref{R:NLS}
   fulfills
  \begin{equation}\label{tau.KAM.2}
\tau_*\vcentcolon=  \frac{\tau(1-2\theta)}{3}-5>0
  \end{equation}
then
\begin{equation}\label{second:mestim}
\frac{ \big| \Xi_0 \setminus \Xi_{\nu} \big| }{ \big| \Xi_0 \big| } \leq
C(\theta) \bigg(\frac{ \alpha_0^{\frac{2}{3-\theta}}}{r_0^{4/3}}+\frac{ \alpha_1^{2/3}}{r_0^{4/3}} \sum_{n=1}^\nu \frac{1}{2^{\tau_*n}} \bigg) ,
\end{equation}
with $C(\theta)>0$ diverging as $\theta \to 0$.

For any  $\nu \geq 1$ the perturbations admit a decomposition
\begin{equation*}
 P_{\nu}=P_{\nu}^{\tt NLS}+P_{\nu}^{\mathcal{NG}}+P_{\nu}^R,
\end{equation*}
with $P_{\nu}^\ng= \Pi_\ng P_\nu^\ng$ and 
\begin{align}
& \bn{X_{P_{\nu}}}^{\mL(\alpha_{\nu}),D_q^{\nu}}_{r_{\nu},q,1-2\theta},\, \, 
\bn{X_{P_{\nu}^{\ng}}}^{\mL(\alpha_{{\nu}}),D_q^{\nu}}_{r_{\nu},q}, \, \, \bn{X_{P^{\tt NLS}_{\nu}}}^{\mL(\alpha_{{\nu}}),D_q^{\nu}}_{r_{\nu},q} \leq \e_{\nu},&
\\
\label{lambro.1.1}
 &\bn{X_{P_{\nu}^R}}^{D_p^{\nu}}_{r_{\nu},p-4} \leq \frac{\e_{\nu} h}{\alpha_{\nu}}.& 
\end{align}

Moreover, for any $\nu \geq 1$ the change of coordinates admit a decomposition
\begin{equation*}
\Phi_{{\nu}}=\Phi_{{\nu}}^{\tt NLS}+\Psi_{{\nu}}^R, 
\end{equation*}
and the following estimates hold 
\begin{align*}
  & \bn{\uno-\Phi_{{\nu}}}^{\mL(\alpha_1),D_q^{{\nu}+1}}_{r_{\blu{1}},q}\sleq \sum_{n=1}^\nu\frac{\e_n}{\alpha_n\sigma_n^\mu}\sleq
    \frac{\e_1}{\alpha_1}, \quad 
\bn{\uno-\Phi_{{\nu}}^{\tt NLS}}^{\mL(\alpha_1),D_q^{{\nu}+1}}_{r_1,q}\sleq
 \sum_{n=1}^\nu\frac{\e_n}{\alpha_n\sigma_n^\mu}\sleq
 \frac{\e_1}{\alpha_1},  &
 \\
& \bn{\Psi_{{\nu}}^R}^{D_p^{{\nu}+1}}_{r_1,p-4} \lesssim
 h\left(\sum_{n=1}^{\nu}\frac{\e_n}{\alpha_n^2\sigma_n^\mu}\right)\sleq
{h \bigg(\frac{\e_1}{\alpha_1^2}\bigg). }&
\end{align*}
\end{lemma}
\begin{remark}
  \label{6.3}
By \eqref{defxin} the non resonant set satisfies
\begin{equation}
  \label{res.vero}
\Xi_0 \setminus\Xi_{n} \subseteq  \bigcup_{\nu=0}^{n} \bigcup_{|k|_1 >K_\nu} \bigg(\mathcal{R}^{\theta_\nu}_{k}(h; \alpha_\nu,\delta_\nu,\Delta_\nu)  \cup \mathcal{R}_{k}(0; \alpha_\nu,\delta_\nu^{\tt NLS},\Delta_\nu^{\tt NLS}) \bigg),
  \end{equation}
with $K_0=0$, $\theta_0\vcentcolon= \theta$ introduced in Lemma \ref{Lemma:step1}, and 
$\theta_\nu=0$ for any $\nu \geq 1$. 
\end{remark}

\begin{proof}
The proof is a slight modification of the proof of the iterative Lemma, Section $4$
\cite{poschel1996kam}. Here we only give the details of the proof of the new part.

{The definition \eqref{defxin} is due to the fact that by the choice 
of the various parameters and the third of conditions \eqref{piccolezza}, the frequencies at step $\nu+1$ automatically fulfill the Melnikov conditions for $|k|_1 \leq K_\nu$ (see \cite{poschel1996kam} Section 3).} By combining \eqref{res.vero}  with \eqref{meas.finalNLS}, \eqref{meas.final} and \eqref{first:mestim}, we have \eqref{second:mestim}.

Concerning the estimates, we only prove the one of $\Psi^R_{\nu+1}$:
 all the others are obtained following the procedure used to
 prove the Iterative Lemma in \cite{poschel1996kam} and Lemma
   \ref{Lemma:step1} ( see also \cite{berti2013kam}). We recall that the transformation $\Phi_{\nu}$ and $\Phi_{\nu}^{\tt NLS}$
  are constructed as  
\begin{align*}
\Phi_{\nu}\vcentcolon= \phi_{1}\circ \phi_{2}\circ...\circ\phi_{\nu},
\qquad
\Phi_{\nu}^{\tt NLS}\vcentcolon= \phi_{1}^{\tt NLS}\circ \phi_{2}^{\tt NLS}\circ...\circ\phi^{\tt NLS}_{\nu},
\end{align*}
where $\phi_{\nu}$ and $\phi_{\nu}^{\tt NLS}$ are the flows of some
auxiliary Hamiltonians $G_{\nu+1}$ and $G_{\nu+1}^{\tt NLS}$ obtained by solving
the cohomological equation through
Lemma \ref{solhomo}. Then, by a variant of the proof of \eqref{resti.6.1}, we have
\begin{equation*}
\phi_{\nu}=\phi_{\nu}^{\tt NLS}+\psi_{\nu}^R,
\end{equation*}
with
\begin{equation*}
 \bn{\psi_{{\nu}}^R}^{D_p^{{\nu}+1}}_{r_{{\nu}+1},p-4} \lesssim \frac{\e_{\nu}
  h}{\alpha_{\nu}^{2}\sigma_\nu^\mu}.
  \end{equation*}
It follows
\[
  \bn{\Psi_{{\nu}}}_{r_1,p-4}^{D_p^{{\nu}+1}}
 \leq
  \bn{\Phi_{\nu-1}\circ\phi_{{\nu}}-\Phi_{\nu-1}\circ\phi_{{\nu}}^{\tt NLS}}_{r_1,p-4}^{D_p^{{\nu}+1}}+
\bn{\Phi_{\nu-1}\circ\phi_{{\nu}}^{\tt NLS}-\Phi_{\nu-1}^{\tt NLS}\circ\phi_{{\nu}}^{\tt NLS}
}_{r_1,p-4}^{D_p^{{\nu}+1}}
\]
and 
\begin{align*}
&\bn{\Phi_{\nu-1}\circ\phi_{{\nu}}-\Phi_{\nu-1}\circ\phi_{{\nu}}^{\tt NLS}}_{r_1,p-4}^{D_p^{{\nu}+1}}  \leq \bn{\diff \Phi_{\nu-1}}_{r_1,r_{\nu}}^{D^\nu_{p-4}} \bn{\phi_{\nu}-\phi_{\nu}^{\tt NLS}}_{r_{\nu},p-4}^{D_p^{\nu+1}}, &
\\
&\bn{\Phi_{\nu-1}\circ\phi_{{\nu}}^{\tt NLS}-\Phi_{\nu-1}^{\tt NLS}\circ\phi_{{\nu}}^{\tt NLS}
}_{r_1,p-4}^{D_p^{{\nu}+1}}  \leq \bn{\Phi_{\nu-1}-\Phi_{\nu-1}^{\tt NLS}
}_{r_1,p-4}^{D_p^{{\nu}}},&
\end{align*}
where $\bn{\diff \Phi_{\nu-1}}_{r_1,r_{\nu}}^{D^\nu_{p-4}}$ is the
sup over $D^\nu_{p-4}$ of the norm of $\diff \Phi_{\nu-1}$ as an operator
from the phase space endowed by the norm with parameter $r_\nu$ to the phase
space endowed by the norm with parameter $r_1$.

Therefore we define iteratively a sequence 
\begin{align*}
&b_{\nu}\vcentcolon= C\bn{\psi_{{\nu}}^R}^{D_p^{{\nu}+1}}_{r_{{\nu}+1},p-4},&
\quad &\nu\geq 1&\\
&a_{{\nu}}\vcentcolon= \bn{\Psi_{{\nu}}}_{r_1,p-4}^{D_p^{{\nu}+1}},&\quad
 &\nu\geq 1,\quad a_0 \vcentcolon= 0,&
\end{align*}
where $C$ is a positive constant such that for any $\nu \geq 1$ we have  $\bn{\diff
    \Phi_{\nu}}_{r_1,r_{\nu}}^{D^{\nu}_{p-4}} \leq C$.
Then for any $\nu \geq 1$
$$
a_{{\nu}}\leq b_{\nu}+a_{\nu-1}
$$
and therefore 
\begin{align}
\bn{\Psi_{\nu}}_{r_1,p-4}^{D_p^{{\nu}+1}}\leq 
\sum_{n=1}^\nu b_n\sleq
\sum_{n=1}^\nu
  \bn{\psi_n^R}_{r_{n+1},p-4}
  ^{D_p^{n+1}}\sleq  \frac{h\e_1}{\alpha_1^2}\,.
\end{align}
\end{proof}
As a consequence of the iterative lemma, we have the following KAM theorem.

\begin{theorem}\label{TeoremaGenerale} {There exist $h_*,r_*, \varrho_*, \tau_*, s>0$} such that if
    \begin{align}
      \label{condizione}
\frac{\e_0}{\alpha_0}+\frac{\e_1}{\alpha_1}<\varrho_*
    \end{align}
 then, for any $h \in (0,h_*)$ and any $r_0\in(0, r_*)$ there exist a Cantor
  set $\Xi_{\infty} \subset \Xi_0$, a Lipeomorphism $\omega_{\infty}:\Xi_0 \to
  \mathbb{R}^N$, with Lipschitz constant
  uniformly bounded with respect to $h$ and $r_0$, and also a family of torus embeddings
\[
\Phi: \mathbb{T}^N \times \Xi_{\infty} \to  \mathcal{P}^{a,p},
\]
Lipschitz in $\xi$ and
such that for all $\xi \in \Xi_{\infty}$, the map $\Phi $ restricted to $\mathbb{T}^N \times \{\xi\}$ is a real analytic
embedding of a rotational torus with frequency
$\omega_{\infty}(\xi)$, for the Hamiltonian
$H=H(\xi)$ in \eqref{H}. In particular, the embedding is real
analytic on $D(s/2)\vcentcolon= \{|\Im x| < s/2, {y=z=0}\}$, and fulfills 
\begin{equation}
  \label{stimePhi}
\bn{ \Phi- \textnormal{id} }^{D(s/2)}_{r_0, p} \lesssim
\bigg(\frac{\e_0}{\alpha_0}+\frac{\e_1}{\alpha_1}\bigg), \qquad \bn{ \Phi- \textnormal{id} }^{\mL, D(s/2)}_{r_0, p} \lesssim
\bigg(\frac{\e_0}{\alpha_0^2}+\frac{\e_1}{\alpha_1^2}\bigg).
\end{equation}
The frequency 
\begin{equation}\label{stim:th:omega}
|\omega_{\infty} - \omega_0|_{\infty} \lesssim (\e_0+\e_1), \qquad |\omega_{\infty} - \omega_0|^\mL_{\infty} \lesssim \bigg(\frac{\e_0}{\alpha_0}+\frac{\e_1}{\alpha_1}\bigg),
\end{equation}
{with $\omega_0$ defined in \eqref{Oomega}}. Moreover, there exist a Lipeomorphism $\omega_{\infty}^{\tt NLS}:\Xi_0 \to
  \mathbb{R}^N$ independent of $h$, and a Lipschitz map $\omega^{R}_{\infty}:\Xi_0 \to
  \mathbb{R}^N$ such that 
\begin{equation}\label{stimaresto}
\omega_{\infty}=h^{-1}+\omega_{\infty}^{\tt NLS}+ \omega^{R}_{\infty}, \qquad |{\omega_0^R}-\omega^{R}_{\infty}|_{\infty}\lesssim h  \bigg(\frac{\e_0}{\alpha_0}+\frac{\e_1}{\alpha_1}\bigg),
\end{equation}
{with $\omega_0^R$ defined in \eqref{OomegaR}}. In addition, there exists a real analytic map $\Psi^R$ fulfilling 
\begin{equation}\label{decomp:emb}
\Psi^R: D(s/2) \times \Xi_{\infty} \to  \mathcal{P}^{a,p}, \qquad \bn{\Psi^R}_{r_0, p-4}^{D(s/2)}\lesssim h\,  \bigg(\frac{\e_0}{\alpha_0^2}+\frac{\e_1}{\alpha_1^2}\bigg),
\end{equation}
and $\Phi= \Phi^{\tt NLS}+ \Psi^R,$ such that the restriction
$\Phi^{\tt NLS} \big|_{\mathbb{T}^n \times \{ \xi \}}$ is a real
analytic embedding of a rotational torus with
frequencies $\omega_{\infty}^{\tt NLS}(\xi)$ for the
Hamiltonian $H^{\tt NLS}$, and they satisfy estimates like \eqref{stimePhi}, \eqref{stim:th:omega}. 

\smallskip

Furthermore, for any $\tau> \tau_*$ we have 
\begin{equation}\label{final:mestim}
\frac{ \big| \Xi_0 \setminus \Xi_\infty \big| }{ \big| \Xi_0 \big| } \lesssim \bigg(\frac{ \alpha_0^{\frac{2}{3-\theta}}}{r_0^{4/3}}+\frac{ \alpha_1^{2/3}}{r_0^{4/3}}\bigg) .
\end{equation}
\end{theorem}

 \section{Proof of main results}
\subsection{Proof of Theorem \ref{VicinanzaTh}}

First, we need to make explicit choices for $\alpha_0$ and $\alpha_1$ introduced in \eqref{alpha0}, \eqref{alpha1}. We decide to consider them as functions of $r$ as follows
\begin{equation}
  \label{scelte}
\alpha_0\vcentcolon= r_0^{a_0},\quad \alpha_1\vcentcolon= r_0^{a_1}.
\end{equation}
We are going to choose $a_0$, $a_1$ and $\theta$ in a such a way
that the smallness condition \eqref{condizione} for the applicability of the KAM theorem \ref{TeoremaGenerale} is
fulfilled for $r_0>0$ small enough. This will also imply that the
transformations $\Phi$ and $\Phi^{\tt NLS}$ are close to
identity by \eqref{stimePhi}. Then, we want the set $\Xi_{\infty}$ to have
asymptotically full measure, namely the r.h.s. of
\eqref{final:mestim} has to tend to zero as $r_0\to0$. Finally, we would like
to minimize $\Psi^{R}$ (in appropriate norm, see \eqref{decomp:emb}), which measures the distance between the NLS tori
and the KG tori.

The condition \eqref{condizione} implies
\begin{equation}
\label{cond.1.aux}
 a_0<2,\quad
a_1<\frac{8}{3}-\frac{a_0}{3},
\end{equation}
while the smallness of the resonant set \eqref{final:mestim} implies
\begin{equation}
\label{cond.2.aux}
a_0>2-\frac{2}{3}\theta,\quad
a_1>2.
\end{equation}
We now set $\theta=1/2-3 \vs$, which satisfies condition \eqref{scelte.theta} for $\vs \in (0,1/6)$. Then, we set 
\begin{equation}
  \label{scelta.4}
a_0\vcentcolon= 2-\frac{2}{3}\theta+\vs=\frac{5}{3}+3\vs,\quad a_1\vcentcolon= 2+\vs,
\end{equation}
which satisfy the restrictions \eqref{cond.2.aux} by definition. Moreover, in this setting, the second condition in \eqref{cond.1.aux} is satisfied for $\vs \in (0,1/18)$. {From this choice of the parameters combined with \eqref{eps1}, \eqref{def:e}, \eqref{Stima0}, \eqref{1Stima}, \eqref{3Stima} and \eqref{def.di.r}
\begin{align}
  \label{scelta.5}
&\frac{\varepsilon_0}{\alpha_0}\lesssim \, r_0^{\frac{1}{3}-3\vs} \lesssim\, \tR^{\frac{1}{2}-\frac{9}{2}\vs},&\quad
&\frac{\varepsilon_1}{\alpha_1}\lesssim \, r_0^{\frac{1}{9}-2\vs}\lesssim \, \tR^{\frac{1}{6}-3\vs},&
\Longrightarrow\quad
&\frac{\varepsilon_0}{\alpha_0}+\frac{\varepsilon_1}{\alpha_1} \lesssim \,\tR^{\frac{1}{6}-3\vs},& \\
 \label{scelta.6}
&\frac{\varepsilon_0}{\alpha_0^2}\lesssim \, r_0^{-\frac{4}{3}-6\vs}\lesssim \,\tR^{-2-9\vs},&\quad
&\frac{\varepsilon_1}{\alpha_1^2}\lesssim
 \, r_0^{-\frac{17}{9}-3\vs}\lesssim \,\tR^{-\frac{17}{6}-\frac{9}{2}\vs},&
\Longrightarrow\quad
&\frac{\varepsilon_0}{\alpha_0^2}+\frac{\varepsilon_1}{\alpha_1^2} \lesssim\, \tR^{-\frac{17}{6}-\frac{9}{2}\vs}.&
\end{align}}
We remark that, by the above relations, the smallness condition \eqref{condizione} is now depending on $\mathtt{R}$. We define $\mathtt{R}_*>0$ such that
\[
\mathtt{R}_*^{\frac{1}{6}-3\vs}=\varrho_*
\]
where $\varrho_*$ is given in Theorem \ref{TeoremaGenerale}. Thus, for any $\mathtt{R}\in (0, \mathtt{R}_*)$, the condition \eqref{condizione} is satisfied.

Now, we study the set of the common frequencies of KAM tori for KG and NLS equations. We introduce the function
\begin{equation}\label{def:omegahat}
\widehat{\omega}_{\infty}=\widehat{\omega}_{\infty}(c)\vcentcolon= \omega_{\infty}-c^2.
\end{equation}
 We
also define  the sets
  \begin{equation}
    \label{omegafinale}
\caW_0=\caW_0(c)\vcentcolon= \widehat{\omega}_{\infty}\big(\Xi_0\big) \cap
\omega_{\infty}^{\tt NLS}\big(\Xi_0\big) 
    ,\quad
\caW _\infty= \caW _\infty(c)\vcentcolon= \widehat{\omega}_{\infty}\big(\Xi_\infty\big) \cap
\omega_{\infty}^{\tt NLS}\big(\Xi_\infty\big),
    \end{equation}
with $\Xi_0$ defined in \eqref{Xi0} and $\Xi_\infty\subset \Xi_0$
obtained in Theorem \ref{TeoremaGenerale}.

We also define the sets
\begin{equation}
  \label{defset}
\Xi^{KG}\vcentcolon= \widehat\omega_\infty^{-1}(\caW_\infty),\quad
\Xi^{\tt NLS}\vcentcolon= (\omega_\infty^{\tt NLS})^{-1}(\caW_\infty),
  \end{equation}
which are the set of the amplitudes corresponding to an invariant torus of NLS and an invariant torus of KG with
  the same frequencies. Their measure is estimated in the following lemma. 

\begin{lemma}\label{lem:quasifinale}
There exist $\tR_*$, $h_*>0$ such that for any
    $\tR \in (0,\tR_*)$, $h \in (0,h_*)$ the following holds
  \begin{align}
    \label{7.sti.1}
\left|\Xi_0\setminus \Xi^{KG}\right|\lesssim
(h+\tR^{2+\frac{6\vs}{5+6\vs}}+ \tR^{2+\vs})\tR^{2(N-1)}, \\
 \label{7.0.sti.1}
 \left|\Xi_0\setminus \Xi^{\tt NLS}\right|\lesssim
(h+\tR^{2+\frac{6\vs}{5+6\vs}}+ \tR^{2+\vs})\tR^{2(N-1)}.
  \end{align}
Furthermore, by assuming that 
\begin{equation}\label{cond:h}
\vs  \in (0,1/18), \qquad h\lesssim \tR^{2+\vs},
\end{equation} 
the above estimates take the form 
\begin{equation}\label{7.sti.fin}
\left|\Xi_0\setminus \Xi^{KG}\right|\lesssim \tR^{2N} \tR^{\vs}, \qquad  \left|\Xi_0\setminus \Xi^{\tt NLS}\right| \lesssim \tR^{2N} \tR^{\vs}.
\end{equation}
\end{lemma}
\begin{proof}
We first estimate
\begin{equation}
\caW_0 \setminus \caW_\infty \subseteq \widehat{\omega}_{\infty}\big(\Xi_0 \setminus \Xi_\infty\big) \cup \omega^{\tt NLS}_{\infty}\big(\Xi_0 \setminus \Xi_\infty \big).
\end{equation}
By using that the frequency maps $\widehat{\omega}_{\infty}, \omega_{\infty}^{\tt NLS}$ are Lipeomorphisms combined with \eqref{final:mestim}, \eqref{scelta.4} and \eqref{def.di.r} we have 
\begin{equation}\label{eq:inc:2}
|\caW_0 \setminus \caW_\infty| \lesssim |\Xi_0 \setminus \Xi_\infty| \lesssim \tR^{2N}(\tR^{\frac{6\vs}{5+6\vs}}+ \tR^{\vs}).
\end{equation}
We now prove 
  \eqref{7.sti.1}.  By definitions \eqref{omegafinale}, we have the following identity
\begin{equation}\label{eq:inc:1}
\widehat{\omega}_{\infty}\big(\Xi_0\big) \setminus \caW_\infty = (\widehat{\omega}_{\infty}\big(\Xi_0\big) \setminus \caW_0) \cup (\caW_0 \setminus \caW_\infty),
\end{equation}
with the second term already estimated. We come to the first
  which coincides with
  \begin{equation}
    \label{xidastimare}
\widehat{\omega}_{\infty}\big(\Xi_0\big)\setminus(\widehat\omega_\infty\big(\Xi_0\big)\cap\omega^{\tt
NLS}_\infty\big(\Xi_0\big) ).
    \end{equation}
We want to apply Lemma \ref{stime_interno} with $f=\widehat\omega_\infty$, $g=\omega_\infty^R$. To this end we need to verify the hypothesis \eqref{stim:2:aux}. We start by estimating the
Lipschitz constant of $f^{-1}$. We apply Lemma \ref{perturbazione} to $\omega_\infty=\omega_0+\delta_\infty$, with $\delta_\infty\vcentcolon= \omega_\infty-\omega_0$ estimated in 
\eqref{stim:th:omega}. For $\tR$ small enough, from \eqref{stim:3:aux} we obtain $L_{f^{-1}} \leq 2L_{\omega_0^{-1}}$. Recalling the definition \eqref{Oomega}, we have  $L_{\omega_0^{-1}}=\|A^{-1}\|$, with an estimate of order zero in $\tR$ and $h$ proved in Lemma \ref{inverto}. On the other hand, the Lipschitz
constant of $g$ is estimated by
\begin{equation}\label{Lg}
|L_g|\sleq \left\|A-A^{\tt
  NLS}\right\|+\left|\omega_0-\omega_\infty\right|^\cL_\infty
+\left|\omega_0^{\tt
  NLS}-\omega_\infty^{\tt
  NLS}\right|^\cL_\infty \sleq
h+\tR^{\frac{1}{6}-3\vs}
\end{equation}
where we used \eqref{OomegaR}, \eqref{def:AKG}, \eqref{matNLS} and \eqref{stim:th:omega}. Then, for $h$ and $\tR$ small enough, the hypothesis \eqref{stim:2:aux} is fulfilled. We can apply the Lemma
\ref{stime_interno}, with
\begin{align*}
\varrho \sleq |g|_\infty \sleq |\omega_0^R|_\infty+\left|\omega_0^R-\omega_\infty^R\right|_\infty
\sleq
h\left(1+\tR^2+ \tR^{\frac{1}{6}-3\vs}\right),
\end{align*}
where we used \eqref{stimeom0} and \eqref{stimaresto}. Under this hypothesis we have
$$
 \widehat{\omega}_\infty\big(\Xi_{-\varrho}\big) \subseteq \widehat{\omega}_\infty\big(\Xi_0\big)\cap\omega^{\tt
NLS}_\infty\big(\Xi_0\big), 
$$
and therefore
$$
\widehat{\omega}_{\infty}\big(\Xi_0\big)\setminus(\widehat{\omega}_\infty\big(\Xi_0\big)\cap\omega^{\tt
NLS}_\infty\big(\Xi_0\big) ) \subset \widehat{\omega}_\infty \big(\Xi_0 \setminus\Xi_{-\varrho}\big).
$$
Recalling that $\widehat{\omega}_\infty$ is a lipeomorphism we have
$$
|\widehat{\omega}_{\infty}\big(\Xi_0\big)\setminus(\widehat{\omega}_\infty\big(\Xi_0\big)\cap\omega^{\tt
NLS}_\infty\big(\Xi_0\big) ) | \lesssim \left|\Xi_0\setminus\Xi_{-\varrho}\right|\sleq \varrho \tR^{2(N-1)}.
$$
This gives \eqref{7.sti.1}. The estimate \eqref{7.0.sti.1} is obtained in the
same way.
\end{proof}

  \begin{lemma}
    \label{distanza.tori}
Let $w\in\caW_\infty$ then there exist a unique couple $\xi^{\tt KG}(w)$,
$\xi^{\tt NLS}(w)$ with the property that
\begin{equation}
  \label{coppia}
w=\widehat{\omega}_\infty(\xi^{ \tt KG} )=\omega^{\tt NLS}_\infty(\xi^{\tt NLS}
),
  \end{equation}
furthermore one has
\begin{equation}
  \label{disti.xi}
\left|\xi^{\tt KG}-\xi^{\tt NLS}  \right|\sleq h.
  \end{equation}
  \end{lemma}
\begin{proof}
To prove \eqref{disti.xi}, recall that $\omega^{\tt NLS}_\infty, \widehat{\omega}_\infty: \Xi_0 \to \mathbb{R}^N$ are lipeomorphisms. Then, for any $w \in \caW_\infty$, by definition \eqref{omegafinale}, there exist $\xi^{\tt KG}, \xi^{\tt NLS} \in \Xi_\infty$ such that
\begin{equation}
w = \widehat{\omega}_\infty(\xi^{\tt KG}) = \omega^{\tt NLS}_\infty(\xi^{\tt NLS}).
\end{equation}
 Furthermore such couple is unique
by the Lipschitz estimates of the inverse functions. Moreover, we have
\begin{align*}
\left|\omega_\infty^{\tt NLS}(\xi^{\tt NLS}) - \omega_\infty^{\tt NLS}(\xi^{\tt KG})\right|
&= \left|w - \omega_\infty^{\tt NLS}(\xi^{\tt KG})\right| \\
&= \left|\widehat{\omega}_\infty(\xi^{\tt KG}) - \omega_\infty^{\tt NLS}(\xi^{\tt KG})\right| \\
&\leq |\omega^R(h)|_\infty \sleq h,
\end{align*}
where we used \eqref{stimaresto} and \eqref{stimeom0}. Finally, using that $(\omega_\infty^{\tt NLS})^{-1}$ is Lipschitz, the thesis follows by
\[
\left|\xi^{\tt KG}-\xi^{\tt NLS}  \right| \leq L_{(\omega_\infty^{\tt NLS})^{-1}} \left|\omega_\infty^{\tt NLS}(\xi^{\tt NLS}) - \omega_\infty^{\tt NLS}(\xi^{\tt KG})\right| \lesssim h.
\]
\end{proof}

\subsection{End of the proof of Theorem \ref{VicinanzaTh} and of Corollary \ref{Maintheorem}}

Now we prove \eqref{dist.main}. First, for any $\omega \in \caW_\infty $  and $x \in D(s/2)$ we define
\begin{align*}
&\Psi_\omega^{\tt KG}(x)\vcentcolon= \cT_0\circ \cT_1\circ\Phi(x, \widehat{\omega}^{-1}_{\infty}(\omega)) \in \ell^{a,p}& 
\\ 
&\Psi_\omega^{\tt
  NLS}(x)\vcentcolon= \cT_0^{\tt NLS}\circ \cT_1\circ\Phi^{\tt NLS}(x, (\omega^{\tt NLS})^{-1}(\omega)) \in \ell^{a,p},&
\end{align*}
where $\cT_0, \cT^{\tt NLS}_0$ are the Birkhoff maps given in Theorem \ref{Teorema0}, $\cT_1$ is the coordinate transformation
\eqref{def.t1} to action angle variables and $\Phi, \Phi^{\tt NLS}$
are the embeddings given in Theorem \ref{TeoremaGenerale}. 
For any $(x,\xi) \in D(s/2) \times \Xi_\infty$ we denote
\begin{align*}
&\Phi(x,\xi)\vcentcolon= (X(x,\xi),Y(x,\xi),Z(x,\xi), \bar{Z}(x,\xi))& \\
&\Phi^{\tt NLS}(x,\xi)\vcentcolon= (X^{\tt NLS}(x,\xi),Y^{\tt NLS}(x,\xi),Z^{\tt NLS}(x,\xi), \bar{Z}^{\tt NLS}(x,\xi))&
\end{align*}
(where $\bar Z$ gives a contribution identical to
  $Z$ due to the real analyticity of the embedding).
By the definition of the norm \eqref{Norma1} and by \eqref{stimePhi},
we have
\begin{equation}\label{XYZ}
\begin{aligned}
&\left|x-X(x,\xi)\right|\sleq
\left(\frac{\e_0}{\alpha_0}+\frac{\e_1}{\alpha_1}\right) \lesssim  \tR^{\frac{1}{6}-3\vs},& \quad &\left| X(x,\cdot)\right|^\cL \lesssim \left(\frac{\e_0}{\alpha_0^2}+\frac{\e_1}{\alpha_1^2}\right) \lesssim \tR^{-\frac{17}{6}-\frac{9}{2}\vs}& \\
&\left|Y(x,\xi)\right|\sleq \tR^3
\left(\frac{\e_0}{\alpha_0}+\frac{\e_1}{\alpha_1}\right) \lesssim \tR^{\frac{19}{6}-3\vs},& \quad &\left|Y(x,\cdot )\right|^{\cL}  \lesssim \tR^3\left(\frac{\e_0}{\alpha_0^2}+\frac{\e_1}{\alpha_1^2}\right) \lesssim \tR^{\frac16- \frac{9}{2}\vs} &
\\
&\left\|Z(x,\xi)\right\|_{a,p}\sleq \tR^{\frac32}
\left(\frac{\e_0}{\alpha_0}+\frac{\e_1}{\alpha_1}\right) \lesssim \tR^{\frac{5}{3}-3\vs} ,& \quad &\left\|Z(x,\cdot)\right\|_{a,p}^\cL \lesssim \tR^{\frac32}\left(\frac{\e_0}{\alpha_0^2}+\frac{\e_1}{\alpha_1^2}\right) \lesssim \tR^{-\frac43- \frac{9}{2}\vs}
\end{aligned}
\end{equation}
with the same estimates for the components of $\Phi^{\tt NLS}$.
Using the notation
\[
\sqrt{\Upsilon} e^{-\mathrm{i} \chi}=(\sqrt{\Upsilon_{\tj_1}} e^{-\mathrm{i} \chi_{\tj_1}}, \dots, \sqrt{\Upsilon_{\tj_N}} e^{-\mathrm{i} \chi_{\tj_N}}),
\]
provided $|\Upsilon|\simeq \tR^2$, we have
\begin{align}
  \label{upsilon}
\left\|\left(\sqrt{\Upsilon_1}e^{-\im
    \chi_1},Z_1\right)-\left(\sqrt{\Upsilon_2}e^{-\im
    \chi_2},Z_2\right)\right\|_{a,p} \sleq
\frac{\left|\Upsilon_1-\Upsilon_2\right|}{\tR}
+\tR\left|\chi_1-\chi_2\right| +\left\|Z_1-Z_2\right\|_{a,p}.
\end{align}
For $(x, \xi) \in D(s/2) \times \Xi_\infty$, we apply this estimate with 
\begin{align*}
&\Upsilon_1=\xi+Y(\xi,x),& \qquad &\Upsilon_2=\xi,& \qquad
&\chi_1=X(x,\xi),& \\
&\chi_2=x,& \qquad &Z_1=Z(x,\xi),& \qquad &Z_2=0.&
\end{align*}
By the estimates in \eqref{XYZ} we obtain 
\begin{equation}  \label{stima.def}
\begin{aligned}
\| \cT_1(x,\xi)-\cT_1\circ\Phi(x,\xi)\|_{a, p}\sleq (\tR^2+\tR+ \tR^{\frac32})
\left(\frac{\e_0}{\alpha_0}+\frac{\e_1}{\alpha_1}\right) \lesssim \tR^{\frac76-3 \vs}.
\end{aligned}
\end{equation}
Therefore, by choosing
\begin{equation}\label{cond:vs}
\vs = \frac{1}{36}
\end{equation}
recalling Remark \ref{cambio norme} and the estimate \eqref{sti.tra.1} for $\cT_0$ for any $(x, \xi) \in D(s/2) \times \Xi_\infty$, we have  
\begin{equation}\label{aux:utile}
 \| \Psi_{\widehat{\omega}_\infty(\xi)}^{\tt KG}(x) - \cT_1(x,\xi) \ \|_{a,p} \lesssim
 \tR^{\frac{7}{6}-3\vs}\sleq \tR^{\frac{13}{12}}.
\end{equation}
The same computations hold true for $\Psi^{\tt NLS}$, and we obtain \eqref{dist.main}

\medskip

We observe that, by the choice \eqref{cond:vs}, the estimates \eqref{misura.0} are implied by the bounds \eqref{7.sti.fin}, which have been proved in Lemma \ref{lem:quasifinale}.

Now we are going to prove  \eqref{dist.sole}. For any $(x,\omega) \in D(s/2) \times \caW_\infty$, we can apply the previous estimate with $\xi=\widehat{\omega}_\infty^{-1}(\omega)$ to obtain the first estimate in \eqref{dist.main}. 
The second inequality in \eqref{dist.main} follows similarly. 
For any $\omega \in \caW_\infty$ we recall the definition of $\xi^{\tt NLS}, \xi^{\tt KG}$ given in Lemma \ref{distanza.tori}. We use now \eqref{upsilon} with 
\begin{align*}
&\Upsilon_1=\xi^{\tt NLS}+Y^{\tt
  NLS}(x,\xi^{\tt NLS}),& \qquad 
&\Upsilon_2=\xi^{\tt KG}+Y(x,\xi^{\tt KG}),& \qquad  &\chi_1=X^{\tt
  NLS}(x,\xi^{\tt NLS}),& \\
 &\chi_2=X(x,\xi^{\tt KG}),& \qquad &Z_1=Z^{\tt
  NLS}(x,\xi^{\tt NLS}),& \qquad &Z_2=Z(x,\xi^{\tt KG}).&
\end{align*} 
By the estimates \eqref{decomp:emb}, \eqref{XYZ} and \eqref{disti.xi} we have
  \begin{align*}
\left|\Upsilon_1-\Upsilon_2\right|&\leq \left|\xi^{\tt KG}-\xi^{\tt
  NLS}\right|+\left|Y(x,\xi^{\tt KG})-Y(x,\xi^{\tt
  NLS})\right|&
\\
&+ \left|Y(x,\xi^{\tt NLS})-Y^{\tt NLS}(x,\xi^{\tt
  NLS})\right| \sleq h+\tR^3 h\left(\frac{\e_0}{\alpha_0^2}+\frac{\e_1}{\alpha_1^2}\right),& 
  \end{align*}
and by similar computations 
\[
|\chi_1-\chi_2| \lesssim h\left(\frac{\e_0}{\alpha_0^2}+\frac{\e_1}{\alpha_1^2}\right), \qquad
\| Z_1-Z_2\|_{a,p-4}  \lesssim \tR^{\frac32} \left(\frac{\e_0}{\alpha_0^2}+\frac{\e_1}{\alpha_1^2}\right).
\]
Applying the above estimates we obtain 
\begin{equation}
  \left\|(\sqrt{\xi^{\tt NLS}+Y^{\tt NLS}}e^{iX^{\tt NLS}},Z^{\tt
      NLS})-(\sqrt{\xi^{\tt KG}+Y}e^{iX},Z)\right\|_{a,p-4} \lesssim h \tR^{-\frac{11}{6}-\frac92\vs}+h\tR^{-1}\lesssim h \tR^{-\frac{47}{24}},
\end{equation}
where we impose \eqref{cond:vs}.

Then, by Remark \ref{cambio norme} and the estimate \eqref{sti.tra.1} for $\cT_0$ for any $(x, \omega) \in D(s/2) \times \caW_\infty$, we have  
\begin{equation}
 \label{stima.perdita}
 \| \Psi_\omega^{\tt KG}(x) - \Psi_\omega^{\tt NLS}(x)  \|_{a,p-4}  \lesssim  h \tR^{-\frac{47}{24}}.
\end{equation}
On the other hand, by the estimates \eqref{aux:utile} for $\Psi^{\tt NLS}$ with $\xi=\xi^{\tt NLS}$ and for $\Psi^{\tt KG}$ with $\xi=\xi^{\tt KG}$, by performing a triangle inequality combined with \eqref{cond:h}, we obtain
\begin{equation}
\label{no.perdita.lip}
 \| \Psi_\omega^{\tt KG}(x) - \Psi_\omega^{\tt NLS}(x,\omega) \ \|_{a,p}  \lesssim
\tR^{\frac{13}{12}}+h\tR^{-1}\lesssim \tR^{\frac{37}{36}},
\end{equation}
where we used \eqref{cond:h}.

\medskip

Finally by using the interpolation estimate
$$
\left\|u\right\|_{a,p-4\sigma }\leq  \left\|u\right\|_{a,p-4 }^\sigma\left\|u\right\|_{a,p}^{1-\sigma}, \qquad \sigma\in [0, 1],
$$
from \eqref{stima.perdita} and \eqref{no.perdita.lip}
 we have 
\begin{equation}\label{stima.finale}
\left\|e^{-\iu c^2t }\psi_{c, \omega}^{KG}(t)-\varphi_\omega^{\tt NLS}(t)
\right\|_{a,p-4\sigma} \lesssim h^\sigma \tR^{\frac{37}{36}-\frac{215}{72} \sigma},
\end{equation}
where 
 \begin{align*}
  \psi_{c,\omega}^{KG}(t)\vcentcolon= \Psi_\omega^{KG}((\omega+c^2) t), \qquad
\varphi^{\tt NLS}_\omega(t)\vcentcolon= \Psi_\omega^{\tt NLS}(\omega t).
\end{align*}
Thus, \eqref{dist.sole} follows.

\medskip

The proof of Corollary \ref{Maintheorem} is a direct consequence of \eqref{stima.finale}. Indeed, once we fix a map 
\[
\omega:(c_*,\infty)\to\R^N, \quad c \mapsto \omega(c)\in\caW(c),
\]
we can apply Theorem \ref{VicinanzaTh} obtaining 
\[
\left\|e^{-\iu c^2t }\psi_{c, \omega(c)}^{KG}(t)-\varphi_{\omega(c)}^{\tt NLS}(t)
\right\|_{a,p-4\sigma} \lesssim h^\sigma, \qquad \forall \, \sigma \in (0,1].
\]
Finally, by taking the limit as $h \to 0^+$ we have \eqref{main.44}.
\qed

\newpage
\appendix
\section{Some auxiliary estimates}\label{flowmap}
  \begin{lemma}\label{algebra}
If $a \geq 0$, $p> 1/2$ and $\beta\geq 0$ the space $\ell^{a, p,\beta}$
is an algebra with respect to the convolution
\begin{equation}\label{conv}
( z\star w)_j\vcentcolon=  \sum_{k\in \Z} z_{k-j}w_k \qquad j\in \Z,
\end{equation}
uniformly with respect to $c$. More precisely there exists a
  constant $K$ independent of $c$, such that
  \begin{equation}
    \label{algebra.unif}
\left\| ( z\star w)\right\|_{a,p,\beta}\leq K\left\|
z\right\|_{a,p,\beta}\left\| w\right\|_{a,p,\beta}.
    \end{equation}
\end{lemma}
\begin{proof}
 We follow the proof given in \cite{poschel1996quasi} for the case
  $\beta=0$. 
 For any $z, w \in \ell^{a,p,\beta}$, we have
\[
\| z \star w \|_{a,p,\beta}^2\vcentcolon=  \sum_{j\in \Z} \jbs{j}^{2p} e^{2a|j|}\tw_j^{2\beta}\bigg| \sum_{k\in \Z} z_{j-k}w_k \bigg|^2.
\]
We consider the case $a=0$, which contains all the difficulties. We define 
\[
\sigma_{jkc}\vcentcolon= \bigg(\frac{\jbs{j-k}\jbs{k}}{\jbs{j}}\bigg)^p \bigg(\frac{\tw_{j-k} \tw_k}{\tw_j} \bigg)^{\beta},
\]
then
\[
\frac{1}{\sigma_{jkc}^2} \lesssim \bigg( \frac{1}{\jbs{j-k}}+\frac{1}{\jbs{k}} \bigg)^{2p} \bigg( \frac{1}{\tw_{j-k}}+ \frac{1}{\tw_k} \bigg)^{2\beta}.
\]
By denoting $K_{jc}\vcentcolon= \sum_k \sigma_{jkc}^{-2}$, we have 
\[
K_{jc} \lesssim \sum_{k\in \Z} \bigg( \frac{1}{\jbs{j-k}}+\frac{1}{\jbs{k}} \bigg)^{2p} \bigg( \frac{1}{\tw_{j-k}}+ \frac{1}{\tw_k} \bigg)^{2\beta} \lesssim \sum_{k\in \Z} \bigg( \frac{1}{\jbs{j-k}^{2p}}+\frac{1}{\jbs{k}^{2p}} \bigg) =:K,
\]
with $K>0$ independent on $c$, where we used that $p>1/2$ to obtain the summability of the last series. Finally, we have
\begin{align*}
\| z \star w \|_{a,p, \beta}^2& \leq K_{jc}^2 \sum_{k\in \Z} \sigma_{jck}^2 |z_{j-k} w_k|^2& \\
& \lesssim \sum_{j,k \in \Z} \jbs{j-k}^{2p} \tw_{j-k}^{2\beta}|z_{j-k}|^2 \jbs{k}^{2p} \tw_k^{2\beta} |w_k|^2= \|z\|_{a,p, \beta}^2 \|w\|_{a,p, \beta}^2.&
\end{align*}
\end{proof}
{Given $n \in \mathbb{N}$ and a set of vectors $\{z^{(k)}\}_{k=1}^n$
  with $z^{(k)}  \in \ell^{a,p,\beta}$, we denote the convolution among them by 
\[
\bigg(\mathop{\bigstar}_{t=1}^n z^{(t)}\bigg)\vcentcolon=  z^{(1)} \star \dots \star z^{(n)},
\]
and by the previous lemma we also have 
  \begin{equation}
    \label{conv.n}
\left\|\bigg(\mathop{\bigstar}_{t=1}^n z^{(t)}\bigg)\right\|_{a,p,\beta}\sleq
\prod_{k=1}^{n}\left\| z^{(k)}\right\|_{a,p,\beta}
    \end{equation}

For $t=1, \dots, n$, we consider  positive sequence
  \begin{equation}\label{bielle}
b^{(t)}\vcentcolon= \left\{b^{(t)}_j\right\}_{j\in\Z}, 
\end{equation}
and we define the operator $b^{(t)}$ acting on compactly supported sequences by
\[
b^{(t)}z\vcentcolon=   \left\{b^{(t)}_jz_j\right\}_{j\in\Z}.
\]
Of course such operators extend to operators on $\ell^{a,p,\beta}$
defined on their maximal domain. One can also define in the obvious way their action on $\cP^{a,p,\beta}$.
Then the following lemma holds.
\begin{lemma}  \label{sti.abs.2}
    Let $F$ be a Hamiltonian of the form\footnote{To be completely
    rigorous one can define $F$ using the formula \eqref{sti.abs.1} on
    the space of compactly supported sequences, then, by the result of
    Lemma \ref{sti.abs.2} $F$ extends to an analytic function on the space of the
    sequences s.t. the r.h.s. of \eqref{sti.abs.6} is finite.}
    \begin{equation}
  \label{sti.abs.1}
F(z)=
\sum_{\vj,\vsigma}F_{\vj, \vsigma}\,b^{(1)}_{j_1}z^{\sigma_1}_{j_1}\,...\,b^{(n)}_{j_n}z^{\sigma_n}_{j_n} 
\end{equation}
where $F_{\vj, \vsigma}$ is s.t 
\begin{align}
\left|F\right|_\infty \vcentcolon= \sup_{\vj,\vsigma}|F_{\vj, \vsigma}|<\infty
\\
F_{\vj, \vsigma}=0\quad\text{if}\quad \vsigma\cdot \vj\not=0.
\end{align}
Then
\begin{align}
  \label{sti.abs.6}
\left\|X_F(z)\right\|_{a,p,\beta}\leq \left|F\right|_\infty
\sum_{t=1}^n  \bigg\| b^{(t)} \bigg( \mathop{\bigstar}_{\substack{k=1\\ k \neq t}}^n w^{(k)} \bigg) \bigg\|_{a,p,\beta},
\end{align}
where
\begin{equation}
  \label{wk}
w^{(k)}_j \vcentcolon= b_j^{(k)}(|z_j^+|+|z_j^-|).
  \end{equation}
\end{lemma}}

\begin{proof}
 For $s \in \{ \pm \}$ we have
  $$
\frac{\partial F}{\partial
  z^{s}_q}=\sum_{k=1}^nb_q^{(k)}\sum_{\substack{(\vj, \vsigma)\in\Z^n \times \{\pm\}^n \\ j_k=q, \, \, \sigma_k=s}}F_{\vj, \vsigma}\prod_{\substack{m=1\\m\not=k}}^{n}b^{(m)}_{j_m}z_{j_m}^{\sigma_m},
$$
and of course the argument is different from zero only when $\vsigma\cdot\vj=0$.
So we have
\begin{align*}
\left|\frac{\partial F}{\partial
  z^{s}_q}\right|&\leq
\sum_{k=1}^nb_q^{(k)}\sum_{\substack{(\vj, \vsigma)\in\Z^n \times \{\pm\}^n \\ j_k=q, \, \, \sigma_k=s}}\left|F_{\vj, \vsigma}\right|
\prod_{\substack{m=1\\m\not=k}}^{n}b^{(m)}_{j_m}\left|z_{j_m}^{\sigma_m}\right|&
\\
&\leq \left|F\right|_\infty \sum_{k=1}^nb_q^{(k)}\sum_{\substack{\vj\in\Z^n\\j_1+...+j_{k-1}+j_{k+1}+...+j_n=-sq}}
\prod_{\substack{m=1\\m\not=k}}^{n}w^{(m)}_{j_m}&
\\
&=\left|F\right|_\infty \sum_{k=1}^n  b^{(k)}_q \bigg(
\mathop{\bigstar}_{\substack{m=1\\ m \neq k}}^n w^{(m)} \bigg)_{-sq},&
\end{align*}
which is the thesis.
\end{proof}

\begin{lemma}\label{modop}[Lemma A.1 \cite{poschel1996kam}]
Consider a bounded linear operator $A:\ell^{2}( \mathbb{Z} ; \mathbb{C}) \to \ell^{2}( \mathbb{Z} ; \mathbb{C}) $ such that 
\[
(Av)_i=\sum_{j \in \mathbb Z} A_{ij}v_j, \qquad \, \, \, \forall \, \, i \in \mathbb{Z}, \, \forall v \in \ell^2,
\]
then also the operators $B_1,B_2$ defined as
\begin{equation}\label{BB}
(B_1v)_i\vcentcolon=  \sum_{j \neq i} \frac{|A_{ij}|}{|i-j|}v_j, \quad (B_2v)_i\vcentcolon=  \sum_{j \neq -i} \frac{|A_{ij}|}{|i+j|}v_j, \quad \forall \, \, i \in \mathbb{Z}, \, \forall v \in \ell^2
\end{equation}
are bounded linear operators on $\ell^{2}$ to itself, and the following estimate for the operator norms holds $\| B_s\|_{\ell^{2};\ell^{2}} \lesssim \|A\|_{\ell^{2};\ell^{2}}$, for $s=1,2$.
\end{lemma}
{\begin{remark}\label{Stima:pesata}
Fix $a \geq 0$, and $p,q>1/2$.  The previous lemma extends to bounded operators $A: \ell^{a,q}\to\ell^{a,p}$.
Indeed, using the diagonal isometries $(S_r z)_j=\langle j\rangle^{r}e^{a |j|}z_j$, for $r=p,q$,
we conjugate $A$ into $\tilde A = S_p A S_q^{-1}:\ell^2\to\ell^2$. Then, for $B$ defined as \eqref{BB} we have
$\|B\|_{\ell^{a,q};\ell^{a,p}}\lesssim \|A\|_{\ell^{a,q};\ell^{a,p}}$.
\end{remark}}
Finally, we state the following Lemma 
\begin{lemma}[Cut-off operators]\label{lemma:cutoff-multiplication}
Let $a \geq 0$, $ p_0 \geq 0$, $p>p_0$. For $K>0$, we define $\mathcal{C}_K \vcentcolon =\{ j \in \mathbb{Z} \mid |j| \leq K \}$, and recalling Definition \eqref{cut-off}, for any $v \in \ell^{a,p}$ we have
\[
\|\Pi_{\mathcal{C}_K}^\perp v\|_{a,p-p_0} \leq K^{-\alpha } \|v\|_{a,p}.
\]
\end{lemma}

\begin{proof}
By definition we have
\[
\begin{aligned}
\|\Pi_{\mathcal{C}_K}^\perp v\|_{a,p-\alpha}^2& = \sum_{|j|>K} \langle j \rangle^{2(p-\alpha)} |v_j|^2 e^{2a|j|} = \sum_{|j|>K} \langle j \rangle^{-2\alpha} \langle j \rangle^{2p} |v_j|^2 e^{2a|j|}& \\
& \leq K^{-2\alpha} \sum_{|j|>K} \langle j \rangle^{2p} |v_j|^2 e^{2a|j|} \leq K^{-2\alpha} \|v\|_{a,p}^2.&
\end{aligned}
\]
%Taking square roots gives the desired bound.
\end{proof}
\section{Abstract estimates on flows}\label{flowmap1}

  \begin{lemma}
  \label{flussi1}
Let $E$ be a complex Banach space and $\cU\subset E$ be an open subset. We consider two vector fields $X_i\in C^{\omega}(\cU;E)$, i=1, 2. We
assume
\begin{equation}
  \label{flowmap.10}
\sup_{x\in\cU}\left\|X_i(x)\right\|_E\leq \delta \qquad i=1, 2,
\end{equation}
for some $\delta>0$. Assume 
  \begin{equation}
    \label{K1}
K_1\vcentcolon= \sup_{{x,y\in\cU}\atop
  x\not=y}\frac{\left\|X_2(x)-X_2(y)\right\|_E}{\left\|x-y\right\|_E} < \infty,
  \end{equation}
  and let $\cV \subseteq \cU$ be the open set with the property that
\begin{equation}
  \label{flowmap.11}
\bigcup_{y\in\cV}B(\delta,y)\subset\cU,
\end{equation}
where $B(\delta,y) \subseteq E$ is the ball of radius $\delta$ and center
$y$.  Then 
  \begin{equation}
    \label{flowmap.12}
\sup_{x\in\cV}\left\|\phi_1^t(x)-\phi_2^t(x)\right\|_E\leq
\frac{e^{K_1t}-1}{K_1} \sup_{x\in\cU}\left\|X_1(x)-X_2(x)\right\|_E
,\quad \forall x\in \mathcal{V}, \quad \forall |t|\leq 1.
    \end{equation}
  \end{lemma}
  \proof
Denote by $\phi_1^t$, $\phi_2^t$, $t\in [-1, 1]$, the local flows generated by $X_1$ and
$X_2$ respectively. Then, as in Subsection \ref{phase.space}, we have
$\phi_i^t\in C^\omega(\cV,\cU)$ for all times $|t|\leq 1$.\\ 
   Denote (only for this proof) $\Delta(x,t)\vcentcolon= 
  \phi_1^t(x)-\phi_2^t(x)$ and observe that $\Delta(x, 0)=0$, then we have
  \begin{equation}
    \label{flowmap14}
\frac{d}{dt}\Delta(x,t)=X_1(\phi_1^t(x))-X_2(\phi_1^t(x))+X_2(\phi_1^t(x))
-X_2(\phi_2^t(x))  
    \end{equation}
from which (since $\phi_i^t(x)\in\cU$, $\forall x\in\cV$ $t \in [-1,1]$)
\begin{equation}
  \label{flowmap.15}
\left\|\frac{d}{dt}\Delta(x,t) \right\|_E\leq
\sup_{y\in\cU}\left\|X_1(y)-X_2(y)\right\|_E
+K_1\left\|\Delta(x,t)\right\|_E.
  \end{equation}
By Gronwall inequality we obtain \eqref{flowmap.12}.

 \qed

\begin{remark}
We point out that the previous lemma is also useful to study the flows of a vector field, Lipschitz with respect to a parameter, and its lipschitz semi-norm.  If $X: \cU \times \Xi \to E$, with $\Xi\subset \R^N$, is such that for any $\xi \in \Xi$ 
the vector filed $X(\xi)$ satisfies the hypothesis of Lemma \ref{flussi1} (with $\delta$ uniform in $\xi$), then for any $\xi_1,\xi_2 \in \Xi$ with $\xi_1 \neq \xi_2$ we can apply \eqref{flowmap.12} obtaining for any $t \in [-1,1]$ the following estimate
\begin{equation}\label{Lip:semnorm}
\sup_{x \in \cV} \frac{\| \phi_X^t(\xi_1;x)-\phi_X^t(\xi_2;x) \|_E}{|\xi_1-\xi_2|} \leq \frac{e^{K_1t}-1}{K_1} \sup_{x \in \cU} \frac{\|X(\xi_1;x)-X(\xi_2;x)\|_E}{|\xi_1-\xi_2|}.
\end{equation}
Then taking the sup over all $\xi_1,\xi_2 \in \Xi$ with $\xi_1 \neq \xi_2$ we have that the lipschitz semi-norm of the flow w.r.t the parameter is bounded by the lipschitz semi-norm of the vector field.
\end{remark}
  \begin{lemma}
    \label{flowmap.perdo}
 Assume the same hypotheses of Lemma \ref{flussi1} and let $\tilde E$ be a Banach space such that $E$ is continuously embedded in $\tilde E$. Assume also that 
  \begin{equation}
    \label{K1}
K_2\vcentcolon= \sup_{{x,y\in\cU}\atop
  x\not=y}\frac{\left\|X_2(x)-X_2(y)\right\|_{\tilde
    E}}{\left\|x-y\right\|_{E}} <\infty,
  \end{equation}
 and let $\cV \subseteq \cU \subseteq E$ be the open set with the property that
\begin{equation}
  \label{flowmap.11}
\bigcup_{y\in\cV}B(\delta,y)\subset\cU.
\end{equation}
Then, the following estimate holds true
  \begin{equation}
    \label{flowmap.12.1}
\sup_{x\in\cV}\left\|\phi_1^t(x)-\phi_2^t(x)\right\|_{\tilde E}\leq
\frac{e^{K_2t}-1}{K_2}
\sup_{x\in\cU}\left\|X_1(x)-X_2(x)\right\|_{\tilde E}
,\quad \forall |t|\leq 1.
    \end{equation}
  \end{lemma}
  
  \begin{proof}
the estimate follows by exploiting \eqref{flowmap14}, but writing \eqref{flowmap.15} for the norm in
  $\tilde E$.
  \end{proof}

\section{Frequencies and small divisors}\label{Smalldivisors}
In this section, we are going to prove some technical lemmas about the frequencies and the small divisors. { We start recalling the following result, proved in \cite{berti2013kam}.
\begin{lemma}[Lemma  7.2 \cite{berti2013kam}]
  \label{lemma4} For any $m>0$ the following holds. Let $\tilde \lambda_i(m)\vcentcolon=\sqrt{i^2+m}$. If  $\{ \bs, \textbf{j}\}  \in \mathcal{L}_\cS \setminus \cIR$ (c.f. Definition \ref{resonant})
the following estimate is satisfied
\begin{equation*}
 \bigg|\sum_{i=1}^4 \sigma_i\tilde{\lambda}_{j_i}\bigg| \geq
  \frac{K m}{(m+J^2)^{3/2}},
\end{equation*}
with $K>0$ an absolute constant and $J \vcentcolon = \max_{j \in \cS} |j|$. Moreover,  if $\sum_{i=1}^4\sigma_i \not=0$ then
\[
 \bigg|\sum_{i=1}^4 \sigma_i\tilde{\lambda}_{j_i}\bigg| \geq \frac{K m}{(J^2+m)^{1/2}}
\]
\end{lemma}}
{Then, remarking that $\lambda_j=c \tilde\lambda_j(c^2)$ (c.f. \eqref{LKG}), we obtain the following result about the small divisors.}
\begin{lemma}\label{divisori}
For any $c>1$ the following holds. If  $(\vsigma, \vj) \in \mathcal{L}_\cS \setminus \cIR$
(cf. \eqref{LN} and definition \ref{resonant}), then the following estimate is satisfied
\begin{equation}
  \label{divisori.1}
  \bigg|\sum_{i=1}^4 \sigma_i\lambda_{j_i}\bigg|\geq
  \frac{K}{(1+J^2)^{3/2}},
\end{equation}
with $K>0$ an absolute constant, $J
$ and $\lambda_i$ defined respectively in Lemma \ref{lemma4} and \eqref{LKG}. Moreover, if $\sum_{i=1}^4\sigma_i \not=0$ then
\begin{equation}
  \label{divisori.2}
   \bigg|\sum_{i=1}^4 \sigma_i\lambda_{j_i}\bigg|\geq
 c^2 \frac{K}{(1+J^2)^{1/2}}.
  \end{equation}
\end{lemma}

We analyze now the frequencies $\omega_0$, $\Omega_0$ and their
  dependence on $\xi$. To this end it is crucial to analyze the
  matrixes $A$ and $B$ introduced in \eqref{def:AKG},\eqref{def:BKG}, so first of all we
  rewite them in a more suitable form. In particular we emphasize that
  $A$ is a rank 1 perturbation of a diagonal matrix and $B$ is the
  projector on one dimensional space.  We have

\begin{align*}
 & v=(v_j)_{j\in\cJ}\vcentcolon= \bigg(\frac{\sqrt{2}}{\tw_j}\bigg)_{j\in\cJ} \in \mathbb{R}^{\cJ}, \qquad
 w=\{w_j\}_{j\in\cJ^c}\vcentcolon= \bigg\{\frac{\sqrt{2}}{\tw_{j}}\bigg\}_{j \in\cJ^c} \in \ell^{\infty, 1},&
  \\
  &D\vcentcolon= \textnormal{diag}_{j\in\cJ}\left(\frac{1}{\tw_j^2}\right), \quad C\vcentcolon= v\langle v, \cdot \rangle, \qquad A=\frac{3}{8\pi}(C-D), \quad B=\frac{3}{8\pi}w\langle v, \cdot \rangle,&
\end{align*}
where $\tw_j$ is given in \eqref{ellappiu}, $\cS$ defined in \eqref{insieme}, $\cS^c \vcentcolon= \mathbb{Z} \setminus \cS$ and $\ell^{\infty, 1}$ defined in \eqref{infbeta}. First, we have the following remark
\begin{remark}
For any $c \geq 1$ and $j \in \mathbb{Z}$ it holds
\begin{equation}\label{bound:peso}
\tw_j \geq 1, \quad \tw_j \geq |j|/c, \quad \mathrm{and} \quad \tw_{j_1} > \tw_{j_2}, \, \, \textnormal{whenever} \, \, |j_1|>|j_2|.
\end{equation}
\end{remark}}
In the following lemmas we denote $\tw\vcentcolon=(\tw_j)_{j\in\cJ}$.
\begin{lemma}
  \label{inverto}
  The matrix $A$ is invertible 
  with 
  \begin{equation}
    \label{inverto.A}
A^{-1}=\frac{8\pi}{3}\left(\frac{2\langle \tw,.\rangle}{2N-1}\tw-D^{-1}\right),
  \end{equation}
  and  
  \begin{equation}
    \label{invert.1}
\left|A^{-1}\right|_{\ell^1_\cJ;\ell^1_\cJ} \leq \frac{8\pi}{3}\frac{4N-1}{2N-1}|\tw|^2 ,
  \end{equation}
  {where $\ell^1_\cJ \vcentcolon = \ell^1( \cJ ; \mathbb{R})$.}
\end{lemma}
\proof We solve the equation
\begin{align}
  \label{risolvo}
(C-D)x=y\quad \iff\quad v\langle v,x\rangle-Dx=y.
\end{align}
Multiplying by $D^{-1}$ and then taking the scalar product with $v$
we have 
$$
\langle v,D^{-1}v\rangle\langle v,x\rangle-\langle v,x\rangle=\langle
v,D^{-1}y\rangle, \ \Longrightarrow \ \langle v,x\rangle=\frac{\langle
v,D^{-1}y\rangle }{\langle
v,D^{-1}v\rangle-1}=\frac{\langle
D^{-1}v,y\rangle }{\langle
v,D^{-1}v\rangle-1}.
$$
By inserting in \eqref{risolvo} and solving, we have
\begin{equation}
  \label{solv.01}
x=\frac{\langle
D^{-1}v,y\rangle }{\langle
v,D^{-1}v\rangle-1}D^{-1}v-D^{-1}y,
\end{equation}
which is also known as Bateman formula. 
By computing explicitly we have 
\begin{align*}
& (D^{-1}v)_j=\tw_j^2\frac{\sqrt{2}}{\tw_j}=\sqrt{2}\tw_j, 
  \\
  &\langle v,D^{-1}
v\rangle=\sum_{j \in \cS} \frac{\sqrt{2}}{\tw_j}\,(\sqrt{2}\tw_j)=\sum_{j \in \cS} 2=2N,
\end{align*}
and inserting in \eqref{solv.01}
\begin{equation}
  \label{solv.11}
x=\frac{2\langle\tw,y\rangle}{2N-1}\tw-D^{-1}y,
\end{equation}
which gives \eqref{inverto.A}. From this we have
\begin{align*}
  \label{solv.3}
| x|_1&\leq  \frac{2N|\tw|^2}{2N-1}|y|_1+| \tw| |y|_1\leq  \frac{(4N-1)|\tw|^2}{2N-1}|y|_1
\end{align*}
which gives the thesis. \qed

\begin{lemma}
{Given $\ell\in \ell^{\infty, 1}$,}
  the solution of the equation
  \begin{equation}
    \label{solv.1.e}
Ax+B^t\ell =0,
  \end{equation}
  {where
    $$
B^t\vcentcolon= \frac{3}{8\pi}\langle w,.\rangle v
    $$
is the transpose of $B$,
  }
is given by
\begin{equation}
  \label{solve.1.2}
{x=\langle w,\ell\rangle\frac{\sqrt{2}}{1-2 N}{\tw}}.
  \end{equation}
\end{lemma}
\proof Since $A$ is invertible, the expression \eqref{solve.1.2} is just given by a direct computation of $-A^{-1}B^t\ell$.
 \qed

\smallskip

\begin{lemma}\label{Lemma:N}
  If $|\ell|_1=1,2$ then \eqref{solve.1.2} fulfills 
  \begin{equation}
    \label{l=1}
{|x_j|\leq \frac{4}{2N-1}\tw_j,\quad j\in\cJ} .
  \end{equation}
  \end{lemma}
\begin{proof}
 In the case $|\ell|_1=1$
we have  $\langle w,\ell\rangle=\pm \sqrt{2}/\tw_n$ for some integer $n\in\cJ^c$. So
we have
$$
|x_j|=\frac{\sqrt{2}}{2N-1}\frac{\tw_j}{\tw_n}\leq \frac{\sqrt{2}}{2N-1}\frac{\tw_j}{\tw_0}=\frac{\sqrt{2}}{2N-1}\tw_j.
$$
When $|\ell|_1=2$ we have
$$
{|\langle  w,\ell\rangle|=\frac{\sqrt{2}}{\tw_n}\pm \frac{\sqrt{2}}{\tw_m}},
$$
for some indices $n<m\in\cJ^c$. So, reasoning in the same way the conclusion
follows.
\end{proof}

\begin{lemma}\label{Stima:grande0}
There exists a constant $C$, which depends only on $N$, such that, if $N\geq
3$ and 
\begin{equation}
  \label{jh}
{h\leq\frac{49}{576J^2}},
\end{equation}
then
\begin{equation}\label{AB:1}
|Ak+B{^t}\ell |_1 \geq C |k|_1,\quad \forall(k,\ell) \in \mathcal{Z}_2.
\end{equation}
\end{lemma}
\begin{proof}
Assume $k\not=0$, otherwise the result is trivial. 
Let $\delta\vcentcolon= Ak+B^t\ell$, and denote by $x$ the solution of the
equation \eqref{solv.1.e}. One has
\[
\delta=A(k-x)+Ax+B^t\ell=A(k-x),\ \Longrightarrow k-x=A^{-1}\delta
\]
and thus, from the boundedness of $A^{-1}$ (c.f. \eqref{invert.1}) we have
\[
|k-x|_1\sleq |Ak+B^t\ell   |_1.
\]
On the other hand, let $\bar\i$ be the index s.t. $|k_{\bar
  \i}|=\max|k_j|$, we have
\[
|k - x|_1
\gtrsim
\left|k_{\bar\i} - x_{\bar\i}\right|
\gtrsim
\left|k_{\bar\i}\right|\!\left(1 - \left|\frac{x_{\bar\i}}{k_{\bar\i}}\right|\right)
\gtrsim
\left|k_{\bar\i}\right|\!\left(1 - \left|x_{\bar\i}\right|\right)
\]
but
\[
\left|x_{\bar\i}\right|\leq
\frac{4\tw_J}{2N-1}\leq\frac{4}{5}\sqrt{1+hJ^2}\leq \frac{5}{6},
\]
where we used $N\geq 3$ and the limitation on $h$. So the thesis follows.
\end{proof}

%%%%%%%%%%%%%%%%%%%%%%%%%%%%%
\begin{lemma}\label{Introomega1}
The frequency maps $\omega:\Xi_0 \to \R^N$,
  $\omega^{\tt NLS}:\Xi_0 \to \R^N$ defined in \eqref{Oomega}, \eqref{OomegaNLS}
are two Lipeomorphisms on $\Xi_0$ defined in \eqref{Xi0}. Moreover, the following estimates hold
\begin{equation}\label{AB:3}
|\omega_0-(c^2+\omega_0^{\tt NLS})|_{\infty} \lesssim
h, \qquad
\sup_{j \in \cS^c}\left|\left(\Omega_0\right)_j-\left(c^2+\Omega_0^{\tt NLS})\right)_j\right|_\infty j^{-4} \lesssim h.
\end{equation}
\end{lemma}

\begin{proof}
The fact that the two maps are Lipeomorphism and the uniform bounds on the Lipschitz constants comes from the invertibility of $\matA$ and \eqref{invert.1}. 
The estimate \eqref{AB:3} follows directly by the estimates \eqref{difffreq} and by the following estimates 
\[
\left|\matA(h)\xi-\matA^{\tt NLS}\xi \right|_\infty \lesssim h r^{4/3}, \quad \left| (\matB(h)\xi)_j-(\matB^{\tt NLS}\xi)_j\right|_\infty \lesssim j^4 h r^{4/3}, \, \, \, \, \forall j \in \cS,
\]
that follow from \eqref{def:AKG}, \eqref{def:BKG}, \eqref{matNLS} and the estimate \eqref{stim:intermedia}.
\end{proof}
\begin{lemma}
  \label{C.55}
  Let $i, j\in \mathbb{Z}$, $c \geq 1$. The following formulae and estimates hold 
  \begin{align}
    \label{C.56}
    &\lambda_j-\lambda_i=\frac{j^2-i^2}{\tw_i+\tw_j}=c(|j|-|i|)\frac{|i|+|j|}{\sqrt{c^2+i^2}+\sqrt{c^2+j^2}},& &\,&
\\
\label{B.2}
&\lambda_j-\lambda_i\leq c(|j|-|i|)& \, \, &\mathrm{for}\,\,|j|>|i|.&
\\
\label{C.57}
    &(\matB\xi)_i-(\matB\xi)_j=\frac{3\sqrt{2}}{8\pi}(|j|-|i|)\frac{|i|+|j|}{\sqrt{c^2+i^2}+\sqrt{c^2+j^2}}\frac{\langle
      v,\xi\rangle}{c^2}\frac{1}{\tw_i\tw_j},&\, \, &i, j \in \cS^c, \, \, \forall \xi \in \Xi_0.&
  \end{align}
Moreover, for any $h$ that satisfies \eqref{jh}  and  for any $i,j \in \mathbb{Z}$ such that $c^3<|i|<|j|$ we have
\begin{equation}\label{AB:2}
\bigg|\frac{ (\Omega_0)_j- (\Omega_0)_i}{ c(|j|-|i|)} -1\bigg|_\infty\sleq
\frac{1}{\tw_i^2},
\end{equation}
where the norm is introduced in \eqref{norma:Omega}.
\end{lemma}

\begin{proof}
The equalities \eqref{C.56} and \eqref{C.57} are obtained by direct computations, recalling the definition of $\lambda_j$ in \eqref{LKG}, the definition of $\tw_j$ in \eqref{ellappiu}, and the definition of $\matB$.
The inequality \eqref{B.2} follows from the second equality of \eqref{C.56} and
  remark that the last fraction is always smaller than 1. 

\smallskip

We prove now \eqref{AB:2}. For $c^3\leq |i|<|j|$ we have
\[
\sqrt{c^2+i^2}=|i| \sqrt{1+\frac{c^2}{i^2}}=|i|\bigg(1+\mathcal{O}\bigg(\frac{c^2}{i^2}\bigg)\bigg),
\]
then
$$
\frac{|i|+|j|}{\sqrt{c^2+i^2}+\sqrt{c^2+j^2}}=1+\cO\left(\frac{c^2}{i^2}\right)=1+\cO\left(\frac{1}{\tw_i^2}\right),
  $$
     and therefore by \eqref{C.56}, \eqref{C.57} and \eqref{Oomega} we have
     $$
\Omega_{0j}-\Omega_{0i}=c(|j|-|i|)\left(1+\cO\left(\frac{1}{\tw_i^2}\right)+\cO\left(\frac{1}{c^2\tw_i^2}\right)\right),
$$
proving the thesis.
\end{proof}

\begin{lemma}
The following estimates hold true 
\begin{equation}\label{Omega:lambda}
|\Omega_0-\Lambda|_{\infty,1} \lesssim r^{4/3}, \qquad |\omega_0-{\lambda}|_{\infty} \sleq  r^{4/3}.
\end{equation}
with the first norm defined in \eqref{norma:Omega}, and $| \cdot |_\infty$ is the sup-norm on $\Xi_0(\tR)$ of a $N-$vector.
\end{lemma}
\begin{proof}
By \eqref{Oomega}, for any $\xi \in \Xi_0$, and $j \in \cS^c$, by the definition of $B$ we have  
\[
|(\Omega_0(\xi))_j-\lambda_j|=\left|(B\xi)_j\right|\lesssim \left|\frac{\sqrt{2}\langle
 v,\xi \rangle}{\tw_j}\right| \sleq \frac{|\xi|}{\tw_j}\sleq \frac{r^{4/3}}{\tw_j}.
\]
The second inequality is obtained in the same way.
\end{proof}

\section{Inverse function theorem for some Lipschitz maps}
We recall Lemma A.2 of \cite{poschel1996kam} about the existence of a Lipschtiz extension. 
\begin{lemma}\label{estensione}
Let $F \subset \mathbb{R}^n$ be closed and $u: F \to \mathbb{R}$ a bounded Lipschitz continuous function. Then there exists an extension $U: \mathbb{R}^n \to \mathbb{R}$ of $u$, which preserves minimum, maximum, and Lipschitz semi-norm.
\end{lemma}
\begin{lemma}\label{perturbazione}
Let $X \subseteq \mathbb{R}^N$ be a subset and a lipeomorphism $f:X \to f(X)$ with
Lipschitz constants 
\begin{equation*}
\sup_{\substack{x_1,x_2 \in X \\ x_1 \neq x_2}} \frac{|f(x_1)-f(x_2)|}{|x_1-x_2|} \leq L_f, \qquad \sup_{\substack{y_1,y_2 \in  f(X) \\ y_1 \neq y_2}} \frac{|f^{-1}(y_1)-f^{-1}(y_2)|}{|y_1-y_2|} \leq L_{f^{-1}}.
\end{equation*}
Let $g:X \to g(X)$ be a Lipschitz map with Lipschitz constant
\begin{equation}\label{stim:2:aux}
\sup_{\substack{x_1,x_2 \in X \\ x_1 \neq x_2}}
\frac{|g(x_1)-g(x_2)|}{|x_1-x_2|}=:L_g  < \frac{1}{L_{f^{-1}}},
\end{equation}
then the map $h=f+g$ is a lipeomorphism from $X$ to
$h(X)$. Furthermore $h^{-1}$ has Lipschitz constant
\begin{equation}\label{stim:3:aux}
L_{h^{-1}}\vcentcolon=  \sup_{\substack{z_1,z_2 \in  h(X) \\ z_1 \neq z_2}} \frac{|h^{-1}(z_1)-h^{-1}(z_2)|}{|z_1-z_2|} \leq L_{f^{-1}}\left(1-L_{f^{-1}}L_g\right)^{{-1}}.
\end{equation}
\end{lemma}
\begin{proof}
For $x_1, x_2 \in X$ we call $y_i=f(x_i)$ and we have
\[
|x_1-x_2|=|f^{-1}(y_1)-f^{-1}(y_2)| \leq L_{f^{-1}} |y_1-y_2|.
\]
Thus 
\begin{equation}\label{aux:2:proof}
|h(x_1)-h(x_2)| \geq |f(x_1)-f(x_2)|-|g(x_1)-g(x_2)| \geq \left(\frac{1}{L_{f^{-1}}}-L_g\right) |x_1-x_2|.
\end{equation}
In particular, it follows that $h$ is injective and therefore
invertible onto its image. Furthermore by \eqref{aux:2:proof} the Lipschitz constant
of $h^{-1}$ is estimated by \eqref{stim:3:aux}.
\end{proof}
  Given a set $X\subset \R^N$ we denote 
\begin{equation}\label{bordino}
X_{-\varrho}\vcentcolon= \left\{x\in
  X\ |\ B_\varrho(x)\subset X \right\}.
  \end{equation}
\begin{lemma}  
  \label{stime_interno}
  Let $f$ and $g$ be as in Lemma \ref{perturbazione} and assume \eqref{stim:2:aux}. Denote
  $$
\epsilon\vcentcolon= \sup_{x\in X}\left|g(x)\right|,
$$
and define $\varrho \vcentcolon= \epsilon L_{f^{-1}}$. Then
\begin{equation*}
f(X_{-\varrho })\subset h(X).
\end{equation*}
Furthermore, for all $x\in X_{-\varrho}$ there exists $\tilde x\in \overline{B_{\varrho}(x)}$
such that $f(x)=h(\tilde x)$.
\end{lemma}
\proof We look for $\tilde x$ as the solution of the equation
$$
\tilde x=f^{-1}(f(x)-g(\tilde x)),
$$
namely as the fixed point of the map $\tilde x\mapsto
f^{-1}(f(x)-g(\tilde x))$. Now, it is easy to see that such a map is a
contraction on $\overline{B_{\varrho}(x)}$, and therefore the conclusion
holds. \qed

\newpage
%\bibliographystyle{alpha}
%\bibliography{BibKAM}

\end{document}